\documentclass[11pt]{article}
\usepackage{amsfonts,amsthm,amssymb}
\usepackage{cite}
\usepackage[pdftex]{graphicx}
\usepackage{color}
\usepackage[fleqn]{amsmath}
\usepackage{hyperref}
\newcommand{\be}{\begin{eqnarray}}
\newcommand{\ee}{\end{eqnarray}}
\newcommand{\bez}{\begin{eqnarray*}}
\newcommand{\eez}{\end{eqnarray*}}
\renewcommand{\d}{\mathrm{d}}
\newcommand{\bd}{\bar{\mathrm{d}}}

\newcommand{\imag}{\mathrm{i}}

\theoremstyle{plain}
\newtheorem{theorem}{Theorem}[section]
\newtheorem{lemma}[theorem]{Lemma}
\newtheorem{proposition}[theorem]{Proposition}
\newtheorem{corollary}[theorem]{Corollary}

\theoremstyle{definition}

\newtheorem{remark}[theorem]{Remark}
\newtheorem{example}[theorem]{Example}

\numberwithin{equation}{section}
\numberwithin{theorem}{section}

\allowdisplaybreaks

\begin{document}

\title{\textbf{A vectorial Darboux transformation for \\ integrable matrix versions of the \\ Fokas-Lenells equation}} 
\author{
 {\sc Folkert M\"uller-Hoissen}$\,^a$ and {\sc Rusuo Ye}$\,^b$ \\ 
 \small 
 $^a$ Institut für Theoretische Physik, Friedrich-Hund-Platz 1,
 37077 G\"ottingen, Germany \\
 \small E-mail: folkert.mueller-hoissen@theorie.physik.uni-goettingen.de \\
 \small $^b$ College of Mathematics and Physics, Wenzhou University, Wenzhou	325035, PR China \\
 \small E-mail: rusuoye@163.com
}

\date{} 

\maketitle

\begin{abstract}
Using bidifferential calculus, we derive a vectorial binary Darboux transformation for an integrable matrix version of the first negative 
flow of the Kaup-Newell hierarchy. A reduction from the latter system 
to an integrable matrix version of the Fokas-Lenells equation is then shown to inherit a corresponding vectorial Darboux transformation. Matrix soliton solutions are derived from the trivial seed solution. 
Furthermore, the Darboux transformation is exploited to determine in a systematic way exact solutions of the two-component vector Fokas-Lenells equation on a plane wave background. This comprises breathers, dark solitons, rogue waves and ``beating solitons". 
\end{abstract}    


\section{Introduction}
By a transformation of variables, the original \emph{Fokas-Lenells (FL) equation} \cite{Fokas95,Lene+Foka09} attains the simpler form \cite{Lene09}
\be
     u_{xt} - u - 2 \imag \, |u|^2 u_x = 0  \, , \label{FLeq}
\ee
where $u$ is a complex function of independent real variables $x$ and $t$, and a subscript $x$ or $t$ indicates a partial derivative with respect to $x$,  respectively $t$. 
In particular, it models propagation of ultrashort light pulses in birefringent optical fibers, taking certain nonlinear effects into account. 
For some integrability features of (\ref{FLeq}), see \cite{MH+Ye25} and the references cited there. The FL equation (\ref{FLeq}) possesses a generalization to integrable matrix versions 
\be
  u_{xt} - u - \imag \, (u \, \sigma_1 u ^\dagger \sigma_2 u_x + u_x \sigma_1 u^\dagger \sigma_2 u ) = 0  \, . \label{mFLeq}
\ee
Here $u$ is now a matrix of any size, where the components are functions of $x$ and $t$, $u^\dagger$ denotes the conjugate transpose of $u$, and $\sigma_1, \sigma_2$ are any constant involutary Hermitian matrices of suitable size. 
In the case, where $u$ is a vector, (\ref{mFLeq}) already appeared in \cite{Guo+Ling12} (see (82) therein), also see \cite{ZYCW17,LFZ18,Yang+Zhang18,KXM18,CYSGB18,ABZMB18,YZCBG19,Xu+Chen19,Yue+Huang22,Gerd+Ivan21,ZHL21,TZLBMYAB22,Ling+Su24} for the two-component vector case, \cite{Wang+Chen19} for the three-component case, and \cite{Mats18,Mats19,SFY25} for the multi-component case using Hirota's bilinearization approach. Some results on the $2 \times 2$ matrix case have been obtained in \cite{ZSL22}. 

In this work, we present a \emph{vectorial} Darboux transformation for (\ref{mFLeq}), which allows to generate exact solutions from any given ``seed" solution $u$.\footnote{The alternative inverse scattering method allows to solve, in principle, the Cauchy (initial value) problem for an evolution equation. It is a deeper and more comprehensive method, though in practise difficult to apply beyond the soliton case and simple boundary conditions.  
Darboux transformations ``only" generate solitons on a given ``background"  solution, but they are comparatively easy to handle. B\"acklund transformations and Hirota's ``direct method" (via bilinearization) have similar restrictions as solution-generating methods.} 
It involves a linear system that depends on $u$ in such a way that its compatibility (or integrability) condition is satisfied if $u$ solves (\ref{mFLeq}), and the dependent variables of the linear system are $n$-component vectors, 
where $n$ may be thought of as corresponding to the number of solitons in the generated class of solutions of (\ref{mFLeq}) (though already some $n=1$ solutions can be regarded as being composed of more elementary solutions). For $n=1$, the linear system reduces to a Lax pair, depending on a ``spectral parameter". In our linear system, this parameter is promoted to a ``spectral matrix". 

The original Darboux transformation \cite{Matv+Sall91} has to be iterated in order to generate multiple soliton solutions of an integrable equation. Moreover, in each step a different spectral parameter value has to be used. Accordingly, important classes of solutions are only obtained indirectly by taking special limits, where   
a priori different spectral parameters tend to the same value. This concerns 
``multiple pole solutions" (in the language of the inverse scattering method), which include ``positons" of the Korteweg-deVries and 
related equations, for which V.B. Matveev set up a \emph{generalized Darboux transformation} \cite{Matv92,Matv02}. Later a corresponding generalization of the classical Darboux transformation has also been formulated for the Nonlinear Schr\"odinger (NLS) and related integrable equations in order to generate higher order rogue waves and (multi-) dark soliton solutions (see, e.g., \cite{GLL12,LZG15}, and  \cite{LFZ18,Yang+Zhang18} for applications to the case of the two-component vector FL equation).   

A vectorial Darboux transformation does not require iteration in order to generate multi-soliton solutions, but does the job in a single step. Furthermore, it offers a direct way to the abovementioned important classes of solutions.\footnote{The ``Cauchy matrix approach", applied in \cite{LLZ25} to the scalar FL equation, can be regarded as a vectorial Darboux transformation restricted to the case of trivial seed.}		
Moreover, it allows us, at least to some extent and in a tractable way, to exhaust the various possibilities of exact solutions within this framework. A beautiful aspect of 
a vectorial Darboux transformation is a simple superposition rule on the level of data (consisting of a ``spectral matrix" and the corresponding solution of the linear system)\footnote{In case of a vectorial \emph{binary} Darboux transformation, we are dealing with \emph{two} spectral matrices and \emph{two} linear systems.} determining a solution of the integrable nonlinear equation for a given seed solution, see Remark~\ref{rem:FL_superposition} below.

The usefulness of a \emph{vectorial} version of the classical Darboux transformation has also been demonstrated by the ``non-recursive Darboux transformation" in \cite{Chen+Miha15}, which has been applied in \cite{YZCBG19} to construct higher order rogue wave solutions of the two-component vector FL equation. It should be noted, however, that it lacks what we consider an important step, namely to promote the spectral parameter more consistently to a matrix. It thus misses the superposition rule for solutions with the same seed. In any case, the  ``non-recursive Darboux transformation" is a step towards a (complete) vectorial Darboux transformation, as proposed and exploited earlier by Ma\~{n}as \cite{Manas96} for the NLS equation, also see \cite{Sakh94} for a different approach.
   
A key feature of a vectorial Darboux transformation is the appearance of a Sylvester equation for the so-called Darboux potential $\Omega$. The input 
for this algebraic matrix equation is the spectral matrix and a matrix built from the solution of the linear system. When the spectral matrix is diagonal and the Sylvester equation has a 
unique $n \times n$ matrix solution, 
it is expected that the generated solution of the respective nonlinear integrable equation can also be reached by an $n$-fold iteration of a classical Darboux transformation (using different spectral parameters). Multiple pole solutions and higher order rogue waves are obtained by choosing the spectral matrix as a Jordan block, and the nature of the generated class of solutions depends on (properties of) its eigenvalue. 

Additional classes of solitons of NLS-type equations are obtained via the complementary case, where the Sylvester equation does not (completely) determine the Darboux potential and imposes constraints on the solution of the respective linear system.\footnote{That the solution of the linear system can indeed be reduced in such a way that these constraints are satisfied, appears to be a bit of a miracle.}  
Vectorial Darboux transformations, like those obtained via bidifferential calculus \cite{DMH00a} (also see, in particular, \cite{DMH13SIGMA,DMH20dc}), come along with additional differential equations that completely determine the Darboux potential in such a ``degenerate case".  Corresponding solution classes will be worked out explicitly for the two-component vector FL equation in Section~\ref{sec:2vFL}. 

In the present work, we exploit the vectorial Darboux transformation (only) with seed solutions, for which the linear system can be turned into a system of linear PDEs with constant coefficients. Then, generically, the general solution is a linear combination of exponentials of a linear expression in $x$ and $t$. But when the corresponding characteristic equation has a multiple root, the linear system ``degenerates" and the general solution also involves a polynomial (or polynomials) of the independent variables. If the seed solution is an exponential describing a plane wave,  and if we have a root of maximal order, then the resulting solution of the vector FL equation (or, more generally, some NLS-type equation) is quasi-rational, i.e., a rational expression multiplied by the plane wave exponential. If the norm of such a solution is localized on the uniform background, it represents a rogue wave (or freak wave), a case that attracted enormous interest, since it describes a wave that suddenly  appears (``from nowhere") and then disappears again. 

There is another, independent mechanism, how rational expressions can enter the solution stage. This is at work in those cases, where the Sylvester equation does not possess a unique solution, which is true if and only if the spectral matrix does \emph{not} satisfy the so-called \emph{spectrum condition}. 

Though essentially self-contained, this work is based on \cite{MH+Ye25} (which, in turn, partly draws from \cite{LLZ25}), where a vectorial Darboux transformation for the \emph{scalar} FL equation has been derived, using bidifferential calculus.
In Section~\ref{sec:bidiff->KN-1} we show that, by choosing a suitable (realization of the) bidifferential calculus, a simple ``Miura transformation" equation (see (\ref{Miura})) is essentially turned into a matrix generalization of the first ``negative flow" of the Kaup-Newell (KN) hierarchy,
\be
u_{xt} - u - \imag \, (u \, v \,  u_x + u_x \, v \, u ) = 0  \, , \qquad
v_{xt} - v + \imag \, (v \, u \,  v_x + v_x \, u \, v ) = 0 \, ,
\label{KN-1} 
\ee
of which (\ref{mFLeq}) is a reduction, via 
\be 
v = \sigma_1 u^\dagger \sigma_2 \, .   \label{red}
\ee 
The crucial point is that the Miura transformation equation possesses a vectorial binary Darboux transformation (bDT). A bDT allows to generate from a given (``seed") solution a new family of solutions, using solutions of a linear system and also of another (``adjoint") linear system. The bDT for (\ref{KN-1}) is derived in Section~\ref{sec:KN-1_bDT}. 

A difficult task is to find an extension of the reduction (\ref{red}) such that the bDT generates solutions of the matrix FL equation (\ref{mFLeq}), i.e., that it preserves the reduction condition (\ref{red}). This problem is solved in Section~\ref{sec:mFL_DT}, essentially following the steps in the scalar case \cite{MH+Ye25}.  In Section~\ref{sec:mFL_DT} we also derive (multi-) matrix soliton solutions from the trivial solution $u=0$. 

Section~\ref{sec:vFL} deals with the vector version of the matrix FL equation, obtained when the matrix $u$ in (\ref{mFLeq}) reduces to a vector (also see \cite{Guo+Ling12,ZYCW17,LFZ18,Mats18,Mats19,Yang+Zhang18,KXM18,Wang+Chen19,YZCBG19,Gerd+Ivan21,Ling+Su24,SFY25}). In Section~\ref{sec:2vFL}, we finally concentrate on the two-component vector case. The vectorial Darboux transformation is used to generate all types of soliton solutions on a plane wave background solution in a systematic way. 

Finally, Section~\ref{sec:concl} contains some concluding remarks.

\section{From the Miura transformation in bidifferential calculus to matrix versions of the first negative flow of the KN hierarchy}
\label{sec:bidiff->KN-1}
A \emph{graded associative algebra} is an associative algebra 
$\boldsymbol{\Omega} = \bigoplus_{r \geq 0} \boldsymbol{\Omega}^r$
over a field $\mathbb{K}$ of characteristic zero, where $\mathcal{A} = \boldsymbol{\Omega}^0$ is an associative 
algebra over $\mathbb{K}$ and $\boldsymbol{\Omega}^r$, $r \geq 1$, are $\mathcal{A}$-bimodules such that 
$\boldsymbol{\Omega}^r \, \boldsymbol{\Omega}^s \subseteq \boldsymbol{\Omega}^{r+s}$. 
Elements of $\boldsymbol{\Omega}^r$ are called $r$-forms.
A \emph{bidifferential calculus} is a unital graded associative algebra $\boldsymbol{\Omega}$, supplied
with two $\mathbb{K}$-linear graded derivations 
$\d, \bd : \boldsymbol{\Omega} \rightarrow \boldsymbol{\Omega}$
of degree one (so that $\d \boldsymbol{\Omega}^r \subseteq \boldsymbol{\Omega}^{r+1}$,
$\bd \boldsymbol{\Omega}^r \subseteq \boldsymbol{\Omega}^{r+1}$), and such that
\bez
 \d^2 = 0 \, , \qquad  \bd^2 = 0 \, , \qquad  \d \bd + \bd \d = 0  \, .   
\eez
In the following, we show that the matrix versions (\ref{KN-1}) of the first negative flow of the KN hierarchy arise from the simple equation 
\be
    \d \phi = (\bd g) \, g^{-1}  \,  ,    \label{Miura}
\ee
for a special realization of the bidifferential calculus.
The interest in this equation stems from the fact that it possesses a bDT \cite{MH+Ye25}, which will be recalled in Section~\ref{sec:KN-1_bDT}. It allows to generate from a given solution $(\phi,g)$ a family of solutions $(\phi',g')$, depending on a 
number of parameters.   

Let $\mathcal{A}$ now be the associative algebra of smooth functions of independent real variables $x$ and $t$. 
We choose $\boldsymbol{\Omega} = \mathrm{Mat}(\mathcal{A}) \otimes \bigwedge(\mathbb{C}^2)$ and a basis $\xi_1,\xi_2$ of the Grassmann algebra $\bigwedge(\mathbb{C}^2)$. $\mathrm{Mat}(\mathcal{A})$ denotes the algebra of all matrices over $\mathcal{A}$, where the product of two matrices is set to zero if the sizes do not match.  
Defining $\d$ and $\bd$ on $\mathrm{Mat}(\mathcal{A})$, they extend in an obvious way to $\Omega$ (treating $\xi_1,\xi_2$ as ``constants" with respect to $\d$ and $\bd$).

For each $n>1$, let $J_n$ be a constant $n \times n$ matrix,  $J_n \neq I_n$, where $I_n$ denotes the $n \times n$ identity matrix. 
For an $m \times n$ matrix $F$ over $\mathcal{A}$, we set
\bez
       \d F = F_x \, \xi_1 + \frac{1}{2} (J_m F - F J_n) \, \xi_2 \, , \qquad
       \bd F = \frac{1}{2} (J_m F - F J_n) \, \xi_1 + F_t \, \xi_2 
\eez
(also see \cite{DKMH11acta,MH+Ye25}). 
In the following, we fix $m >1$. 
Let $\phi$ and $g$ be $m \times m$ matrices over $\mathcal{A}$, 
$g$ invertible. We write $J$ instead of $J_m$. 
Then the Miura (transformation) equation (\ref{Miura}) becomes the 
``Miura system"
\be
  \phi_x = \frac{1}{2} \, [J,g] \, g^{-1} \, , \qquad g_t g^{-1} = \frac{1}{2} \, [J,\phi] 
  \, .   \label{Miura_sys}
\ee

\begin{proposition}
\label{prop:det(g)_t=0}
(\ref{Miura_sys}) implies that $\det(g)$ does not depend on $t$.
\end{proposition}
\begin{proof}
In terms of $\tilde{g} = \det(g)^{-1/m} \, g$, so that $\det(\tilde{g}) = 1$, (\ref{Miura_sys}) becomes
\bez
  \phi_x = \frac{1}{2} \, [J,\tilde{g}] \, \tilde{g}^{-1} \, , \qquad
  \tilde{g}_t \tilde{g}^{-1} + \frac{1}{m} (\log \det(g))_t \, I_m = \frac{1}{2} \, [J,\phi] \, .
\eez
By taking the trace of the second equation, we obtain $\det(g)_t = 0$.
\end{proof}

Choosing
\bez
       J =\mbox{block-diagonal}(I_{m_1}, -I_{m_2})  \, , \qquad  m=m_1+m_2 \, ,
\eez
and decomposing the matrices $\phi$ and $g$ accordingly,
\bez
       \phi =  \left( \begin{array}{cc} \phi_{11} & \phi_{12} \\
  	\phi_{21} & \phi_{22} \end{array} \right) \, , \qquad
      g = \left( \begin{array}{cc} g_{11} & g_{12} \\
    	g_{21} & g_{22} \end{array} \right) \, , 
\eez
assuming invertibility of $g_{22}$ and also of the Schur complement of $g_{22}$ in $g$,
\bez
        S = g_{11} - g_{12} \, g_{22}^{-1} g_{21} \, , 
\eez
we have
\bez
       g^{-1} = \left( \begin{array}{cc} S^{-1} & - S^{-1} g_{12} \, g_{22}^{-1} \\
       - g_{22}^{-1} g_{21} S^{-1} & g_{22}^{-1} \big( I_{m_2} 
       + g_{21} S^{-1} g_{12} \, g_{22}^{-1} \big) 
          \end{array} \right) \, ,       
\eez
and the Miura system (\ref{Miura_sys}) takes the form
\be
&&  \left( \begin{array}{cc} g_{11} & g_{12} \\
    	g_{21} & g_{22} \end{array} \right)_t 
  = \left( \begin{array}{cc} \phi_{12} \, g_{21} & \phi_{12} \,  g_{22} \\
  	- \phi_{21} \, g_{11} & - \phi_{21} \, g_{12} \end{array} \right)\, , \label{Miura_sys1} \\
 && \left( \begin{array}{cc} \phi_{11} & \phi_{12} \\
  	\phi_{21} & \phi_{22} \end{array} \right)_x 
 = \left( \begin{array}{cc} - g_{12} \, g_{22}^{-1} \, g_{21} S^{-1} & 
       g_{12} \, g_{22}^{-1} \big( I_{m_2} + g_{21} S^{-1} g_{12} \, g_{22}^{-1} \big)  \\
 	- g_{21} S^{-1} & g_{21} S^{-1} g_{12} \,  g_{22}^{-1} \end{array} \right)  \, .	 
                                           \label{Miura_sys2} 	
\ee

\begin{remark}
The Miura system, (\ref{Miura_sys1}) and (\ref{Miura_sys2}), is invariant under $g \mapsto f \, g$ with 
any nowhere vanishing function $f(x)$. Any physical equation derived from it, should therefore only 
involve quotients of components of $g$. 
According to Proposition~\ref{prop:det(g)_t=0}, for each solution we can arrange that $\det(g)=1$ 
by a scaling.
\hfill $\Box$ 
\end{remark}

Let us introduce
\be
   u = \phi_{21} \, , \qquad 
   v = \imag \, g_{12} \, g_{22}^{-1}  \, ,   \label{u,v}
\ee
which are matrices of size $m_2 \times m_1$, respectively $m_1 \times m_2$.
The following proposition generalizes a result of \cite{LLZ25,MH+Ye25} from the 
scalar case to the arbitrary size matrix case.  

\begin{proposition}
\label{prop:cmFL_sol}	
As a consequence of (\ref{Miura_sys1}) and  (\ref{Miura_sys2}), $u$ and $v$ satisfy the corresponding matrix version (\ref{KN-1}) of the first negative flow of the KN hierarchy.
\end{proposition}
\begin{proof}
The equation for $\phi_{21x}$ in (\ref{Miura_sys2}) reads 
\be
                  u_x =-  g_{21} S^{-1} \, .    \label{u_x}
\ee
With its help, we can write $S = g_{11} - \imag \, v \, u_x \, S$, so that
\bez
     g_{11} = ( I_{m_1} + \imag \, v \, u_x )  \, S \, .
\eez
Using (\ref{Miura_sys1}), we obtain
\bez
        S_t = - \imag \, v \, u \, S \, .
\eez
Now we have
\bez
        u_{xt} = -  g_{21t} S^{-1} +  g_{21} S^{-1} S_t S^{-1}  
        = u \,  g_{11} S^{-1} + \imag \, u_x \, v \, u 
        = u + \imag \, \big( u \, v \, u_x + u_x \, v \, u \big) \, ,
\eez
and
\bez
       v_{tx} &=& \imag \, \big( g_{12t} \, g_{22}^{-1} - g_{12} g_{22}^{-1} g_{22t} \, g_{22}^{-1} \big)_x 
      = \imag \, \big( \phi_{12} - \imag \, v \, u \, g_{12} \,  g_{22}^{-1} \big)_x 
      = \imag \, \big( \phi_{12} - v \, u \, v \big)_x  \\
      &=& v - \imag \, ( v_x \, u \, v + v \, u \, v_x \big) \, ,
\eez
where we used  
\bez
    \phi_{12x} = g_{12} \, g_{22}^{-1} \big( I_{m_2} 
                        + g_{21} S^{-1} g_{12} \, g_{22}^{-1} \big)
                     = -\imag \, v \, ( I_{m_2} + \imag \, u_x v )  \, .
\eez
\end{proof}

Moreover, any solution of (\ref{KN-1}) determines a solution of the Miura system.

\begin{lemma}
\label{lem:cmFL_sol}	
Let $(u,v)$ be a solution\footnote{A solution may not be defined everywhere on $\mathbb{R}^2$.} 
of (\ref{KN-1}). Then 
\bez
	&&   g_{22t} = \imag \, u \, v \, g_{22}  \, ,  \qquad
	g_{12} = - \imag \, v \, g_{22}  \, , \qquad
	g_{21} = - u_x S \, , \qquad
	g_{11} = (I_{m_1} + \imag \, v \, u_x) \, S  \, ,  \\
	&&  \phi_{21} = u \, , \qquad \phi_{12} = - \imag \, v_t + v \, u \, v  \, , \qquad
    \phi_{11} = - \imag \, \int v \, u_x \, dx \, , \qquad
	\phi_{22} = \imag \, \int u_x v \, dx  \, ,
\eez
where $S$ is an invertible solution of 
\bez
	S_t = - \imag \, v \, u \, S \, , 
\eez
determines a solution of the Miura system, (\ref{Miura_sys1}) and  (\ref{Miura_sys2}).
\end{lemma}
\begin{proof}
The linear first order matrix ODEs for $g_{22}$ and $S$ possess invertible solutions (see \cite{Misg+Bale75}, for example), 
wherever $u$ and $v$ are smooth. 
The statement of the lemma is then easily verified. 
\end{proof}

\section{A vectorial binary Darboux transformation for matrix versions of the first negative flow of the KN hierarchy}
\label{sec:KN-1_bDT}
We recall the following result of bidifferential calculus \cite{MH+Ye25}.

\begin{theorem}
\label{thm:Miura}
Given a bidifferential calculus with maps $\d, \bd$, 
let 0-forms $\Delta, \Gamma$ and 1-forms  $\kappa,\lambda$ satisfy  
\be
	&&  \bd \Delta + [\lambda , \Delta] = (\d \Delta) \, \Delta \, , \qquad  
	\bd \lambda + \lambda^2 = (\d \lambda) \, \Delta \, , \nonumber \\ 
	&&  \bd \Gamma - [\kappa , \Gamma] = \Gamma \, \d \Gamma  \, , \qquad  
	\bd \kappa - \kappa^2 = \Gamma \, \d \kappa \, .   \label{Delta,lambda,Gamma,kappa_eqs}
\ee
Let 0-forms $\phi$ and $g$, the latter assumed to be invertible, be a solution of the Miura equation (\ref{Miura})
and $\theta, \eta$ solutions of the linear equations 
\be
	\bd \theta = (\d \phi) \, \theta + (\d \theta) \, \Delta + \theta \, \lambda \, , \qquad
	\bd \eta = - \eta \, \d \phi + \Gamma \, \d \eta + \kappa \, \eta \, .   \label{linsys1}
\ee
Furthermore, let $\Omega$ be an invertible solution of the linear equations
\be
	&& \Gamma \, \Omega - \Omega \, \Delta = \eta \, \theta \, ,  \label{preSylv} \\
	&& \bd \Omega = (\d \Omega) \, \Delta - (\d \Gamma) \, \Omega + \kappa \, \Omega + \Omega \, \lambda + (\d \eta) \, \theta \, .
	\label{linsys2}
\ee
Then, if $\Omega$ is invertible, 
\be
	\phi' = \phi - \theta \, \Omega^{-1} \eta + C \, , \qquad
	g' = g - \theta \,  \Omega^{-1} \Gamma^{-1} \eta \, g \, ,
	\label{phi',g'}
\ee
where $\d C =0$, also solve the Miura equation (\ref{Miura}).
\hfill $\Box$
\end{theorem}

Elaborating this result for the choice of bidifferential calculus in Section~\ref{sec:bidiff->KN-1}, leads to the following theorem.\footnote{Here we restrict our considerations to a framework of \emph{finite}-dimensional matrices. But under suitable technical assumptions, we may let $n$ in Theorem~\ref{thm:bDT_matrix_pKN-1} to be infinite, or better 
generalize the matrices involving the matrix size $n$ to 
become (bounded) operators on a Banach space. Also see \cite{Manas96}, 
for example.}	 
	
\begin{theorem}
\label{thm:bDT_matrix_pKN-1}
Let $\Delta$ and $\Gamma$ be invertible constant $n \times n$ matrices. Let $(u,v)$ (where $u$ and $v$ have matrix size $m_2 \times m_1$, respectively $m_1 \times m_2$) be a solution of the matrix version (\ref{KN-1}) of the first negative flow of the KN hierarchy on an open set $\mathcal{U} \subset \mathbb{R}^2$. 
Furthermore, let $\theta_1,\theta_2$  be $m_1 \times n$, respectively $m_2 \times n$, matrix solutions of the linear system
\be
&&   \theta_{1x} \, \Delta -  ( \frac{1}{2} I_{m_1} + \imag \, v u_x ) \, \theta_1 \,
     - \imag \, v \, (I_{m_2} + \imag \, u_x v) \,  \theta_2 = 0 \, , \qquad
   \theta_{2x} \, \Delta + (\frac{1}{2} I_{m_2} + \imag \, u_x v) \, \theta_2 \, + u_x \, \theta_1 = 0 \, , \nonumber \\
&&  \theta_{1t} - \frac{1}{2} \theta_1 \Delta + (\imag \, v_t - v u v) \, \theta_2 = 0 \, , \qquad
    \theta_{2t} + \frac{1}{2} \theta_2 \Delta +u \, \theta_1 = 0 \, ,  \label{mFL_linsys_theta}
\ee
and $\eta_1, \tilde{\eta}_2$ matrix solutions of size $n \times m_1$, respectively $n \times m_2$ of 
the adjoint linear system
\be
&&  \Gamma \, \eta_{1x} + \frac{1}{2} \eta_1 - \tilde{\eta}_2 \, u_x = 0 \, , \qquad
	\Gamma \, \tilde{\eta}_{2x} - \frac{1}{2} \tilde{\eta}_2 + \imag \, \Gamma \eta_1 v_x = 0 \, ,   \nonumber \\
&& \eta_{1t} + \frac{1}{2} \Gamma \, \eta_1 - \imag \, \eta_1 \, v \, u  - \tilde{\eta}_2 \, u = 0 \, , \qquad
   \tilde{\eta}_{2t} - \frac{1}{2} \Gamma \, \tilde{\eta}_2 + \imag \, \tilde{\eta}_2 \, u \, v - \imag \, \Gamma \eta_1 v = 0 \, .
     \label{mFL_linsys_eta}
\ee
Furthermore, let $\Omega$ be an $n \times n$ matrix solution of the Sylvester equation 
\be
    \Gamma \, \Omega - \Omega \, \Delta = \eta_1 \, (\theta_1 + \imag \, v \, \theta_2) + \tilde{\eta}_2 \, \theta_2 
    \, ,  \label{mFL_Sylv}
\ee
and the linear differential equations
\be
 \Omega_x \Delta = - \eta_{1x} (\theta_1 + \imag \, v \, \theta_2) - (\imag \, \eta_1 v_x + \tilde{\eta}_{2x} ) \,  \theta_2 \, , \qquad   
 \Omega_t = - \frac{1}{2} \eta_1 (\theta_1 - \imag \, v \, \theta_2) + \frac{1}{2} \tilde{\eta}_2 \, \theta_2  \, .   \label{mFL_Om_deriv}
\ee
Then, in any open subset of $\mathcal{U}$ where $\Omega$ is invertible,
\be
 &&  u' = u - \theta_2\,  \Omega^{-1} \eta_1 \, , \label{u'} \\
 &&   v' = \big( v - \imag \, \theta_1 (\Gamma \Omega)^{-1} \tilde{\eta}_2 \big)  
\big( I_{m_2} - \theta_2 (\Gamma \Omega)^{-1} \tilde{\eta}_2 \big)^{-1} \,  ,  
  \label{v'}
\ee
solve (\ref{KN-1}).
\end{theorem}
\begin{proof}
Writing  $\kappa = \kappa_1 \xi_1 + \kappa_2 \xi_2$, $\lambda = \lambda_1 \xi_1 + \lambda_2 \xi_2$,
choosing $\kappa_1 = \frac{1}{2} J_n$, $\kappa_2 = - \frac{1}{2} \Gamma J_n$, 
$\lambda_1 = - \frac{1}{2} J_n$, $\lambda_2 = \frac{1}{2} J_n \Delta$, and using 
\bez
       \d \phi = \phi_x \, \xi_1 + \frac{1}{2} [J, \phi] \, \xi_2 \, , \qquad
     \frac{1}{2} [J , \phi] = \left( \begin{array}{cc} 0 & \phi_{12} \\ - \phi_{21} & 0 \end{array} \right) \, ,
\eez
(\ref{linsys1}) takes the form
\bez
&&   \frac{1}{2} J \, \theta = \phi_x \, \theta + \theta_x \Delta  \, , \qquad 
\theta_t = \frac{1}{2} [J , \phi] \, \theta + \frac{1}{2} J \, \theta \Delta  \, ,  \\
&&   \frac{1}{2} \eta \, J =  \eta \, \phi_x - \Gamma \eta_x \, , \qquad
\eta_t = - \frac{1}{2} \eta \, [J , \phi] - \frac{1}{2} \Gamma \eta \, J \, .
\eez
Writing further
\bez
    \theta = \left( \begin{array}{c} \theta_1 \\ \theta_2 \end{array} \right) \, , 
    \qquad
    \eta = \left( \begin{array}{cc} \eta_1 & \eta_2 \end{array} \right) \, ,
\eez
with an $m_1 \times n$ matrix $\theta_1$, an $m_2 \times n$ matrix $\theta_2$, an $n \times m_1$ matrix $\eta_1$ and an $n \times m_2$ matrix
$\eta_2$, and using the solution of the Miura system corresponding to $(u,v)$ according to Lemma~\ref{lem:cmFL_sol}, the last system becomes (\ref{mFL_linsys_theta}) and 
\be
&&  \Gamma \, \eta_{1x} +\eta_1 \, ( \frac{1}{2} I_{m_1} + \imag \, v u_x ) - \eta_2 u_x = 0 \, , \qquad
	\Gamma \, \eta_{2x} - \eta_2 \,  (\frac{1}{2} I_{m_2} + \imag \, u_x v) 
	+ \imag \, \eta_1 \,  v \, (I_{m_2} + \imag \, u_x v) = 0 \, ,   \nonumber \\
&& \eta_{1t} + \frac{1}{2} \Gamma \, \eta_1 - \eta_2 \, u = 0 \, , \qquad
   \eta_{2t} - \frac{1}{2} \Gamma \, \eta_2 - \eta_1 (\imag \, v_t - v u v) = 0 \, .
\ee
Introducing 
\bez
       \tilde{\eta}_2 = \eta_2 - \imag \, \eta_1 v \, ,
\eez
this is turned into (\ref{mFL_linsys_eta}).  
The conditions in (\ref{Delta,lambda,Gamma,kappa_eqs}) amount to $\Delta$ and $\Gamma$ being constant. (\ref{linsys2}) takes the form 
\bez
   \Omega_x \Delta = - \eta_x \theta \, , \qquad   
   \Omega_t = - \frac{1}{2} \eta \, J \, \theta \, ,
\eez
which is (\ref{mFL_Om_deriv}). To derive the second equation, we used the Sylvester equation. The statement of the theorem now follows from Theorem~\ref{thm:Miura}, which says that
\bez
 \phi' = \phi - \theta \, \Omega^{-1} \eta + C \, , \qquad
 g' = (I_m - \theta \,  \Omega^{-1} \Gamma^{-1} \eta) \, g \, ,
\eez
where $C$ is any $\d$-constant $2 \times 2$ block-matrix, solve the Miura equation, which via Proposition~\ref{prop:cmFL_sol} implies that $u' = \phi'_{21}$ and
\bez
  &&  v' = \imag \, g'_{12} \, {g'}_{22} ^{-1} \\
       &=& \imag \, \big( (I_{m_1} - \theta_1 \Omega^{-1} \Gamma^{-1} \eta_1 ) \, g_{12}  
          - \theta_1 \Omega^{-1} \Gamma^{-1} \eta_2 \, g_{22} \big) 
        \big(  (I_{m_2} - \theta_2 \Omega^{-1} \Gamma^{-1} \eta_2) \, g_{22} 
         - \theta_2 \Omega^{-1} \Gamma^{-1} \eta_1 \, g_{12} \big)^{-1} 
\eez
solve (\ref{KN-1}). Noting that $C$ is $\d$-constant if and only 
if $C = \mbox{block-diag}(c_1,c_2)$ with $c_{ix}=0$, $i=1,2$, and using (\ref{u,v}), we arrive at (\ref{u'}) and (\ref{v'}). 
\end{proof}

\begin{remark}
\label{rem:Delta,Gamma_Jnf}
(\ref{mFL_linsys_theta}) and (\ref{mFL_linsys_eta}) are invariant under
\bez
    &&  \theta_i \mapsto \theta_i S \, , \qquad \eta_i \mapsto T \eta_i \qquad i=1,2 \, , \\
    && \Delta \mapsto S^{-1} \Delta S \, , \qquad \Gamma \mapsto T \Gamma T^{-1} \, ,
\eez
with any constant invertible $n \times n$ matrices $S$ and $T$. As a consequence, without restriction of generality, we may assume that $\Delta$ and $\Gamma$ are in (possibly different, upper or lower triangular) Jordan normal form. \hfill $\Box$
\end{remark}

\section{A vectorial Darboux transformation for the matrix FL equation}
\label{sec:mFL_DT}
As mentioned in the introduction, (\ref{KN-1}) reduces to the matrix FL equation (\ref{mFLeq}) by imposing the reduction 
condition (\ref{red}). In this section, we extend this reduction to the vectorial binary Darboux transformation, expressed in Theorem~\ref{thm:bDT_matrix_pKN-1}, in such a way that a vectorial Darboux transformation for the matrix FL equation (\ref{mFLeq}) is obtained.

\begin{theorem}
	\label{thm:mFL_DT}
Let $\Gamma$ be an invertible constant $n \times n$ matrix, such that $\Gamma$ satisfies the ``spectrum condition", i.e., $\Gamma$ and $-\Gamma^\dagger$ have no common eigenvalue. Let $u$ be a solution of the $m_2 \times m_1$ matrix Fokas-Lenells equation (\ref{mFLeq}) on an open set 
$\mathcal{U} \subset \mathbb{R}^2$. 
Furthermore, let $\eta_1$ and $\tilde{\eta}_2$ be $n \times m_1$, respectively $n \times m_2$, matrix solutions of the linear system
\be
	&&  \Gamma \, \eta_{1x} + \frac{1}{2} \eta_1 - \tilde{\eta}_2 \, u_x = 0 \, , \qquad\quad
	\eta_{1t} + \frac{1}{2} \Gamma  \eta_1 - \imag \,  \eta_1 \, \sigma_1 u^\dagger \sigma_2 u  - \tilde{\eta}_2 \, u = 0
	 \, ,   \nonumber \\
   && \Gamma \tilde{\eta}_{2x} - \frac{1}{2} \tilde{\eta}_2 + \imag \, \Gamma \eta_1 \,  \sigma_1 u^\dagger_x \sigma_2 = 0 \, , \qquad
     \tilde{\eta}_{2t} - \frac{1}{2} \Gamma  \tilde{\eta}_2
      + \imag \, \tilde{\eta}_2 \,  u \, \sigma_1 u^\dagger \sigma_2  - \imag \, \Gamma \eta_1 \, \sigma_1 u^\dagger \sigma_2 = 0 \, .
	\label{mFLeq_linsys}
\ee
Furthermore, let $\Omega$ be an $n \times n$ matrix solution of the Lyapunov equation 
\be
   \Gamma \Omega + \Omega \Gamma^\dagger 
 = \imag \, \eta_1 \sigma_1 \eta_1^\dagger \Gamma^\dagger 
   + \tilde{\eta}_2 \, \sigma_2 \,  \tilde{\eta}_2^\dagger \, . \label{mFLeq_Lyap}
\ee
Then, in any open subset of $\mathcal{U}$ where $\Omega$ is invertible,
\be
    u' = u - \sigma_2 \, \tilde{\eta}_2^\dagger \, \Omega^{-1} \eta_1   \label{mFLeq_u'}
\ee
solves the $m_2 \times m_1$ matrix FL equation (\ref{mFLeq}).
\end{theorem}
\begin{proof}
In Theorem~\ref{thm:bDT_matrix_pKN-1}, we impose the conditions
\be
\Delta = - \Gamma^\dagger \, , \quad
\theta_1 = \imag \, \sigma_1 ( \eta_1^\dagger \, \Gamma^\dagger - u^\dagger \, \tilde{\eta}_2^\dagger ) \, , \quad
\theta_2 = \sigma_2 \,  \tilde{\eta}_2^\dagger \, , \label{mFL_red}
\ee
where now 
\bez
     \tilde{\eta}_2 = \eta_2 - \imag \, \eta_1 \sigma_1 u^\dagger \sigma_2 \, .
\eez
Then, by using (\ref{red}), the four equations (\ref{mFL_linsys_theta}) turn out to be consequences of (\ref{mFL_linsys_eta}), and the latter system can be equivalently replaced by (\ref{mFLeq_linsys}).
The Sylvester equation (\ref{mFL_Sylv}) becomes the Lyapunov equation (\ref{mFLeq_Lyap}). We have 
$\Omega = \Omega_1 + \Omega_2$, where 
\bez
 \Gamma \Omega_1 + \Omega_1 \Gamma^\dagger = \imag \, \eta_1 \sigma_1 \eta_1^\dagger \Gamma^\dagger \, , \qquad
 \Gamma \Omega_2 + \Omega_2 \Gamma^\dagger = \tilde{\eta}_2 \, \sigma_2 \, \tilde{\eta}_2^\dagger \, .
\eez
Since we assume the spectrum condition for $\Gamma$, the Lyapunov equation has a unique solution.
It follows that $\Omega_2^\dagger = \Omega_2$, and thus
\bez
    \tilde{\eta}_2 \, \sigma_2 \, \tilde{\eta}_2^\dagger 
   = \Gamma \Omega + \Omega^\dagger \Gamma^\dagger 
       - (\Gamma \Omega_1 + \Omega_1^\dagger \Gamma^\dagger) \, .
\eez
The last bracket vanishes since the Lyapunov equation for $\Omega_1$ implies $(\Gamma \Omega_1)^\dagger = - \Gamma \Omega_1$.
Hence
\be
    \tilde{\eta}_2\,  \sigma_2 \, \tilde{\eta}_2^\dagger  = \Gamma \Omega + \Omega^\dagger \Gamma^\dagger \, .    
                       \label{mFLeq_key}
\ee
$u'$ in (\ref{u'}) takes the form (\ref{mFLeq_u'}), and $v'$ can be written as
\bez
     v' = G \, H^{-1} \, , 
\eez
with
\bez
 G = \sigma_1 u^\dagger \sigma_2 + \sigma_1 \big( \Gamma \eta_1 - \tilde{\eta}_2 \, u \big)^\dagger \, (\Gamma \Omega)^{-1} \tilde{\eta}_2  \, , \qquad
 H = I_{m_2} - \sigma_2 \, \tilde{\eta}_2^\dagger \, (\Gamma \Omega)^{-1} \tilde{\eta}_2 \, .
\eez
Now 
\bez
  \sigma_1  u'^\dagger \sigma_2 \, H &=&  \big( \sigma_1 u^\dagger \sigma_2 -\sigma_1 \, \eta_1^\dagger \, \Omega^{\dagger -1} \tilde{\eta}_2 \big)   \big( I_{m_2} - \sigma_2 \, \tilde{\eta}_2^\dagger \, (\Gamma \Omega)^{-1} \tilde{\eta}_2 \big) \\
   &=& \sigma_1 u^\dagger \sigma_2 
     - \sigma_1 u^\dagger \, \tilde{\eta}_2^\dagger \, (\Gamma \Omega)^{-1} \tilde{\eta}_2
     -\sigma_1 \eta_1^\dagger \, \Omega^{\dagger -1} \tilde{\eta}_2 
     +\sigma_1 \eta_1^\dagger \, \Omega^{\dagger -1} \tilde{\eta}_2 \sigma_2 \tilde{\eta}_2^\dagger 
     \, (\Gamma \Omega)^{-1} \tilde{\eta}_2 = G
   \, ,
\eez 
where we used (\ref{mFLeq_key}) in the last step, shows that $v' =\sigma_1 u'^\dagger \sigma_2$. 
Since, according to Theorem~\ref{thm:bDT_matrix_pKN-1}, $(u',v')$ solves (\ref{KN-1}), it follows that $u'$ solves the matrix Fokas-Lenells equation (\ref{mFLeq}). 
\end{proof}

\begin{corollary}
\label{cor:mFL_no_sc}	
The spectrum condition for $\Gamma$ can be dropped in Theorem~\ref{thm:mFL_DT} if the assumptions there are supplemented by	
\be
 && \Omega_x  = \imag \, \eta_{1x} \, \sigma_1 \, \eta_1^\dagger 
    + (\tilde{\eta}_{2x} \, \sigma_2 + \imag \, \eta_1 \sigma_1 \, u^\dagger_x ) \, \tilde{\eta}_2^\dagger \,  \Gamma^{\dagger -1} \, ,  
      \nonumber \\
 &&  \Omega_t = \frac{1}{2} \Big( -\imag \, \eta_1 \sigma_1 \eta_1^\dagger  \Gamma^\dagger + 2 \imag \, \eta_1 \sigma_1 u^\dagger \tilde{\eta}_2^\dagger  + \tilde{\eta}_2 \sigma_2 \tilde{\eta}_2^\dagger \Big) \, , 
    \label{mFLeq_Om_deriv}
\ee	
and if in addition (\ref{mFLeq_key}) is imposed. 
\end{corollary}
\begin{proof}
(\ref{mFLeq_Om_deriv}) results from (\ref{mFL_Om_deriv}), using (\ref{mFL_red}). According to Theorem~\ref{thm:bDT_matrix_pKN-1},
taking a seed that solves the matrix FL equation, the resulting 
$u'$ and $v'$ then solve (\ref{KN-1}). 
Now (\ref{mFLeq_key}) guarantees that $u'$ actually solves 
the matrix FL equation.  	

As a consequence of (\ref{mFLeq_Om_deriv}) and the linear system (\ref{mFLeq_linsys}), it follows that
\bez
   \big( \Gamma \Omega + \Omega^\dagger \Gamma^\dagger - \tilde{\eta}_2 \, \sigma_2 \tilde{\eta}_2^\dagger \big)_x = 0 \, , \qquad
   \big( \Gamma \Omega + \Omega^\dagger \Gamma^\dagger - \tilde{\eta}_2 \, \sigma_2 \tilde{\eta}_2^\dagger \big)_t = 0 \, ,
\eez
so that (\ref{mFLeq_key}) and (\ref{mFLeq_Om_deriv}) are indeed compatible.
\end{proof}

\begin{remark}
If the spectrum condition for $\Gamma$ holds, 
then (\ref{mFLeq_key}) and (\ref{mFLeq_Om_deriv}) are automatically satisfied as a consequence of the Lyapunov equation (\ref{mFLeq_Lyap}).
\hfill $\Box$
\end{remark}

\begin{remark}
	\label{rem:FL_superposition}
A crucial feature of the vectorial Darboux transformation is the following ``superposition property''. 
Let $(\Gamma^{(i)},\eta^{(i)}_1,\tilde{\eta}^{(i)}_2,\Omega^{(i)})$,
$i=1,2$, be data determining solutions ${u'}^{(i)}$ of the matrix FL equation, for the same seed solution $u$. Let
\bez
   \Gamma = \left( \begin{array}{cc} \Gamma^{(1)} & 0 \\
   0 & \Gamma^{(2)} \end{array} \right) \, , \quad
   \eta_1 = \left(\begin{array}{c} \eta^{(1)}_1 \\ \eta^{(2)}_1
   	\end{array} \right) \, , \quad
   \tilde{\eta}_2 = \left(\begin{array}{c} \tilde{\eta}^{(1)}_2 \\ \tilde{\eta}^{(2)}_2
   \end{array} \right) \, , \quad	
   \Omega = \left( \begin{array}{cc} \Omega^{(1)} & \Omega^{(12)} \\
   	\Omega^{(21)} & \Omega^{(2)} \end{array} \right) \, .
\eez
The corresponding linear system (\ref{mFLeq_linsys}) is then solved. Essentially, it only remains to determine the matrices $\Omega^{(12)}$ and $\Omega^{(21)}$ in order to get a new solution $u'$, which is then a kind of nonlinear superposition of ${u'}^{(1)}$ and ${u'}^{(2)}$. The Lyapunov equation (\ref{mFLeq_Lyap}) 
reduces to the two Sylvester equations
\bez
 &&  \Gamma^{(1)} \Omega^{(12)} + \Omega^{(12)} \Gamma^{(2)\dagger}
   = \imag \, \eta^{(1)}_1 \sigma_1 \, \eta^{(2)\dagger}_1 \Gamma^{(2)\dagger}
   + \tilde{\eta}^{(1)}_2 \sigma_2 \, \tilde{\eta}^{(2)\dagger}_2 \, , 
    \\
 && \Gamma^{(2)} \Omega^{(21)} + \Omega^{(21)} \Gamma^{(1)\dagger}
 = \imag \, \eta^{(2)}_1 \sigma_1 \, \eta^{(1)\dagger}_1 \Gamma^{(1)\dagger}
 + \tilde{\eta}^{(2)}_2 \sigma_2 \, \tilde{\eta}^{(1)\dagger}_2 \, .
\eez
If $\Gamma^{(1)}$ and $-\Gamma^{(2)\dagger}$ have no common eigenvalue, these equations possess unique solutions $\Omega^{(12)}$ and $\Omega^{(21)}$. Otherwise additional equations, obtained from (\ref{mFLeq_key}) and (\ref{mFLeq_Om_deriv}), have to be solved.	
\hfill $\Box$
\end{remark}

\subsection{Zero seed solutions}
For $u=0$ (zero seed), (\ref{mFLeq_linsys}) is completely solved by 
\bez
 \eta_1 = e^{-\frac{1}{2} (\Gamma^{-1} x + \Gamma \, t)} \, a_1 
\, ,  \qquad
 \tilde{\eta}_2 = e^{\frac{1}{2} (\Gamma^{-1} x + \Gamma \, t)} \, a_2
\, ,
\eez
with constant $n \times m_1$, respectively $n \times m_2$, matrices $a_1$ and $a_2$. The Lyapunov equation reads 
\bez
\Gamma \, \Omega + \Omega \, \Gamma^\dagger 
= \imag \, e^{- \frac{1}{2} (\Gamma^{-1} \, x + \Gamma \, t)} \, 
\boldsymbol{a} \, \Gamma^\dagger \, 
e^{- \frac{1}{2} ({\Gamma^\dagger}^{-1} \, x + \Gamma^\dagger \, t)}
+ e^{\frac{1}{2} (\Gamma^{-1} \, x + \Gamma \, t)} \, \boldsymbol{b} \, e^{\frac{1}{2} ({\Gamma^\dagger}^{-1} \, x + \Gamma^\dagger \, t)}  \, ,
\eez
where we introduced the Hermitian $n \times n$ matrices
\bez
 \boldsymbol{a} = a_1 \, \sigma_1 \, a_1^\dagger\, , \qquad \boldsymbol{b} = a_2 \, \sigma_2 \, a_2^\dagger\, .
\eez
If the spectrum condition for $\Gamma$ holds, the Lyapunov equation possesses a unique solution $\Omega$. If not, then the Lyapunov equation leads to constraints and we also have to use the equations in Corollary~\ref{cor:mFL_no_sc} to determine $\Omega$.
According to Theorem~\ref{thm:mFL_DT} and Corollary~\ref{cor:mFL_no_sc}, 
\be
 u' = - \sigma_2 \, \tilde{\eta}_2^\dagger \, \Omega^{-1} \eta_1 
    = - \sigma_2 \, a_2^\dagger \, e^{\frac{1}{2} ({\Gamma^\dagger}^{-1} x 
    	+ \Gamma^\dagger \, t)} \, \Omega^{-1} e^{-\frac{1}{2} (\Gamma^{-1} x + \Gamma \, t)} \, a_1     
	\label{mFL_zero_seed_u'}
\ee
solves the matrix FL equation (\ref{mFLeq}). In the following, we set
\bez
    \varphi(\Gamma) = \frac{1}{2} (\Gamma^{-1} \, x + \Gamma \, t) \, ,
    \label{mFL_zero_seed_varphi}
\eez
and work out briefly some classes of solutions. 

\begin{example}
	\label{ex:bright-solitons}
Ordinary matrix multi-soliton solutions are obtained if $\Gamma = \mathrm{diag}(\gamma_1,\ldots,\gamma_n)$, with $\gamma_i \neq - \gamma_j^\ast$ for $i,j=1,\ldots,n$.
Then $\Omega$ is given by the generalized Cauchy matrix with components
\bez
	\Omega_{ij} = \frac{1}{\gamma_i + \gamma_j^\ast} \, \Big( 
	\imag \, \boldsymbol{a}_{ij} \, \gamma_j^\ast \, e^{-\varphi(\gamma_i) - \varphi(\gamma_j^\ast)} 
	+ \boldsymbol{b}_{ij} \, e^{\varphi(\gamma_i) + \varphi(\gamma_j^\ast)} \Big) \, .   	   
\eez
Now (\ref{mFL_zero_seed_u'}) yields an $n$-soliton solution of the $m_2 \times m_1$ matrix FL equation. 
For $n=1$, this is the matrix single soliton solution
\bez
 u' = - \frac{2 \, \mathrm{Re}(\gamma) \, e^{-2 \, \imag \, \mathrm{Im}(\varphi)}}
	{\big( \imag \, \mathrm{Re}(\gamma) + \mathrm{Im}(\gamma) \big) \,  \boldsymbol{a} \, e^{-2 \mathrm{Re}(\varphi)}
 + \boldsymbol{b} \, e^{2 \mathrm{Re}(\varphi)}} \, \sigma_2 \, a_2^\dagger \, a_1 \, ,
\eez
where $a_2^\dagger \, a_1$ is a Kronecker product and we wrote $\varphi$ for $\varphi(\gamma)$. 
Since $\mathrm{Re}(\gamma) \neq 0$ and $\boldsymbol{a}, \boldsymbol{b} \in \mathbb{R}$, this is non-singular (i.e., defined everywhere on $\mathbb{R}^2$) and, 
if $\boldsymbol{a} \neq 0$ and $\boldsymbol{b} \neq 0$, it vanishes as $x\to \pm \infty$ or $t \to \pm \infty$, except on the lines where $x +|\gamma|^2 t$ is constant. 

By Remark~\ref{rem:FL_superposition}, for $n>1$ a nonlinear superposition of $n$ single solitons is expected, which asymptotically (i.e., for large $|x|$ or $|t|$) separate. Regularity conditions for 
the above $n$-soliton solution is still an open problem\footnote{If the right hand side of the Lyapunov equation 
is a rank one matrix and if $\Gamma$ is diagonal, then the solution is a Cauchy-like matrix and its invertibility easily decided, see Appendix~\ref{app:Lyapunov} (also see \cite{Schie10} for more general results in the rank one case). In the present case, however, the right hand side of the Lyapunov equation has higher rank and corresponding applicable invertibility results are not yet available, to the best of our knowledge. Already for the \emph{scalar} FL equation ($m_1=m_2=1$) the right hand side of the Lyapunov equation has rank two (whereas we have rank one in the case of the scalar NLS equation), and general non-singularity results have also \emph{not} been reached in previous studies, independent of the solution-generating method.}, 
for $n=2$ this is partly resolved in Appendix~\ref{app:2-soliton-solution,inv_of_Omega}. 

For vector FL equations, corresponding solutions appeared 
in \cite{LFZ18,Mats18}. 
\hfill $\Box$
\end{example}

\begin{example}
Let $\Gamma = - \imag \, \mathrm{diag}(k_1,\ldots,k_n)$, $k_i \in \mathbb{R} \setminus \{0\}$ and $k_i\neq k_j$, $i,j=1,\ldots,n$. Then the Lyapunov equation becomes the constraint 
\bez
    \boldsymbol{b}_{ii} = k_i \, \boldsymbol{a}_{ii} \, , 
   \qquad	i=1,\ldots,n  \, .
\eez
The differential equation for $\Omega_{ii}$ reads
\bez
	\Omega_{iix} = \frac{\boldsymbol{a}_{ii}}{k_i} \, , \qquad 
	\Omega_{iit} = k_i \, \boldsymbol{a}_{ii} \, ,
\eez
so that
\bez
	\Omega_{ii} = \boldsymbol{a}_{ii} (k_i^{-1} x + k_i \, t) + c_i \, .
\eez
with a constant $c_i$, for which (\ref{mFLeq_key}) requires
\bez
	\mathrm{Im}(c_i) = \frac{1}{2} \boldsymbol{a}_{ii} \, .
\eez
For $i\neq j$, we obtain
\bez
	\Omega_{ij} = \frac{1}{\imag \, (k_i-k_j)} \, \Big( k_j \,  \boldsymbol{a}_{ij} \, e^{-\frac{\imag}{2} (k_i^{-1} - k_j^{-1}) (x+k_ik_j\,t)} - \boldsymbol{b}_{ij} \, e^{\frac{\imag}{2} (k_i^{-1}-k_j^{-1}) (x+k_ik_j\,t)} \Big) \, .
\eez
Inserting our results in (\ref{mFL_zero_seed_u'}), we obtain  another $n$-soliton solution of the $m_2 \times m_1$ matrix FL equation. For $n=1$, this is the quasi-rational matrix single soliton 
\bez
	u' = -\frac{e^{-\imag \, \big(k^{-1} x-k \,t \big)}}{
		 \boldsymbol{a} \, (k^{-1} x + k \, t) + \mathrm{Re}(c_1) + \frac{1}{2} \imag \, \boldsymbol{a}} 
		 \, \sigma_2 \, a_2^\dagger \, a_1 \, ,
\eez
which is non-singular. The absolute values of the components fall off toward zero away from the line $x = -k^2 t - \mathrm{Re}(c_1) \, k/\boldsymbol{a}$. In this sense, the solution is concentrated along this line.
\hfill $\Box$
\end{example}

\begin{example}
Let $\Gamma = -\imag \, k \, I_n$, where $k \in \mathbb{R} \setminus \{0\}$. Then the Lyapunov equation becomes the constraint $\boldsymbol{b} = k \, \boldsymbol{a}$,
and (\ref{mFLeq_Om_deriv}) leads to $\Omega = (k^{-1} x + k \, t) \, \boldsymbol{a} + \boldsymbol{c}$, with a constant $n \times n$ matrix $\boldsymbol{c}$, for which (\ref{mFLeq_key}) requires $\boldsymbol{a} + \imag \, (\boldsymbol{c}-\boldsymbol{c}^\dagger) =0$. The resulting solution is 
$u'= -e^{- \imag \, (k^{-1} x - k \, t)} \, \sigma_2 \, a_2^\dagger \, \Omega^{-1} a_1$, wherever $\Omega$ is invertible.
\hfill $\Box$
\end{example}

\begin{example}
	\label{ex:zero_seed_n=2_anti-conjugate_gammas}
Let $n=2$ and $\Gamma = \mathrm{diag}(\gamma,-\gamma^\ast)$ 
with $\mathrm{Re}(\gamma) \neq 0$. Then we have
\bez
 \eta_1 = \left( \begin{array}{cc} e^{-\varphi(\gamma)} & 0 \\ 0 & e^{ \varphi(\gamma^\ast)} \end{array} \right) \, a_1 \, , \qquad 
\tilde{\eta}_2 = \left( \begin{array}{cc} e^{\varphi(\gamma)} & 0 \\ 0 & e^{-\varphi(\gamma^\ast)}
\end{array} \right) \, a_2 \, .
\eez
Hermiticity of the matrices $\boldsymbol{a}$ and $\boldsymbol{b}$ means that 
the components $\boldsymbol{a}_{11}, \boldsymbol{a}_{22}, \boldsymbol{b}_{11}, \boldsymbol{b}_{22}$ are real and $\boldsymbol{a}_{21} = \boldsymbol{a}_{12}^\ast, \boldsymbol{b}_{21} = \boldsymbol{b}_{12}^\ast$. 
The off-diagonal parts of the Lyapunov equation amount to the constraint
\bez
    \boldsymbol{b}_{12} = \imag \, \gamma \, \boldsymbol{a}_{12} \, , 
\eez
and the diagonal parts determine the diagonal components of $\Omega$,
\bez
   \Omega_{11} &=& \frac{1}{2 \mathrm{Re}(\gamma)} \big( \imag \, \gamma^\ast \boldsymbol{a}_{11} \, e^{-2 \mathrm{Re}(\varphi(\gamma))} + \boldsymbol{b}_{11} \, e^{2 \mathrm{Re}(\varphi(\gamma))} \big) \, , \\
   \Omega_{22} &=& \frac{1}{2 \mathrm{Re}(\gamma)} \big( \imag \, \gamma \, \boldsymbol{a}_{22} \, e^{2 \mathrm{Re}(\varphi(\gamma))} - \boldsymbol{b}_{22} \, e^{-2 \mathrm{Re}(\varphi(\gamma))} \big) \, .
\eez
The differential equations for $\Omega$ are now equivalent to
\bez
    \Omega_{12x} = - \frac{\imag \, \boldsymbol{a}_{12}}{\gamma} \, , \quad
    \Omega_{21x} = \frac{\imag \, \boldsymbol{a}_{12}^\ast}{\gamma^\ast} \, , \quad
    \Omega_{12t} = \imag \, \gamma \, \boldsymbol{a}_{12} \, , \quad
    \Omega_{21t} = -\imag \, \gamma^\ast \boldsymbol{a}_{12}^\ast \, ,
\eez
which integrates to
\bez
   \Omega_{12} = - \imag \, \boldsymbol{a}_{12} \, (\gamma^{-1} x - \gamma t) + c_{12} \, , \qquad
   \Omega_{21} = \imag \, \boldsymbol{a}_{12}^\ast \, (\gamma^{\ast-1} x - \gamma^\ast \, t) + c_{21} \, , 
\eez 
with complex constants $c_{12}$ and $c_{21}$. The remaining condition (\ref{mFLeq_key}) becomes
\bez
    c_{21} = c_{12}^\ast + \imag \, \boldsymbol{a}_{12}^\ast \, .
\eez 
Inserting the results in (\ref{mFL_zero_seed_u'}) yields the following solution of 
the $m_2 \times m_1$ matrix FL equation,
\bez
 u' =  -\frac{1}{\mathrm{det}(\Omega)} \, \sigma_2 \, a_2^\dagger \left( \begin{array}{cc} e^{-2\imag \, \mathrm{Im}(\varphi(\gamma))} \, \Omega_{22} & -e^{2\varphi(\gamma^\ast)} \, \Omega_{12} \\ -e^{-2\varphi(\gamma)} \, \Omega_{21} & e^{-2\mathrm{i} \mathrm{Im}(\varphi(\gamma))} \, \Omega_{11} 
\end{array} \right) a_1 \, ,
\eez
where
\bez
\mathrm{det}(\Omega) &=& \frac{1}{4\mathrm{Re}(\gamma)^2} \Big( \imag \, \gamma \, \boldsymbol{a}_{22} \boldsymbol{b}_{11} \, e^{4\mathrm{Re}(\varphi(\gamma))} - \imag \, \gamma^\ast \boldsymbol{a}_{11} \boldsymbol{b}_{22} \, e^{-4\mathrm{Re}(\varphi(\gamma))} - |\gamma|^2 \boldsymbol{a}_{11} \boldsymbol{a}_{22} - \boldsymbol{b}_{11} \boldsymbol{b}_{22} \Big) \\
&& - \big| \imag \, \boldsymbol{a}_{12}  (\gamma^{-1} x - \gamma t) - c_{12} \big|^2 -|\boldsymbol{a}_{12}|^2 (\gamma^{-1} x - \gamma \, t) - \imag \,  \boldsymbol{a}_{12}^\ast \, c_{12} \, .
\eez
It is easy to find conditions to be imposed on the parameters, such that the real and the imaginary part of $\det(\Omega)$ do not 
have a common zero, so that the solution is non-singular.  
Besides the exponential dependence on $\gamma^{-1} x + \gamma t$ (and its complex conjugate), the solution also depends rationally on 
$\gamma^{-1} x - \gamma t$ (and its complex conjugate). 

Fig.~\ref{fig:scalarFL_anti-conjugate_gammas} shows plots of the absolute values for 
two solutions of the scalar FL equation ($m_1=m_2=1$) from this class.
\begin{figure}[h]
	\begin{center}
		\includegraphics[scale=.4]{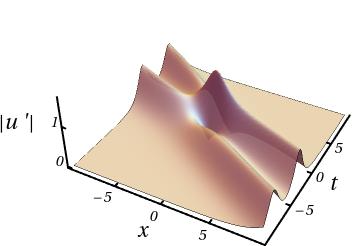} 
		\hspace{1cm}
		\includegraphics[scale=.4]{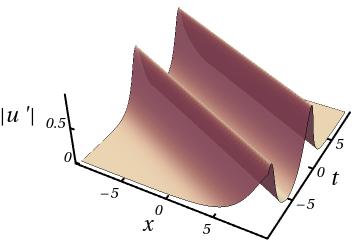} 
		\parbox{15cm}{
		\caption{Plots of the absolute value of solutions of the scalar FL 
		equation (with $\sigma_1=\sigma_2=1$) obtained with an anti-conjugate pair of spectral parameters from the family in Example~\ref{ex:zero_seed_n=2_anti-conjugate_gammas}. We set $\gamma = 1+\imag$,  $a_1 = (1,1)^T$, $a_2 = (\imag \, \gamma,1)^T$ and $c_{12}=1$ (left plot), respectively $c_{12}=50$ (right plot).  				\label{fig:scalarFL_anti-conjugate_gammas} } 
		}
	\end{center}
\end{figure} 
\hfill $\Box$
\end{example}

\begin{example}
Let $n=2$ and $\Gamma$ be the $2 \times 2$ Jordan block 
\bez
	\Gamma = \left( \begin{array}{cc} \gamma & 0 \\ 1 & \gamma 
	\end{array} \right) \, ,
\eez 	
with eigenvalue $\gamma$. Then we have
\bez
	\eta_1 = e^{ -\varphi} \, \left( \begin{array}{cc} 1 & 0 \\ \rho  & 1
	\end{array} \right) \, a_1 \, , \qquad 
	\tilde{\eta}_2 = e^{\varphi} \, \left( \begin{array}{cc} 1 & 0 \\ -\rho  & 1
	\end{array} \right) \, a_2 \, ,
\eez
where $\varphi$ stands for $\varphi(\gamma)$, and
\bez
	\rho = \frac{1}{2} \big( \gamma^{-2} x - t \big) \, ,
\eez 
which introduces the same dependence on the independent variables that we met in the preceding example. 
After having determined 
\bez
    \Omega = \left( \begin{array}{cc} \Omega_{11} & \Omega_{12} \\
    	\Omega_{21} & \Omega_{22} \end{array} \right) \, , \qquad
   \Omega^{-1} = \frac{1}{\mathrm{det}(\Omega)}
   \left( \begin{array}{cc} \Omega_{22} & -\Omega_{12} \\ -\Omega_{21} & \Omega_{11} \end{array} \right) \, , 	
\eez
from (\ref{mFL_zero_seed_u'}) we obtain the solution
\be
u' = -\frac{e^{-2 \, \imag \,
	\mathrm{Im}(\varphi)}}{\det(\Omega)} \, \sigma_2 \, a_2^\dagger \left( \begin{array}{cc} \Omega_{22}-\rho \, \Omega_{12} + \rho^\ast \, \Omega_{21} -|\rho|^2 \, \Omega_{11} & - \Omega_{12} - \rho^\ast \Omega_{11} \\ -\Omega_{21} + \rho \, \Omega_{11} & \Omega_{11} \end{array} \right) \, a_1 \, .
    \label{u'_ts_Jordan}
\ee

\paragraph{(1)}  
If $\kappa := 2 \mathrm{Re}(\gamma) \neq 0$, the solution of the Lyapunov equation is found from Example~\ref{ex:n=2Jordan_Lyapunov_sol} in Appendix~\ref{app:Lyapunov} to be 
\bez
  \Omega = \imag \, \tilde{\Omega}_1 \Gamma^\dagger + \Omega_2 
\eez
with  
\bez
\tilde{\Omega}_1 &=& \frac{1}{\kappa}
\left( \begin{array}{cc} A_{11} & A_{12} + \rho_-^\ast \, A_{11}  \\
	A_{12}^\ast + \rho_- \, A_{11} & (|\rho_-|^2 + \kappa^{-2}) \, A_{11} + \rho_- \, A_{12} + \rho_-^\ast A_{12}^\ast + A_{22} \end{array} \right) \, , \\
\Omega_2 &=& \frac{1}{\kappa}
\left( \begin{array}{cc} B_{11} & B_{12} - \rho_+^\ast \, B_{11}  \\
	B_{12}^\ast - \rho_+ \, B_{11} & (|\rho_+|^2 + \kappa^{-2}) \, B_{11} - \rho_+ \, B_{12} - \rho_+^\ast B_{12}^\ast + B_{22} \end{array} \right) \,  ,
\eez 
where  
\bez
A_{ij} = \boldsymbol{a}_{ij} \, e^{-2 \, \mathrm{Re}(\varphi)} \, , \qquad 
B_{ij} = \boldsymbol{b}_{ij} \, e^{2 \, \mathrm{Re}(\varphi)} \, , \qquad 
\rho_\pm = \rho \pm \kappa^{-1} \, .
\eez 
Hence, $\Omega$ has the components
\bez
\Omega_{11} &=& \kappa^{-1}  ( \imag \, \gamma^\ast \, A_{11}
+ B_{11}  ) \, , \\
\Omega_{12} &=& \kappa^{-1} \Big( \imag \, \big( A_{11} + \gamma^\ast (A_{12} + \rho_-^\ast A_{11}) \big)
+ B_{12} - \rho_+^\ast B_{11} \Big) \, , \\
\Omega_{21} &=&\kappa^{-1} \big( \imag \, \gamma^\ast (A_{12}^\ast + \rho_- \, A_{11}) +B_{12}^\ast - \rho_+ \, B_{11} \big) \, , \\
\Omega_{22} &=& \kappa^{-1} \Big( \imag \, (A_{12}^\ast + \rho_- \, A_{11} ) + \imag \, \gamma^\ast \big( (|\rho_-|^2 + \kappa^{-2}) \, A_{11} + \rho_- \, A_{12} + \rho_-^\ast A_{12}^\ast + A_{22} \big) \\
&& + ( |\rho_+|^2 + \kappa^{-2}) \, B_{11} - \rho_+ \, B_{12} - \rho_+^\ast B_{12}^\ast + B_{22} \Big) \, .       
\eez
Besides the exponential dependence via $A_{ij}$ and $B_{ij}$, the
solution (\ref{u'_ts_Jordan}) also depends on the independent variables $x$ and $t$ through the linear expression $\rho$.
Conditions can be imposed on the parameters, such that the real and the imaginary part of $\det(\Omega)$ do not 
have a common zero, so that $u'$ is non-singular. We obtain matrix generalizations of ``double pole solutions", see \cite{LWZ22,MH+Ye25} for the scalar case.

\paragraph{(2)} If $\gamma$ is imaginary, i.e., $\gamma=-\imag \, k$ with $k \in \mathbb{R} \setminus \{0\}$, then we have 
$\rho = -\frac{1}{2} (x/k^2+t) = \rho^\ast$ and the Lyapunov equation requires (cf. Example~\ref{ex:Lyap_n=2_Jordan_imag_gamma} in Appendix~\ref{app:Lyapunov})
\be
&& k \, \boldsymbol{a}_{11} - \boldsymbol{b}_{11} = 0 \, , \label{ex4.7(2)_constr1} \\
&& \Omega_{11} = -\rho\,( k\, \boldsymbol{a}_{11} + \boldsymbol{b}_{11}) +\imag \, \boldsymbol{a}_{11} - k \, \boldsymbol{a}_{12} +\boldsymbol{b}_{12} 
 = - \rho \, ( k \, \boldsymbol{a}_{11} + \boldsymbol{b}_{11}) - k \, \boldsymbol{a}_{12}^\ast + \boldsymbol{b}_{12}^\ast \, , \nonumber \\
&& \Omega_{12} + \Omega_{21} = - \rho^2( k \, \boldsymbol{a}_{11} - \boldsymbol{b}_{11}) + \rho \, (\imag \, \boldsymbol{a}_{11} - k \,  \boldsymbol{a}_{12} - \boldsymbol{b}_{12}) - \rho\,( k \, \boldsymbol{a}_{12}^\ast + \boldsymbol{b}_{12}^\ast) \nonumber \\
&&\hspace{2.2cm} + \imag \, \boldsymbol{a}_{12}^\ast - k \, \boldsymbol{a}_{22} + \boldsymbol{b}_{22} \, .   \nonumber
\ee
Using the first equation to eliminate $\boldsymbol{b}_{11}$ from the others, leads to
\bez
&& \Omega_{11} = - 2 \, k \, \rho \, \boldsymbol{a}_{11} +\imag \, \boldsymbol{a}_{11} - k \, \boldsymbol{a}_{12} + \boldsymbol{b}_{12}
 = - 2 \, k \, \rho \, \boldsymbol{a}_{11} - k \, \boldsymbol{a}_{12}^\ast +\boldsymbol{b}_{12}^\ast \, , \\
&& \Omega_{12} + \Omega_{21} = \rho \, (\imag \, \boldsymbol{a}_{11} - k \, \boldsymbol{a}_{12} - \boldsymbol{b}_{12}) - \rho \, ( k \, \boldsymbol{a}_{12}^\ast + \boldsymbol{b}_{12}^\ast) + \imag \, \boldsymbol{a}_{12}^\ast - k \, \boldsymbol{a}_{22} + \boldsymbol{b}_{22} \, .
\eez
Consistency of the two expressions for $\Omega_{11}$ requires
\bez
 \boldsymbol{a}_{11} = 2 \, \mathrm{Im}(k \, \boldsymbol{a}_{12} - \boldsymbol{b}_{12}) \, ,
\eez
which is
\be
    \mathrm{Im}(a_{21} \, a_{22}^\ast) 
 - \sigma_1 \sigma_2 \, \Big( \frac{1}{2} |a_{11}|^2 - k \, \mathrm{Im}(a_{11} \, a_{12}^\ast) \Big) = 0 \, .  \label{ts_Jordan_constr2}
\ee 
The remaining components of $\Omega$ are determined by
\bez
\Omega_{12x} &=& \frac{1}{2} k^{-2} \big( \imag \, \boldsymbol{a}_{11} + k \, \boldsymbol{a}_{12} + \boldsymbol{b}_{12} \big) \, ,\qquad
\Omega_{12t} = - \frac{1}{2} \big( \imag \, \boldsymbol{a}_{11} - k \, \boldsymbol{a}_{12} - \boldsymbol{b}_{12} \big) \, , \\
\Omega_{21x} &=& - \frac{1}{2} k^{-2} \big( 2 \imag \, \boldsymbol{a}_{11} - k \, \boldsymbol{a}_{12}^\ast - \boldsymbol{b}_{12}^\ast \big) \, , \qquad 
\Omega_{21t} = \frac{1}{2} ( k \, \boldsymbol{a}_{12}^\ast + \boldsymbol{b}_{12}^\ast ) \, , \\
\Omega_{22x} &=& \frac{1}{2} k^{-3} \Big( 2 k^2 \rho^2 \, \boldsymbol{a}_{11} - k \, \rho \, \big( - k \, (\boldsymbol{a}_{12} + \boldsymbol{a}_{12}^\ast) + \imag \, \boldsymbol{a}_{11} + \boldsymbol{b}_{12} + \boldsymbol{b}_{12}^\ast \big) + k^2 \boldsymbol{a}_{22} - k \, (\imag \, \boldsymbol{a}_{12} - \boldsymbol{b}_{22}) \\
&& + \boldsymbol{a}_{11} - \imag \, (\boldsymbol{b}_{12} - \boldsymbol{b}_{12}^\ast) \Big) \, , \\
\Omega_{22t} &=& \frac{1}{2} k^{-1} \Big( 2 k^2 \rho^2 \, \boldsymbol{a}_{11} - k \, \rho \, \big( - k \, (\boldsymbol{a}_{12} + \boldsymbol{a}_{12}^\ast ) + \imag \, \boldsymbol{a}_{11} + \boldsymbol{b}_{12} + \boldsymbol{b}_{12}^\ast \big) + k^2 \boldsymbol{a}_{22} - k \, (\imag \, \boldsymbol{a}_{12}^\ast - \boldsymbol{b}_{22}) \Big) \, .
\eez
This system is solved by
\bez
\Omega_{12} 
&=& - \rho \, (k \, \boldsymbol{a}_{12} + \boldsymbol{b}_{12} + \imag \, \boldsymbol{a}_{11}) - \imag \, \boldsymbol{a}_{11} \, t + c_{12} \, , \\
\Omega_{21} 
&=& - \rho \, (k \, \boldsymbol{a}_{12}^\ast + \boldsymbol{b}_{12}^\ast - 2 \imag \, \boldsymbol{a}_{11}) + \imag \, \boldsymbol{a}_{11} \, t + c_{21}\, , \\
\Omega_{22} &=& - \frac{2}{3} k \, \rho^3 \boldsymbol{a}_{11} + \frac{1}{2} \rho^2 \big( -k \, (\boldsymbol{a}_{12} + \boldsymbol{a}_{12}^\ast) + \imag \, \boldsymbol{a}_{11} + \boldsymbol{b}_{12} + \boldsymbol{b}_{12}^\ast \big) - \rho \, ( k \, \boldsymbol{a}_{22} + \boldsymbol{b}_{22} - 2 \imag \,  \boldsymbol{a}_{12} + \imag \, \boldsymbol{a}_{12}^\ast )  \\ &&
+\imag \, (\boldsymbol{a}_{12} - \boldsymbol{a}_{12}^\ast) \, t + c_{22} \, ,
\eez
with constants $c_{ij}$, for which (\ref{mFLeq_key}) and the above equation for $\Omega_{12}+\Omega_{21}$ require
\bez
&& \mathrm{Re}(c_{12}) = \frac{1}{2}(- k \, \boldsymbol{a}_{22} +  
   \boldsymbol{b}_{22}) \, , \qquad 
 \mathrm{Re}(c_{21}) = \frac{1}{2} \big(-k \, \boldsymbol{a}_{22} + 2 \, \mathrm{Im}(\boldsymbol{a}_{12}) + \boldsymbol{b}_{22} \big)\, , \\
&& \mathrm{Im}(c_{21}) = \mathrm{Re}(\boldsymbol{a}_{12}) - \mathrm{Im}(c_{12}) \, , \qquad \mathrm{Im}(c_{22}) = \frac{1}{2} \boldsymbol{a}_{22} \, ,
\eez
so that $\mathrm{Re}(c_{22})$ and $\mathrm{Im}(c_{12})$ remain as 
free parameters. Let $k$ be determined by (\ref{ex4.7(2)_constr1}). The constants $a_{ij}$, $i,j=1,2$, are still constrained by (\ref{ts_Jordan_constr2}). Then (\ref{u'_ts_Jordan})
is a quasi-rational solution of the $m_2\times m_1$ matrix FL equation.  
If $\boldsymbol{a}_{11} \neq 0$, for constant $t$ (respectively $x$) the numerators of components of $u'$ grow at most with $x^3$ (respectively $t^3$), whereas $\det(\Omega)$ grows with the fourth power. On a  
line, along which $\rho$ and thus $x + k^2 t$ is constant, $|u'| \sim 1/|t|$ as $t \to \pm \infty$. As a consequence, the solution must be localized along a worldline, which is not a straight line.  
In particular, our results supplement the treatment of the scalar FL equation ($m_1=m_2=1$) in Section~4.1 of \cite{MH+Ye25}, also see  Fig.~\ref{fig:scalarFL_2x2Jordan_imag_gamma} for two examples. 
\begin{figure}[h]
	\begin{center}
		\includegraphics[scale=.4]{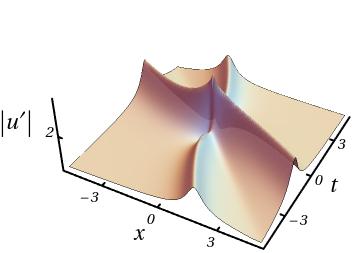} 
		\hspace{1cm}
		\includegraphics[scale=.4]{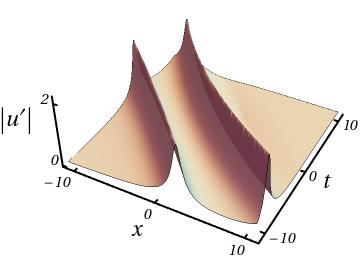} 
		\parbox{15cm}{
	\caption{Plots of the absolute value of solutions of the scalar FL 
	equation (with $\sigma_1=\sigma_2=1$) obtained with a $2\times 2$	 Jordan block $\Gamma$ with eigenvalue $\gamma = -\imag$, and $a_1 = (1, -\frac{1}{2} \imag)^T$, $a_2 = (1,0)^T$,  $\mathrm{Im}(c_{12})=0$, and $\mathrm{Re}(c_{22})=0$, respectively $=50$. Also see Fig.~\ref{fig:scalarFL_2x2Jordan_imag_gamma_2} for details on the first example.
	\label{fig:scalarFL_2x2Jordan_imag_gamma} } 
		}
	\end{center}
\end{figure} 
\begin{figure}[h]
	\begin{center}
		\includegraphics[scale=.25]{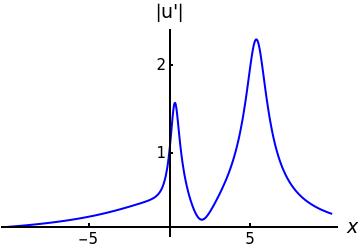} 
		\includegraphics[scale=.25]{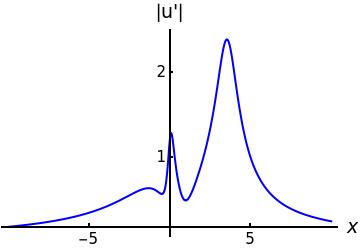} 
		\includegraphics[scale=.25]{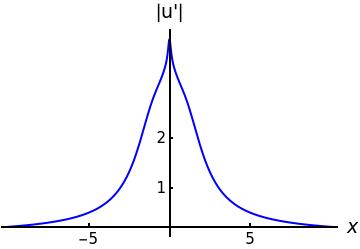} 
		\includegraphics[scale=.25]{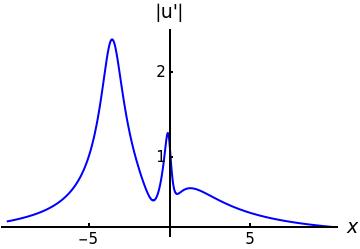} 
		\includegraphics[scale=.25]{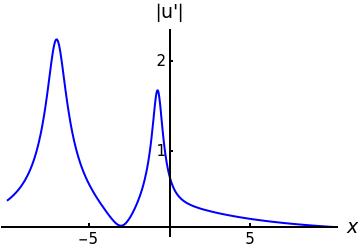}   \\
		\includegraphics[scale=.45]{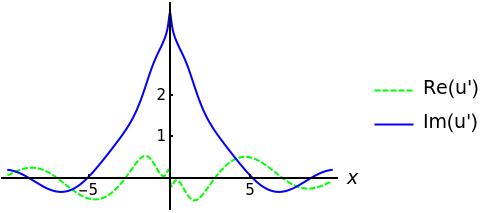}
		\parbox{15cm}{
		\caption{Plots of the absolute value of the first solution of the scalar FL equation displayed in Fig.~\ref{fig:scalarFL_2x2Jordan_imag_gamma}, at $t=-2, -1,0, 1, 3$. Below is a plot of $\mathrm{Re}(u')$ and $\mathrm{Im}(u')$ at $t=0$. The apparent cusp of $\mathrm{Im}(u')$ at $t=0$ is just a sharp maximum.
			\label{fig:scalarFL_2x2Jordan_imag_gamma_2} } 
		}
	\end{center}
\end{figure} 
\hfill$\Box$
\end{example}

The preceding examples exhaust the possibilities of solutions up to $n=2$. Mostly, they also present corresponding extensions to arbitrary $n$, but in the case, where $\Gamma$ is a Jordan block, we restricted our computations to $n=2$. The step to higher $n$ is straightforward, but the resulting expressions lengthy. According to Remark~\ref{rem:FL_superposition}, we can superpose any number of solutions from the set obtained above.

\section{The vector FL equation}
\label{sec:vFL}

Choosing $m_1 = 1$, $u = (u_1,\ldots,u_{m_2})^T$ is an $m_2$-component column vector. Without restriction of generality, we may set 
$\sigma_1 = 1$. Renaming $\sigma_2$ to $\boldsymbol{\sigma}$, (\ref{mFLeq}) reads
\be
    u_{xt} = (1 + \imag \, u^\dagger \boldsymbol{\sigma} \, u_x) \, u + \imag \, (u^\dagger \, \boldsymbol{\sigma} \, u) \, u_x \, ,  \label{vFLeq}
\ee
which first appeared in \cite{Guo+Ling12}, also see \cite{ZYCW17,LFZ18,Mats18,Mats19,Yang+Zhang18,KXM18,Wang+Chen19,YZCBG19,Gerd+Ivan21,Ling+Su24,SFY25}.

Setting all components but one to zero, (\ref{vFLeq}) reduces to the scalar FL equation for the remaining component. 
If all components of $u$ are equal, then $u_\mu/\sqrt{m_2}$ satisfies 
the scalar FL equation. 

Now we specialize Theorem~\ref{thm:mFL_DT} to the vector FL case. In this case, the linear system (\ref{mFLeq_linsys}) can be decoupled and determines $\eta_1$ in terms of $\tilde{\eta}_2$.

\begin{theorem}
\label{thm:vFL_DT}
Let $\Gamma$ be an invertible constant $n \times n$ matrix, such that $\Gamma$ and $-\Gamma^\dagger$ have no common eigenvalue (spectrum condition). Let $u$ be a solution of the vector FL equation (\ref{vFLeq}) on an open set $\mathcal{U} \subset \mathbb{R}^2$, where $u^\dagger \boldsymbol{\sigma} u \neq 0$ and $u_x^\dagger \boldsymbol{\sigma} u \neq 0$. 
Furthermore, let $\tilde{\eta}$ be an $n \times m_2$ matrix solution of the linear system
\be
  && \hspace*{-.8cm}  \Big( \tilde{\eta}_t   
  - \frac{1}{2} \Gamma \tilde{\eta}    
  + \frac{u^\dagger \boldsymbol{\sigma} u}{u_x^\dagger \boldsymbol{\sigma} u} \big( \Gamma \tilde{\eta}_x 
    - \frac{1}{2} \tilde{\eta} \big) 
  + \imag \, u^\dagger \boldsymbol{\sigma} u \, \tilde{\eta} \Big) \, u  = 0 
  \, , \label{vFLeq_linsys_1} \\
  && \hspace*{-.8cm} \tilde{\eta}_{xx} \, u + \tilde{\eta}_x \Big( u_x - \frac{(u_x^\dagger \boldsymbol{\sigma} u)_x}{u_x^\dagger \boldsymbol{\sigma} u} \, u \Big) 
  + \frac{1}{2} \Big( \frac{(u_x^\dagger \boldsymbol{\sigma} u)_x}{u_x^\dagger \boldsymbol{\sigma} u} I_n
  - \frac{1}{2} \Gamma^{-1} \Big) \Gamma^{-1} \tilde{\eta} \, u 
  + ( \imag \, u_x^\dagger \boldsymbol{\sigma} u -  \frac{1}{2} ) \, \Gamma^{-1} \tilde{\eta} \, u_x = 0 
    \, ,  \qquad  \label{vFLeq_linsys_2}  
\ee
and
\be
( \tilde{\eta}_x - \frac{1}{2} \Gamma^{-1}  \tilde{\eta} ) \, P = 0 \, , \qquad
( \tilde{\eta}_t - \frac{1}{2} \Gamma \, \tilde{\eta} ) \, Q = 0 
  \, ,  \label{vFLeq_linsys_3}
\ee
where 
\bez
 P =  I_{m_2} - \frac{u \, u^\dagger_x \boldsymbol{\sigma}}{u_x^\dagger \boldsymbol{\sigma} u} 
     \, , \qquad  
 Q = I_{m_2} - \frac{u \, u^\dagger \boldsymbol{\sigma}}{u^\dagger \boldsymbol{\sigma} u}  \, ,
\eez
are projections (both with kernel spanned by $u$). 
Furthermore, let $\Omega$ be an $n \times n$ matrix solution of the Lyapunov equation 
\be
	\Gamma \Omega + \Omega \Gamma^\dagger 
	= \frac{\imag}{|u_x^\dagger \boldsymbol{\sigma} u|^2} \big( \tilde{\eta}_x 
	- \frac{1}{2} \Gamma^{-1} \tilde{\eta} \big) \, u \, \big[  
	\big( \tilde{\eta}_x 
	- \frac{1}{2} \Gamma^{-1} \tilde{\eta} \big) u \big]^\dagger
	\, \Gamma^\dagger 
	+ \tilde{\eta} \, \boldsymbol{\sigma} \, \tilde{\eta}^\dagger \, . \label{vFLeq_Lyap}
\ee
Then, in any open subset of $\mathcal{U}$ where $\Omega$ is invertible,
\be
  u' = u - \frac{\imag}{u_x^\dagger \boldsymbol{\sigma} u} \, \boldsymbol{\sigma} \,  \tilde{\eta}^\dagger \, \Omega^{-1}  \big( \tilde{\eta}_x 
  - \frac{1}{2} \Gamma^{-1} \tilde{\eta} \big) \, u 
  \label{vFLeq_u'}
\ee
solves the vector FL equation (\ref{vFLeq}).
\end{theorem}
\begin{proof}
Writing $\tilde{\eta}$ instead of $\tilde{\eta}_2$, the linear system (\ref{mFLeq_linsys}) reads
\bez
&&  \Gamma \, \eta_{1x} + \frac{1}{2} \eta_1 - \tilde{\eta} \, u_x = 0 \, , \qquad\quad
\eta_{1t} + \big( \frac{1}{2} \Gamma  - \imag \,  u^\dagger \boldsymbol{\sigma} u \,  I_n \big) \, \eta_1 - \tilde{\eta} \, u = 0
\, ,   \nonumber \\
&& \Gamma \tilde{\eta}_x - \frac{1}{2} \tilde{\eta} + \imag \, \Gamma \eta_1 \,  u^\dagger_x \boldsymbol{\sigma} = 0 
       \, , \qquad
 \tilde{\eta}_t - \frac{1}{2} \Gamma  \tilde{\eta} + \imag \, \tilde{\eta} \,  u u^\dagger \boldsymbol{\sigma} 
 - \imag \, \Gamma \eta_1 \, u^\dagger \boldsymbol{\sigma} = 0 \, .  
\eez
Multiplying the last two equations by $u$ from the right, and solving them for $\eta_1$, we get
\bez
    \frac{\imag}{u_x^\dagger \boldsymbol{\sigma} u} \big( \tilde{\eta}_x 
     - \frac{1}{2} \Gamma^{-1} \tilde{\eta} \big) \, u 
   =  \eta_1 
   = - \frac{\imag}{u^\dagger \boldsymbol{\sigma} u} \big(\Gamma^{-1} \tilde{\eta}_t   
     - \frac{1}{2} \tilde{\eta} \big) u + \Gamma^{-1} \tilde{\eta} \, u  \, .
\eez
Inserting this back in the last two equations of the linear system, 
they are turned into (\ref{vFLeq_linsys_3}).
Using the expressions for $\eta_1$ in the first two equations of the linear system, yields (\ref{vFLeq_linsys_2}) and 
\bez
  && \Big( \big( \tilde{\eta}_t - \frac{1}{2} \Gamma \,  \tilde{\eta} \big) \, u \Big)_t + \Big( \frac{1}{2} \Gamma -\big( \imag \, u^\dagger \boldsymbol{\sigma} u + \frac{(u^\dagger \boldsymbol{\sigma} u)_t}{u^\dagger \boldsymbol{\sigma} u} \big) I_n \Big) \big( \tilde{\eta}_t - \frac{1}{2} \Gamma  \tilde{\eta} \big) \, u 
+ \imag \, u^\dagger \boldsymbol{\sigma} u \, ( \tilde{\eta} u )_t \\
  && - u^\dagger \boldsymbol{\sigma} u \, \big( \frac{1}{2} \imag \, \Gamma - u^\dagger \boldsymbol{\sigma} u \,  I_n \big) \, \tilde{\eta} \, u = 0 \, .
\eez
The two expressions for $\eta_1$ imply (\ref{vFLeq_linsys_1}), and the last equation can be shown to be a consequence of  (\ref{vFLeq_linsys_1}) and  (\ref{vFLeq_linsys_2}).
(\ref{vFLeq_Lyap}) and (\ref{vFLeq_u'}) result from (\ref{mFLeq_Lyap}) and (\ref{mFLeq_u'}), respectively, by using the first expression for $\eta_1$ above. 
\end{proof}

\begin{corollary}
	\label{cor:vFL_DT_nosc}
The spectrum condition for $\Gamma$ can be dropped in  Theorem~\ref{thm:vFL_DT} if the assumptions there are supplemented by 
\be
	\Gamma \Omega + \Omega^\dagger \Gamma^\dagger = \tilde{\eta} \, \boldsymbol{\sigma} \, \tilde{\eta}^\dagger \, ,  
	\label{vFLeq_key}
\ee	
and
\be
	&& \Omega_x = \frac{\imag}{u^\dagger \boldsymbol{\sigma} u_x} \Big( \frac{1}{u_x^\dagger \boldsymbol{\sigma}u} 
	\Xi \Big)_x 
	\, \Xi^\dagger + \Big( \tilde{\eta}_x \, \boldsymbol{\sigma} - \frac{1}{u_x^\dagger \boldsymbol{\sigma} u} \Xi \, u_x^\dagger \Big) \, \tilde{\eta}^\dagger \, \Gamma^{\dagger -1} \, , \nonumber \\
	&& \Omega_t = - \frac{\imag}{2 |u_x^\dagger \boldsymbol{\sigma} u|^2} \Xi \, \Xi^\dagger \, \Gamma^\dagger 
	+ \big( \frac{1}{2} \tilde{\eta} \, \boldsymbol{\sigma} - \frac{1}{u_x^\dagger \boldsymbol{\sigma} u} \Xi \, u^\dagger  \big) \, \tilde{\eta}^\dagger 
	\, ,	\label{vFLeq_Om_deriv}
\ee	
where we introduced
\be
	\Xi = (\tilde{\eta}_x - \frac{1}{2} \Gamma^{-1} \tilde{\eta} ) \, u \, . 
	\label{Xi}
	\ee
\end{corollary}
\begin{proof}
This follows directly from Corollary~\ref{cor:mFL_no_sc}, using intermediate results in the proof of Theorem~\ref{thm:mFL_DT}.
\end{proof}

\begin{remark}
Using a familiar formula for the determinant of a $2 \times 2$ block matrix, the new solution (\ref{vFLeq_u'}) can be expressed as follows in terms of a quotient of determinants,
\bez 
  u'_\mu &=& u_\mu - \frac{\imag}{(u_x^\dagger \boldsymbol{\sigma} u) \, \det( \Omega)} \, \det \left( \begin{array}{cc}
  	\Omega & \big( \tilde{\eta}_x 
  	- \frac{1}{2} \Gamma^{-1} \tilde{\eta} \big) \, u \\
  	  \boldsymbol{\sigma} \,  \tilde{\eta}_\mu^\dagger & 0 \end{array} \right) \\
  &=& \frac{1}{\det(\Omega)} \, \det \left( \begin{array}{cc}
  	   \Omega & \big( \tilde{\eta}_x 
  	   - \frac{1}{2} \Gamma^{-1} \tilde{\eta} \big) \, u \\
  	   \imag \, (u_x^\dagger \boldsymbol{\sigma} u)^{-1} \boldsymbol{\sigma} \,  \tilde{\eta}_\mu^\dagger & u_\mu \end{array} \right) \qquad \quad
  	   \mu=1,\ldots,m_2 \, .
\eez
In the cited references, where solutions have been generated via the classical (iterative) Darboux transformation method, they are mostly 
presented in such a form, with a specialized matrix $\Omega$.
\hfill $\Box$  	   
\end{remark}

\begin{remark}
For $m_2=1$, the equations (\ref{vFLeq_linsys_3}) are identically satisfied. For $m_2 >1$ and if $\boldsymbol{\sigma}$ is positive (or negative) definite, 
at each point of $\mathbb{R}^2$, where $u$ is defined and $u \neq 0$, the vector $u$ determines a splitting of the $m_2$-dimensional complex vector space, 
$\mathbb{C}^{m_2} = \mathrm{span}(u) \oplus V$, where $V$ is the orthogonal complement of the linear span of $u$ with respect to the inner product $\langle w | w' \rangle = w^\dagger \boldsymbol{\sigma} w'$. Then the second of (\ref{vFLeq_linsys_3}) implies that $(\tilde{\eta}_t - \frac{1}{2} \Gamma \tilde{\eta}) \, w = 0$ for all $w \in V$. Whereas $Q$ is an orthogonal projection, $P$ is not. We have
\bez
       P(v) = - \frac{\langle u_x | v \rangle}{\langle u_x | u \rangle} \, u + v  \qquad \forall v \perp u \, .
\eez
The first equation of (\ref{vFLeq_linsys_3}) says that $\tilde{\eta}_x - \frac{1}{2} \Gamma^{-1} \tilde{\eta}=0$ holds on the orthogonal complement of  $\mathrm{span}(u_x)$. 
In particular, we have $\tilde{\eta}_x - \frac{1}{2} \Gamma^{-1} \tilde{\eta} = 0$ and $\tilde{\eta}_t - \frac{1}{2} \Gamma \tilde{\eta} = 0$ on  the orthogonal complement of $\mathrm{span}(u,u_x)$.   \hfill $\Box$
\end{remark}

\begin{proposition}
	\label{prop:vFL_unitary_symmetry}
Let $\boldsymbol{U}_0$ be any constant $\boldsymbol{\sigma}$-unitary matrix (i.e., a 
constant $m_2 \times m_2$ matrix satisfying $\boldsymbol{U}_0^\dagger \, \boldsymbol{\sigma} \, \boldsymbol{U}_0 = \boldsymbol{\sigma}$). 
Then 
$u \mapsto \boldsymbol{U}_0 \, u$ is a symmetry of the vector FL 
equation (\ref{vFLeq}) and also a symmetry of the Darboux transformation in Theorem~\ref{thm:vFL_DT} and Corollary~\ref{cor:vFL_DT_nosc}. 	
\end{proposition}
\begin{proof}
It is easily verified that, if $u$ is a solution of (\ref{vFLeq}), 
then also $\boldsymbol{U}_0 \, u$. 
Furthermore, we note that 
\bez
    u \mapsto \boldsymbol{U}_0 \, u \, , \qquad
    \tilde{\eta} \mapsto \tilde{\eta} \, \boldsymbol{U}_0^{-1} \, ,
\eez
is a symmetry of the linear system (\ref{vFLeq_linsys_1})-(\ref{vFLeq_linsys_3}), 
since $P \mapsto \boldsymbol{U}_0 \, P \, \boldsymbol{U}_0^{-1}$ 
and $Q \mapsto \boldsymbol{U}_0 \, Q \, \boldsymbol{U}_0^{-1}$. 
$\Omega$ is invariant, hence it follows that $u' \mapsto \boldsymbol{U}_0 \, u'$. 
We used that $\boldsymbol{\sigma}$ is Hermitian and involutary. 
\end{proof}

\begin{remark}
	\label{rem:vFL_superposition}
We specialize the important ``superposition property", formulated in Remark~\ref{rem:FL_superposition}, to the case of the vector FL equation. 
Let $(\Gamma^{(i)},\tilde{\eta}^{(i)},\Omega^{(i)})$,
$i=1,2$, be data determining solutions ${u'}^{(i)}$ of the vector FL equation via the vectorial Darboux transformation, using the same seed solution $u$. Let
\bez
	\Gamma = \left( \begin{array}{cc} \Gamma^{(1)} & 0 \\
		0 & \Gamma^{(2)} \end{array} \right) \, , \quad
	\tilde{\eta} = \left(\begin{array}{c} \tilde{\eta}^{(1)} \\ \tilde{\eta}^{(2)}
	\end{array} \right) \, , \quad	
	\Omega = \left( \begin{array}{cc} \Omega^{(1)} & \Omega^{(12)} \\
		\Omega^{(21)} & \Omega^{(2)} \end{array} \right) \, .
\eez
The linear system (\ref{vFLeq_linsys_1})-(\ref{vFLeq_linsys_3}) is then solved, and  it only remains to determine $\Omega^{(12)}$ and $\Omega^{(21)}$ in order to get a new solution $u'$, which then represents a nonlinear superposition of ${u'}^{(1)}$ and ${u'}^{(2)}$. These non-diagonal blocks of $\Omega$ are subject to the Sylvester equations
\bez
 && \Gamma^{(1)} \Omega^{(12)} + \Omega^{(12)} \Gamma^{(2)\dagger}
	= \frac{\imag}{|u_x^\dagger \boldsymbol{\sigma} u|^2} \Xi^{(1)} \, \Xi^{(2) \dagger} \, \Gamma^{(2)\dagger} 
	+ \tilde{\eta}^{(1)} \, \boldsymbol{\sigma} \, \tilde{\eta}^{(2)\dagger} \, , \\
 && \Gamma^{(2)} \Omega^{(21)} + \Omega^{(21)} \Gamma^{(1)\dagger}
	= \frac{\imag}{|u_x^\dagger \boldsymbol{\sigma} u|^2} \Xi^{(2)} \, \Xi^{(1) \dagger} \, \Gamma^{(1)\dagger} 
	+ \tilde{\eta}^{(2)} \, \boldsymbol{\sigma} \, \tilde{\eta}^{(1)\dagger}
	\, ,
\eez	
where 
\bez
    \Xi^{(i)} = \Big( \tilde{\eta}_x^{(i)} 
    - \frac{1}{2} \Gamma^{(i)-1} \tilde{\eta}^{(i)} \Big) \, u \, , \qquad i=1,2 \, .
\eez
If $\Gamma^{(1)}$ and $-\Gamma^{(2)\dagger}$ have no common eigenvalue, these equations possess unique solutions. Otherwise additional equations, obtained from those in Corollary~\ref{cor:vFL_DT_nosc}, have to be solved.

In the simplest case of a superposition of two $n=1$ solutions, setting $\Gamma^{(1)} = \gamma_1$ and $\Gamma^{(2)} = \gamma_2$, 
and assuming $\gamma_1 \neq - \gamma_2^\ast$, we have
\bez
  &&  \Omega^{(12)} = \frac{1}{\gamma_1 + \gamma_2^\ast} \Big( 
  \frac{\imag \, \gamma_2^\ast}{|u_x^\dagger \boldsymbol{\sigma} u|^2} \, \Xi^{(1)} \, \Xi^{(2)\ast}
    + \tilde{\eta}^{(1)} \, \boldsymbol{\sigma} \, \tilde{\eta}^{(2)\dagger} \Big) \, , \\
  &&  \Omega^{(21)} = \frac{1}{\gamma_2 + \gamma_1^\ast} \Big( \frac{\imag \, \gamma_1^\ast}{|u_x^\dagger \boldsymbol{\sigma} u|^2} \, \Xi^{(2)} \, \Xi^{(1)\ast}
 + \tilde{\eta}^{(2)} \, \boldsymbol{\sigma} \, \tilde{\eta}^{(1)\dagger} \Big) \, ,  
\eez	
where the $\Xi^{(i)}$ are scalars.
\hfill $\Box$
\end{remark}

In the following subsection and the next section we concentrate 
on a plane wave seed solution in the above Darboux transformation for the vector FL equation. Inserting the plane wave ansatz
\bez
  u = e^{\imag \, ( \boldsymbol{\alpha} x + \boldsymbol{\beta} t)} \, A 
\eez
in (\ref{vFLeq}), with a constant column vector $A$ and commuting constant Hermitian matrices $\boldsymbol{\alpha}, \boldsymbol{\beta}$, leads to
\bez
\boldsymbol{\alpha} \, \boldsymbol{\beta} = (u^\dagger \boldsymbol{\sigma} \boldsymbol{\alpha} \, u -1) \, I_{m_2} 
+ (u^\dagger \boldsymbol{\sigma} u) \, \boldsymbol{\alpha} \, .
\eez
Assuming that $\boldsymbol{\alpha}$ is invertible, this yields
\bez
\boldsymbol{\beta} =  (u^\dagger \boldsymbol{\sigma} \boldsymbol{\alpha} \, u -1) \, \boldsymbol{\alpha}^{-1} 
+ (u^\dagger \boldsymbol{\sigma} u) \, I_{m_2} \, ,
\eez
which requires that the scalars $u^\dagger \boldsymbol{\sigma} u$ and $u^\dagger \boldsymbol{\sigma} \boldsymbol{\alpha} \, u$ are constant. This is guaranteed is $\boldsymbol{\alpha}$ commutes with $\boldsymbol{\sigma}$. 
The assumptions $u^\dagger \boldsymbol{\sigma} u \neq 0$ and $u_x^\dagger \boldsymbol{\sigma} u \neq 0$ in Theorem~\ref{thm:vFL_DT} now read
\bez
A^\dagger  \boldsymbol{\sigma} A \neq 0 \, , \qquad 
A^\dagger  \boldsymbol{\alpha} \, \boldsymbol{\sigma} A \neq 0 \, .
\eez
Throughout, they will be assumed in the following.

\subsection{Plane wave seed}
\label{subsec:vFL_pws}
We restrict our considerations to the case where $\boldsymbol{\sigma} = I_{m_2}$ 
and $\boldsymbol{\alpha} = \mathrm{diag}(\alpha_1,\ldots,\alpha_{m_2})$, with 
$\alpha_\mu \in \mathbb{R}$. 
The plane wave solution of the vector FL equation is then given by
\bez
     u_\mu = A_\mu \, e^{\imag \, \varphi_\mu} \, ,  \qquad
     \varphi_\mu = \alpha_\mu x + \beta_\mu t \, , 
      \qquad 
     \mu=1,\ldots, m_2 \, ,
\eez
with $A_\mu \in \mathbb{C}$ and $\alpha_\mu \in \mathbb{R}$. The absolute value of $\alpha_\mu$ is the wave number of the respective mode. Assuming $\alpha_\mu \neq 0$ for all $\mu$, we have $\boldsymbol{\beta} = \mathrm{diag}(\beta_1,\ldots,\beta_{m_2})$ with
\bez
\beta_\mu = \sum_{\nu=1}^{m_2} \Big( 1 + \frac{\alpha_\nu}{\alpha_\mu} \Big) \, |A_\nu|^2  - \frac{1}{\alpha_\mu} 
\qquad \quad \mu=1,\ldots, m_2    \, .
\eez
This solution will be taken as the seed to generate further exact solutions via Theorem~{\ref{thm:vFL_DT}}.
We have
\bez
u_x = \imag \, \boldsymbol{\alpha} \, u \, , \qquad
u_t = \imag \, \boldsymbol{\beta} \, u  \, .
\eez
Let $\boldsymbol{U}$ be a unitary $m_2 \times m_2$ matrix such that 
\bez
 u = e^{\imag \, (\boldsymbol{\alpha} x + \boldsymbol{\beta} t) } \, A
   = \boldsymbol{U} \left( \begin{array}{c}    
     \| A \|  \\ 0 \\ \vdots \\ 0 \end{array} \right) \, , \qquad
 A =  \left( \begin{array}{c} A_1 \\ \vdots \\ A_{m_2} \end{array} \right)            \, ,
\eez
where $\|A\|$ is the standard norm of the complex vector $A$. 
This requires $\boldsymbol{U}_{\mu1} = e^{\imag \, \varphi_\mu} A_\mu/\|A\|$, $\mu=1,\ldots,m_2$.  We set
\bez
      \tilde{\eta} = \chi \, \boldsymbol{U}^\dagger \, , \qquad  \chi = (\chi_1,\ldots,\chi_{m_2}) \, ,
\eez
with $n$-component column vectors $\chi_\mu$, $\mu=1, \ldots, m_2$. Since we assume $A^\dagger \boldsymbol{\alpha}  A \neq 0$,   (\ref{vFLeq_linsys_1}) and the second of equations  (\ref{vFLeq_linsys_3}) become
\be
  && \chi_{1t} - \frac{1}{2} \Gamma \, \chi_1 + \imag \, \| A\|^2 \chi_1 - \imag \, \chi \boldsymbol{U}^\dagger \boldsymbol{\beta} \, \boldsymbol{U} \, (1,0,\ldots,0)^T  \nonumber \\
  && \qquad + \imag \, \frac{\|A\|^2}{ A^\dagger \boldsymbol{\alpha}  A} \big( \Gamma ( \chi_{1x} 
  - \imag \, \chi \boldsymbol{U}^\dagger \boldsymbol{\alpha} \, \boldsymbol{U} \, (1,0,\ldots,0)^T ) - \frac{1}{2} \chi_1 \big) = 0  \, , \nonumber \\
 &&    \chi_{\mu t} =  \frac{1}{2} \Gamma \, \chi_\mu 
    - \sum_{\nu=1}^{m_2} \chi_\nu \, (\boldsymbol{U}_t^\dagger \boldsymbol{U})_{\nu \mu} 
  \qquad \mu=2,\ldots,m_2 \, .  \label{vFL_chi_t}
\ee
Here we used $\boldsymbol{U}_x^\dagger u = - \boldsymbol{U}^\dagger u_x = - \imag \, \boldsymbol{U}^\dagger \boldsymbol{\alpha} \, u$ and $\boldsymbol{U}_t^\dagger u = - \boldsymbol{U}^\dagger u_t = - \imag \, \boldsymbol{U}^\dagger \boldsymbol{\beta} \, u$.
The first of (\ref{vFLeq_linsys_3}) tells us that we have $\tilde{\eta}_x -  \frac{1}{2} \Gamma^{-1} \, \chi = 0$
on the orthogonal complement of $u_x$. A vector (field) $w$ is orthogonal to $u_x$ if and only if $\tilde{w} = \boldsymbol{U}^\dagger w$ is orthogonal to $\boldsymbol{U}^\dagger u_x = \imag \, \boldsymbol{U}^\dagger \boldsymbol{\alpha} \, u = \imag \, \boldsymbol{U}^\dagger \boldsymbol{\alpha} \, \boldsymbol{U} (\|A\|,0,\ldots,0)^T$, i.e., 
\bez
  0 = (1, 0, \ldots, 0) \,  \boldsymbol{U}^\dagger \boldsymbol{\alpha} \, \boldsymbol{U} \, \tilde{w}
     = \sum_{\nu=1}^{m_2} (\boldsymbol{U}^\dagger \boldsymbol{\alpha} \, \boldsymbol{U})_{1\nu} \,  \tilde{w}_\nu \, .
\eez
Then (\ref{vFLeq_linsys_3}) reads
\be
      ( \chi_x + \chi \, \boldsymbol{U}_x^\dagger \boldsymbol{U} -  \frac{1}{2} \Gamma^{-1} \, \chi ) \, \tilde{w} =0 \, ,  \label{chi_x_orthog}
\ee
for all such $\tilde{w}$. The remaining equation  (\ref{vFLeq_linsys_2}) of the linear system becomes
\be
   &&   \chi_{1xx}  - \frac{1}{4} \Gamma^{-2} \chi_1 + \Big(  2 \chi_x \boldsymbol{U}_x^\dagger \boldsymbol{U} + \chi \, \boldsymbol{U}_{xx}^\dagger \boldsymbol{U}  
 + \imag \, ( \chi_x + \chi \, \boldsymbol{U}_x^\dagger \boldsymbol{U}) \, \boldsymbol{U}^\dagger \boldsymbol{\alpha} \, \boldsymbol{U} \nonumber \\
 &&  + \imag \, ( A^\dagger \boldsymbol{\alpha} A - \frac{1}{2} ) \,  \Gamma^{-1} \chi \, \boldsymbol{U}^\dagger \boldsymbol{\alpha} \, \boldsymbol{U}  \Big) \, (1,0,\ldots,0)^T = 0 \, . \label{chi1_xx}
\ee
The Lyapunov equation (\ref{vFLeq_Lyap}) reads
\be
    \Gamma \, \Omega + \Omega \, \Gamma^\dagger 
 = \frac{\imag}{(A^\dagger \boldsymbol{\alpha} A)^2} \, \Xi \, \Xi^\dagger \Gamma^\dagger + \sum_{\mu=1}^{m_2} \chi_\mu \, \chi_\mu^\dagger \, ,   \label{pw_vFLeq_Lyap}
\ee
now with
\bez
      \Xi = \|A\| \, \big( \chi_{1x} - \frac{1}{2} \Gamma^{-1} \chi_1 
   - \imag \, \chi \, \boldsymbol{U}^\dagger \boldsymbol{\alpha} \, \boldsymbol{U}  \, (1,0,\ldots,0)^T \big) \, .
\eez
If $\Gamma$ satisfies the spectrum condition, then (\ref{pw_vFLeq_Lyap}) has a unique solution $\Omega$. If $\Gamma$ does not satisfy the spectrum condition, then the additional equations 
in Corollary~\ref{cor:vFL_DT_nosc} have to be taken into account in order to determine $\Omega$.
(\ref{vFLeq_u'}) takes the form
\bez
  u' &=& u + \frac{1}{A^\dagger \boldsymbol{\alpha} A} \, \boldsymbol{U} \, \chi^\dagger \, \Omega^{-1} \Xi \\
  &=& \boldsymbol{U} \left[ \left( \begin{array}{c}    
  	\|A\| \\ 0 \\ \vdots \\ 0 \end{array} \right)  
  	+ \frac{1 }{A^\dagger \boldsymbol{\alpha} A} \, 
  	\left( \begin{array}{c}   
  	\chi_1^\dagger \, \Omega^{-1} \, \Xi  \\ 
  	\vdots  \\ 
  	\chi_{m_2}^\dagger \, \Omega^{-1} \, \Xi \end{array} \right) \right]   \, .
\eez
In the following section, we will concentrate on the two-component vector FL equation.

\section{Solutions of a two-component vector FL equation}
\label{sec:2vFL}
In this section, we elaborate the set of solutions of the two-component vector FL equation, i.e, (\ref{vFLeq}) with $m_2=2$, generated via the Darboux transformation obtained in the preceding section. 
In order to somewhat reduce the complexity, we restrict our considerations to the case where $\boldsymbol{\sigma} = I_2$ and $\boldsymbol{\alpha} = \mathrm{diag}(\alpha_1,\alpha_2)$, with $\alpha_\mu \in \mathbb{R}$, $\mu=1,2$.
Then we can choose
\be
    \boldsymbol{U} = \left( \begin{array}{cc} e^{\imag \, \varphi_1} 
    	&  0  \\  0  & e^{\imag \, \varphi_2}  \end{array} \right) \, \boldsymbol{U}_0 \, , \qquad
      \boldsymbol{U}_0 = \frac{1}{\|A\|}  \left( \begin{array}{cc} A_1 & - A_2^\ast  \\
        A_2  & A_1^\ast  \end{array} \right) \, ,   \label{2vFL_U}
\ee
so that $\boldsymbol{U}_x = \imag \,  \boldsymbol{\alpha} \, \boldsymbol{U}, \boldsymbol{U}_t = \imag \, \boldsymbol{\beta} \, \boldsymbol{U}$ and
\bez
    \boldsymbol{U}^\dagger \boldsymbol{\alpha} \, \boldsymbol{U} 
 =  \boldsymbol{U}_0^\dagger \, \boldsymbol{\alpha} \, \boldsymbol{U}_0   
 =  \frac{1}{\|A\|^2} \left( \begin{array}{cc}  \alpha_1 |A_1|^2 + \alpha_2 |A_2|^2 
    	&  (\alpha_2 - \alpha_1) \, A_1^\ast A_2^\ast  \\  (\alpha_2 - \alpha_1) \, A_1 A_2  & 
       \alpha_2 |A_1|^2 + \alpha_1 |A_2|^2  \end{array} \right) \, .
\eez
From (\ref{chi_x_orthog}) we obtain
\be
   \chi_{2x}  
 - \Big( \frac{1}{2} \Gamma^{-1} + \imag \, \frac{ \alpha_1 \alpha_2 \, \|A\|^2}{A^\dagger \boldsymbol{\alpha} A} \, I_n \Big) \, \chi_2 = \frac{(\alpha_2 - \alpha_1) A_1^\ast A_2^\ast}{A^\dagger \boldsymbol{\alpha} A} \big(  \chi_{1x} - \frac{1}{2} \Gamma^{-1} \chi_1 \big) \, ,   \label{2vFL_chi_2x}
\ee
and (\ref{chi1_xx}) becomes
\bez
  &&  \chi_{1xx} - \frac{1}{4} \Gamma^{-2} \chi_1 - \frac{\imag \, A^\dagger \boldsymbol{\alpha} A}{\|A\|^2}
   \big(  \chi_{1x} + ( \frac{1}{2} - A^\dagger \boldsymbol{\alpha} A ) \, \Gamma^{-1} \chi_1 \big)  \\
  && - \frac{\imag \, A_1 A_2}{\|A\|^2}  (\alpha_2 - \alpha_1) \,  \big( \chi_{2x} + ( \frac{1}{2} - A^\dagger \boldsymbol{\alpha} A ) \, \Gamma^{-1} \chi_2 \big) = 0 \, .
\eez
Eliminating $\chi_{2x}$ with the help of the preceding equation, we get
\be
 && \chi_{1xx} - \imag \, \frac{A^\dagger \boldsymbol{\alpha}^2 A}{A^\dagger \boldsymbol{\alpha} A} \, \chi_{1x} - \frac{1}{4} \Gamma^{-2} \chi_1    \nonumber \\
 && + \frac{\imag}{2 \|A\|^2 \, A^\dagger \boldsymbol{\alpha} A} 
 \Big( (\alpha_2-\alpha_1)^2 |A_1|^2 |A_2|^2  - (A^\dagger \boldsymbol{\alpha} A)^2 (1-2 A^\dagger \boldsymbol{\alpha} A) \Big) \, \Gamma^{-1} \chi_1  \nonumber \\
 && - (\alpha_2-\alpha_1) \frac{A_1 A_2}{\|A\|^2} \Big( \imag \, (1 - A^\dagger \boldsymbol{\alpha} A) \, \Gamma^{-1} - \frac{\alpha_1 \alpha_2 \|A\|^2}{A^\dagger \boldsymbol{\alpha} A} I_n \Big) \, \chi_2 = 0   \, .     \label{alpha-branch_eq}
\ee
We note that
\be
    \Xi &=& \|A\| \, \big( \chi_{1x} - \frac{1}{2} \Gamma^{-1} \chi_1 
    - \imag \, \chi \, \boldsymbol{U}_0^\dagger \, \boldsymbol{\alpha} \, \boldsymbol{U}_0  \, (1,0)^T \big) \nonumber \\
   &=& \|A\| \, \Big( \chi_{1x} - \big( \frac{1}{2}  \Gamma^{-1} + \imag \, \frac{A^\dagger \boldsymbol{\alpha} A}{\|A\|^2} \, I_n \big) \, \chi_1 \Big)
   - \imag \, (\alpha_2-\alpha_1) \, \frac{A_1 A_2}{\|A\|} \, \chi_2   
    \, .  \label{2vFL_Xi}
\ee
The Lyapunov equation (\ref{pw_vFLeq_Lyap}) now reads
\be
\Gamma \, \Omega + \Omega \, \Gamma^\dagger = \frac{\imag}{(A^\dagger \boldsymbol{\alpha} A)^2} \Xi \, \Xi^\dagger \Gamma^\dagger + \chi_1 \chi_1^\dagger + \chi_2 \chi_2^\dagger \, .
        \label{2vFL_Lyap}
\ee

From this point on, we have to treat the cases $\alpha_1=\alpha_2$ and $\alpha_1 \neq \alpha_2$ separately. 
This is done in Sections~\ref{subsec:alpha1=alpha2} and \ref{subsec:alpha1_neq_alpha2}, respectively. The first case is non-generic, but it exhibits interesting features, in particular a $U(2)$-symmetry and ``beating solitons". It therefore deserves attention, although it should be expected that solutions obtained in this case are recovered by taking a limit $\alpha_2 \to \alpha_1$ of corresponding ``generic" solutions.
Moreover, this special case is considerably simpler than the generic case, notably because the characteristic equation is only quadratic, and provides a convenient first step into the plethora of solutions generated via Theorem~\ref{thm:vFL_DT}.

The corresponding analysis is different, depending on whether the spectrum condition for the matrix $\Gamma$ in the Lyapunov equation holds or not. It is necessary to distinguish several cases. Moreover, for special choices of $\Gamma$, the linear system, which here is equivalent to a constant coefficient linear system, ``degenerates" and admits quasi-rational solutions. In Sections~\ref{subsec:alpha1=alpha2} and \ref{subsec:alpha1_neq_alpha2}, we will organize the various possibilities into a collection of examples. 

The resulting classes of solutions will be discussed only briefly in order to keep this work at a reasonable length. Satisfactory statements concerning absence of singularities and the behavior of the (multi-parameter) solutions are often not reached. These rather difficult problems are also met in the previous publications about the (vector) FL equation. 
If we provide plots of special solutions, then we regard them as displaying characteristic features of solutions from the respective class, which is partly based on evaluations using computer algebra.    

\begin{remark}
If $\boldsymbol{\sigma} = \mathrm{diag}(1,\sigma)$, with $\sigma \in \{\pm 1\}$, we choose 
\bez
\boldsymbol{U} = \frac{1}{\sqrt{A^\dagger \boldsymbol{\sigma} A}}  \left( \begin{array}{cc} e^{\imag \, \varphi_1} 
	&  0  \\  0  & e^{\imag \, \varphi_2}  \end{array} \right)
\left( \begin{array}{cc} A_1 & - \sigma \, A_2^\ast  \\
	A_2  & A_1^\ast  \end{array} \right) \, ,
\eez
which satisfies $(A_1 e^{\imag \, \varphi_1},A_2 e^{\imag \, \varphi_2})^T 
= \boldsymbol{U} \, (\sqrt{A^\dagger \boldsymbol{\sigma} A},0)^T$ and the $\boldsymbol{\sigma}$-unitarity condition $\boldsymbol{U}^\dagger \boldsymbol{\sigma} \boldsymbol{U} = \boldsymbol{\sigma}$. If $\sigma=-1$, so that $\boldsymbol{\sigma}$ is indefinite, we have to restrict $A$ such that $|A_1| \neq |A_2|$. \hfill $\Box$
\end{remark}

\begin{remark}
The form of the two-component vector FL equation in \cite{LFZ18,Ling+Su24} and corresponding solutions are obtained from ours via $x \mapsto 2 x$ and $t \mapsto -t/2$. In \cite{CYSGB18,YZCBG19,TZLBMYAB22} the two-component vector FL equation is taken in the form 
\bez
 && \imag \, D_\tau \frac{\partial v_1}{\partial \xi} + \frac{\sigma}{2} \frac{\partial^2 v_1}{\partial \tau^2}
 + (2 |v_1|^2 + |v_2|^2) \, D_\tau v_1 + v_1 v_2^\ast D_\tau v_2 =0  \, , \\
 && \imag \, D_\tau \frac{\partial v_2}{\partial \xi} + \frac{\sigma}{2} \frac{\partial^2 v_2}{\partial \tau^2}
 + (2 |v_2|^2 + |v_1|^2) \, D_\tau v_2 + v_2 v_1^\ast D_\tau v_1 =0  \, ,
\eez	
where $D_\tau = 1 + \imag \, \epsilon \, \partial/\partial \tau$, 
$\xi$ and $\tau$ are independent variables, and $\epsilon \neq 0$  and $\sigma$ are real constants. Then
\bez
    v_j(\tau,\xi) = e^{\imag \, (\epsilon^{-1} \tau + \sigma \epsilon^{-2} \xi)} \, u_j(x,t)  \qquad j=1,2 \, , 
\eez 
with $x = -\frac{1}{4} \sigma \epsilon^{-4} \, (\sigma \, \xi +  \epsilon \, \tau)$ and $t=\xi$, maps (solutions of) the latter system to (solutions of) our form of the two-component vector FL equation, and vice versa (also see \cite{Ling+Su24}).
\hfill $\Box$
\end{remark}

\subsection{The special plane wave case $\alpha_1 = \alpha_2$}
\label{subsec:alpha1=alpha2}
If $\alpha_1 = \alpha_2 =: \alpha$, we have $\beta := \beta_1 = \beta_2 = 2 \|A\|^2-1/\alpha$. 
Writing
\bez
    \chi_1 = e^{\frac{\imag}{2} \varphi} \, \tilde{\chi}_1 \, , \qquad
    \chi_2 = e^{\imag \, \varphi} \, \tilde{\chi}_2 \, ,
\eez
with 
\bez
 \varphi = \alpha \, x + \beta \, t
         = \alpha \, x + (2 \|A\|^2-1/\alpha) \, t \, ,  
\eez
the linear system becomes
\be
    \tilde{\chi}_{1xx} - R^2 \, \tilde{\chi}_1 = 0 \, , \quad
    \tilde{\chi}_{1t} + \frac{\imag}{\alpha} \Gamma \, \tilde{\chi}_{1x} = 0 \, ,
    \label{2vFL_spw_linsys1}
\ee
where\footnote{If the right hand side does not admit a matrix root, we shall simply regard $R^2$ as an abbreviation for the right hand side of (\ref{special_quadratic_R_eq}).}
\be
   R^2 = \frac{1}{4} \Gamma^{-2} - \frac{\imag}{2} \, \alpha^2 \, \beta \, \Gamma^{-1} - \frac{\alpha^2}{4} I_n  \, ,
    \label{special_quadratic_R_eq}
\ee
and
\bez
 \tilde{\chi}_{2x} - \frac{1}{2} \Gamma^{-1} \tilde{\chi}_2 = 0 \, , \quad
 \tilde{\chi}_{2t} - \frac{1}{2} \Gamma \, \tilde{\chi}_2 = 0 \, ,
\eez
which is solved by
\bez
  \tilde{\chi}_2 = e^{ \frac{1}{2} (\Gamma^{-1} x + \Gamma \, t)} \, c_0 \, ,
\eez
where $c_0$ is a constant $n$-component column vector.

After determining the solution of the linear system and then $\Omega$, finally, 
\be
 u' &=& e^{\imag \, \varphi} \,
 \left( \begin{array}{cc} A_1 + (\alpha \, \|A\|^3)^{-1} (A_1 \chi_1^\dagger - A_2^\ast \chi_2^\dagger) \, \Omega^{-1} \Xi  \\[4pt]
   A_2 + (\alpha \, \|A\|^3)^{-1} (A_2 \chi_1^\dagger + A_1^\ast \chi_2^\dagger) \, \Omega^{-1} \Xi \end{array} \right) \nonumber \\
   &=& e^{\imag \, \varphi} \, \boldsymbol{U}_0 \, \left[ \left( \begin{array}{c} \|A\| \\ 0 \end{array} \right) + \frac{1}{\alpha \, \|A\|^2} \left( \begin{array}{c} 
   	\chi_1^\dagger \, \Omega^{-1} \Xi \\ \chi_2^\dagger \, \Omega^{-1} \Xi \end{array} \right) \right]  \, ,
      \label{special_pw_u'}
\ee
with 
\bez
  \Xi = \|A\| \Big( \chi_{1x} - \big( \frac{1}{2}  \Gamma^{-1} + \imag \, \alpha \, I_n \big) \, \chi_1 \Big) \, ,
\eez
is a solution of the two-component vector FL equation. Here $\boldsymbol{U}_0$ is the 
constant unitary matrix in (\ref{2vFL_U}). According to Proposition~\ref{prop:vFL_unitary_symmetry}, it yields a symmetry of the vector FL equation and of its Darboux transformation. It is therefore 
sufficient to elaborate (\ref{special_pw_u'}) with $\boldsymbol{U}_0$ replaced by the identity matrix, since all other solutions are obtained by application of a constant unitary matrix. 
This particular solution corresponds to having a plane wave background only in the first component ($A_1=\|A\|$) and zero background in the second ($A_2=0$). 
\vspace{.2cm}

We have to distinguish the following two cases.
\vspace{.1cm}

\noindent
\textbf{(1)} 
If the spectrum condition for $\Gamma$ holds, the Lyapunov equation (\ref{2vFL_Lyap}) possesses a unique solution. In particular, if 
$\Gamma = \mathrm{diag}(\gamma_1,\ldots,\gamma_n)$, with 
$\gamma_i \neq -\gamma_j^\ast$, $i,j=1,\ldots,n$, 
then $\Omega$ is given by the generalized Cauchy matrix with components 
\bez
\Omega_{ij} =\frac{1}{\gamma_i + \gamma_j^\ast} \Big( \frac{\imag}{\alpha^2 \|A\|^4} \,  \Xi_i \Xi_j^\ast \gamma_j^\ast + \chi_{1i} \chi_{1j}^\ast + \chi_{2i} \chi_{2j}^\ast \Big)  \, .
\eez
Given an $n=1$ soliton solution, generated with some value of $\gamma$, an ``$n$-th order" extension of it is determined by an $n \times n$ (lower triangular) Jordan block with eigenvalue $\gamma$.    
\vspace{.1cm}

\noindent
\textbf{(2)} 
If the spectrum condition does \emph{not} hold, the Lyapunov equation (\ref{2vFL_Lyap}) 
imposes a constraint on the solution of the linear system and determines at most part of $\Omega$. 
According to Corollary~\ref{cor:vFL_DT_nosc}, we also have to solve the equation (\ref{vFLeq_key}), which is now
\be
    \Gamma \Omega + \Omega^\dagger \Gamma^\dagger  = \chi_1 \chi_1^\dagger + \chi_2 \chi_2^\dagger \, ,
          \label{2vFLeq_pw}
\ee
and the compatible differential equations
\be
\Omega_x &=& \frac{\imag}{\alpha^2 \|A\|^4} \Xi_x  
\Xi^\dagger + \frac{1}{2}\Gamma^{-1} \, \chi_1\chi_1^\dagger \, \Gamma^{\dagger -1} 
 + \Big( \chi_{2x} - \imag \, \alpha \, \chi_2 \Big)\chi_2^\dagger \Gamma^{\dagger -1}
      \, , \nonumber \\
\Omega_t &=& -\frac{\imag}{2 \alpha^2 \|A\|^4} \, \Xi \, \Xi^\dagger \, \Gamma^\dagger  
 -\frac{\imag}{\alpha \, \|A\|} \Xi \, \chi_1^\dagger + \frac{1}{2} (\chi_1\chi_1^\dagger+\chi_2\chi_2^\dagger) \, .   \label{2vFL_Om_diff}
\ee

If the expression on the right hand side of (\ref{special_quadratic_R_eq}) is invertible, then it indeed possesses a matrix root $R$. 
(\ref{2vFL_spw_linsys1}) is then solved by
\be
 \chi_1 = e^{\frac{\imag}{2} \varphi} \, \sum_\kappa e^{R_\kappa (I_n x -\frac{\imag}{\alpha}  \Gamma \, t)} \, c_\kappa \, ,                 
   \label{special_pw_sol_chi1}
\ee
where $c_\kappa$ are constant $n$-component column vectors and
$\{R_\kappa\}$ is a set of different roots of (\ref{special_quadratic_R_eq}).

If the right hand side of (\ref{special_quadratic_R_eq}) is \emph{not} invertible, so that typically no square root exists (and $R^2$ in (\ref{2vFL_spw_linsys1}) should then be replaced by the right hand side of (\ref{special_quadratic_R_eq})), we have to step back to (\ref{2vFL_spw_linsys1}) to determine its solution. This is actually an important case, since it leads in particular to rogue wave solutions.  

Let us recall that a solution $u'$ is called quasi-rational if it has the form $u'= e^{\imag \, \varphi} \, (v_1,v_2)^T$ with (in $x$ and $t$) rational expressions $v_1,v_2$. If $u_1$ is a solution of the scalar FL equation, then 
$u = U_0 \, (u_1,0)^T$, with a constant unitary matrix $U_0$ solves the two-component vector FL equation. In this case we say that $u$ is equivalent to the solution $u_1$ of the scalar FL equation. 

\begin{proposition}
	\label{prop:2vFL_spw_quasi-rational}
Any quasi-rational solution of the two-component vector FL equation, obtained by the vectorial Darboux transformation from the special plane wave seed with $\alpha_1=\alpha_2$, is equivalent to a quasi-rational solution of the scalar FL equation.  
\end{proposition}
\begin{proof}
In order for (\ref{special_pw_u'}) to be quasi-rational, it is necessary that 
$\chi_2 =0$. But then the second component of $u'$ in the last expression in (\ref{special_pw_u'}) vanishes. 
\end{proof}

\begin{remark}
	\label{rem:other_vFL}
In \cite{Yang+Zhang18,ZYCW17,Xu+Chen19,Wang+Chen19,Yue+Huang22} only the special ($\alpha_1=\alpha_2$) plane wave background case has been used as a seed solution,   
applying a generalized Darboux transformation, with a limit procedure that takes different spectral parameters to the same value, in order to obtain higher order solutions, in particular higher order rogue waves. In our framework, the corresponding $n$-th order solution is obtained by choosing $\Gamma$ to be an $n \times n$ Jordan block. 
Proposition~\ref{prop:2vFL_spw_quasi-rational} indicates that all quasi-rational solutions presented in the mentioned publications are equivalent to quasi-rational solutions of the scalar FL equation, which is indeed the case. This is a useful insight within the \emph{vectorial} Darboux transformation framework. 
The publication \cite{LFZ18} also contains solutions with the special plane wave background, but steps into the richer case of different wave numbers ($\alpha_1 \neq \alpha_2$, see Section~\ref{subsec:alpha1_neq_alpha2}), though without treating rogue waves.
\hfill $\Box$
\end{remark}

\subsubsection{$n=1$}
\label{sec:2vFL_special_pw_n=1}
For $n=1$, the roots of the quadratic characteristic equation (\ref{special_quadratic_R_eq}) are 
$r_\pm  = \pm w/(2 \gamma)$ with
\be
     w = \sqrt{1 - 2 \imag \, \alpha^2 (2 \|A\|^2 - \alpha^{-1}) \,  \gamma - \alpha^2 \gamma^2}   \, .  \label{special_pw_roots}
\ee
Throughout this subsection, $\chi_2$ is always given by 
\bez
   \chi_2 = c_0 \, e^{\imag \, \varphi} \, e^{\frac{1}{2} (\gamma^{-1} x + \gamma t)} \, .
\eez   

\begin{example}
	\label{ex:special_pw_breather}
Let $\mathrm{Re}(\gamma) \neq 0$ and also $w \neq 0$.  Then (\ref{special_pw_sol_chi1}) yields the general solution of the linear system for $\chi_1$,
\bez
    \chi_1 = e^{\frac{1}{2} \imag \, \varphi} \, \big( 
    c_1 \, e^\rho + c_2 \, e^{-\rho} \big) \, ,
\eez
where
\bez
   \rho = \frac{1}{2} w \, \big( \frac{x}{\gamma} - \imag \, \frac{t}{\alpha} \big) 
 = \frac{1}{2} \big( \mathrm{Re}(w/\gamma) \, x + \alpha^{-1} \mathrm{Im}(w) \, t \big) + \frac{1}{2} \imag \, \big( \mathrm{Im}(w/\gamma) \, x - \alpha^{-1} \mathrm{Re}(w) \, t \big) \, .
\eez
We obtain 
\bez
   \Xi &=& - \frac{\|A\|}{2 \gamma} \, e^{\frac{1}{2} \imag \, \varphi} \, W
    \, , \qquad
   W = c_1 \, (1+\imag \, \alpha \gamma-w) \, e^\rho 
   + c_2 \, (1+ \imag \, \alpha \gamma +w) \, e^{-\rho}  \, , \\
   \Omega &=& \frac{1}{2 \mathrm{Re}(\gamma)} \Big( \frac{\imag \, \gamma^\ast}{\alpha^2 \|A\|^4} |\Xi|^2 + |\chi_1|^2 + |\chi_2|^2 \Big) \\
   &=& \frac{1}{2 \mathrm{Re}(\gamma)} \, \Big( |c_0|^2 \, e^{\mathrm{Re}(\gamma) ( x/|\gamma|^2 + t)}
   + \big| c_1 \, e^\rho + c_2 \, e^{-\rho} \big|^2 
   + \frac{\imag}{4 \alpha^2 \, \|A\|^2 \, \gamma } \, |W|^2 \Big)
    \, .
\eez
As mentioned before, it is sufficient to elaborate the particular case of (\ref{special_pw_u'}) where $\boldsymbol{U}_0$ is replaced by the identity matrix. Inserting the above expression for $\Omega$, then
yields the following solution of the two-component vector FL equation,
\bez
  u_1' = \|A\| \, e^{\imag \, \varphi} \, \Big( 1 -  
  \frac{W \, (c_1^\ast \, e^{\rho^\ast} + c_2^\ast \, e^{-\rho^\ast})}{2 \alpha \, \gamma \, \|A\|^2 \, \Omega} \Big)  \, , \qquad
  u_2' = - \frac{c_0^\ast \, W}{2 \alpha \gamma \, \|A\| \, \Omega} e^{ \frac{1}{2} \big( \imag \, \varphi + \gamma^{\ast -1} x + \gamma^\ast t \big)}  \, .
\eez
If $c_0 =0$ and either $c_1 \neq 0$ or $c_2 \neq 0$, we have $u_2'=0$ and $u_1'$ is then a solution of the scalar FL equation. Therefore we shall now assume $c_0 \neq 0$. If $c_1=c_2=0$, $u_1'$
reduces to the plane wave background (seed) solution and $u_2'=0$. Therefore, besides $c_0 \neq 0$, at least one of the constants $c_1,c_2$ should be non-zero. 

If $c_2=0$, the above solution reduces to 
\bez
    u_1' &=& \|A\| \, e^{\imag \, \varphi} \, \Big( 1 - \mathrm{Re}(\gamma) \,
  \frac{ \gamma^\ast (1-w) + \imag \, \alpha \, |\gamma|^2 }{2 \alpha \, |\gamma|^2 \, \|A\|^2 \, 
 	e^{ 2 \delta} } \, \big[ 1 + \mathrm{tanh}\big( \vartheta + \imag \, \mathrm{Im}(\delta) \big) \big] \Big) \, ,  \\
  u_2' &=& \mathrm{Re}(\gamma)  \, \frac{ \gamma^\ast (1-w) + \imag \, \alpha \, |\gamma|^2}{2 \alpha \, |\gamma|^2 \, \|A\| \, e^{\mathrm{Re}(\delta)} } \, 
  e^{\imag \, \varphi} \, e^{-\frac{1}{2} \imag \, \psi} \, 
  \mathrm{sech}\big( \vartheta + \imag \, \mathrm{Im}(\delta) \big) \, ,
\eez
after removing a constant, $\mathrm{Re}(\delta) + \ln(|c_1/c_0|)$, in the argument of 
the hyperbolic functions via a translation in $x$ or $t$, and using the unitary symmetry of the vector FL equation to eliminate constant phase factors. Here we set
\bez
\vartheta &=& \frac{1}{2 |\gamma|^2}  \mathrm{Re}\big( \gamma \, (w^\ast - 1) \big) \, x + \frac{1}{2} \big( \alpha^{-1} \mathrm{Im}(w) - \mathrm{Re}(\gamma) \big) \, t \, , \\
\psi &=& \big( \alpha - \mathrm{Im}(\gamma)/|\gamma|^2 - \mathrm{Im}(w/\gamma) \big) \, x + \big( 2 \|A\|^2 + \mathrm{Im}(\gamma) + \alpha^{-1} (\mathrm{Re}(w)-1) \big) \, t \, , \\
e^{2\delta} &=& 1 + \imag \, \gamma^\ast \frac{|1+\imag \, \alpha \gamma - w|^2}{4 \alpha^2 |\gamma|^2 \, \|A\|^2} 
 = 1 + \big( \mathrm{Im}(\gamma) + \imag \, \mathrm{Re}(\gamma) \big) \frac{|1+\imag \, \alpha \gamma - w|^2}{4 \alpha^2 |\gamma|^2 \, \|A\|^2} \, .
\eez
The above expressions also cover the alternative case $c_1=0$ (and $c_2 \neq 0$) if 
we allow for $w$ also the negative of the square root in (\ref{special_pw_roots}). 

If $\mathrm{Re}\big( \gamma \, (w^\ast - 1) \big) = 0$, both components of $u'$ are functions of $t$ only, multiplied by a phase factor with an exponent linear in $x$. 

If $\mathrm{Re}\big( \gamma \, (w^\ast - 1) \big) \neq 0$, then the wave, described by $u'$, travels with constant velocity  
\bez
 v = |\gamma|^2 \, \frac{\mathrm{Re}(\gamma) - \alpha^{-1} \mathrm{Im}(w) }{\mathrm{Re}\big( \gamma \, (w^\ast - 1) \big)} \, ,
\eez
along the $x$-axis, while $u_2'$ oscillates (i.e., ``beats") with period
\bez
    T = 4 \pi \, \big| 2 \|A\|^2 - \alpha^{-1} +\mathrm{Im}(\gamma) + \alpha^{-1} \mathrm{Re}(w) \big|^{-1} \, ,
\eez
disregarding the phase factor $e^{\imag \, \varphi}$, which causes a modulation and is common to both components.  
In the comoving coordinate $x' = x - v \, t$, we have
\bez
   \psi &=& \Big( \big( \alpha - \mathrm{Im}(\gamma)/|\gamma|^2 - \mathrm{Im}(w/\gamma) \big) \, v + 2 \|A\|^2 + \mathrm{Im}(\gamma) + \alpha^{-1} (\mathrm{Re}(w) - 1) \Big) \, t \\
   && + \big( \alpha - \mathrm{Im}(\gamma)/|\gamma|^2 - \mathrm{Im}(w/\gamma) \big) \, x' \, , 
\eez
so that $u_2'$ exhibits, at constant $x'$, the period
\bez
   T' = 4 \pi \, \Big| \big( \alpha - \mathrm{Im}(\gamma)/|\gamma|^2 - \mathrm{Im}(w/\gamma) \big) \, v + 2 \|A\|^2 + \mathrm{Im}(\gamma) + \alpha^{-1} (\mathrm{Re}(w) - 1) \Big|^{-1} \, . 
\eez  
The absolute squares of the components of $u'$ are found to be given by
\bez
|u_1'|^2 = \|A\|^2 \, \Big| 1 + \frac{2 \, \mathrm{Re}(\gamma) \, \mathrm{Re}(\gamma \, (w^\ast-1))}{\alpha \, |\gamma|^2 \, \|A\|^2 \, |e^{\vartheta+2 \delta} + e^{-\vartheta}|^2} \Big| \, , \qquad
|u_2'|^2 = \mathrm{Re}(\gamma)^2 \, \frac{|1+\imag \, \alpha \gamma-w|^2 }{\alpha^2 \, |\gamma|^2 \, \|A\|^2 \, |e^{\vartheta+2\delta} + e^{-\vartheta}|^2} \, .
\eez
As a consequence, 
$u_1'$ represents a dark soliton if $\alpha \, \mathrm{Re}(\gamma) \, \mathrm{Re}(\gamma \, (w^\ast-1)) < 0$. If the maximum of the absolute value in the expression for $|u_1'|^2$ is smaller than 1, we have a single dark soliton. An example is shown in Fig.~\ref{fig:2vFL_spw_beat}.
\begin{figure}[h]
	\begin{center}
		\includegraphics[scale=.34]{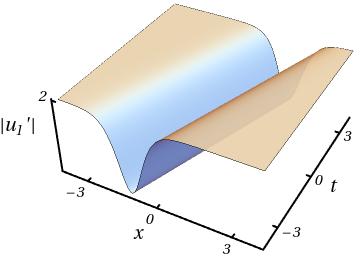} 
		\hspace{1cm}
		\includegraphics[scale=.34]{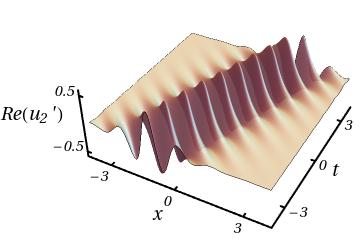} 
		\parbox{15cm}{
			\caption{Plots of the absolute values of the first, and the real part (the imaginary part is similar) of the second component for an $n=1$ solution from the class in Example~\ref{ex:special_pw_breather}. The parameters are 
				$\alpha = \|A\| = 2$, $\gamma = 1 + \imag$, $c_0=c_1=1$, $c_2=0$, but we 
				replaced $w$ by $-w$. The second plot shows the ``beating". The 
				velocity is $v \approx 0.3$, the beating period $T' \approx 1$.  
				\label{fig:2vFL_spw_beat} } 
		}
	\end{center}
\end{figure}
If the maximum is greater than 1, there will be two dark soliton parts on both sides of an ``anti-dark" one, i.e., a bright soliton on the plane wave background. 

When $\alpha \, \mathrm{Re}(\gamma) \, \mathrm{Re}(\gamma\,(w^\ast-1)) > 0$, then $u_1'$ represents an anti-dark soliton, since its amplitude is then greater than the background density value and $|u_1'|$ falls off toward it as $x \to \pm \infty$. 

In each case, $u_2'$ describes a bright soliton (on vanishing background). Solutions of the kind described here already appeared in \cite{LFZ18} (see Section~5.2 therein).  
\hfill $\Box$
\end{example}

\begin{example}
		\label{ex:special_pw_beating_solitons}
Applying $\boldsymbol{U}_0$, given in (\ref{2vFL_U}), with $A_2 = A_1 \in \mathbb{R}$, to the solution with $c_2=0$ in the preceding example, the resulting new solution, 
\be
    u'_{\mathrm{new}} = \frac{1}{2} \left( \begin{array}{c} u_1' - u_2' \\ u_1' + u_2' \end{array} \right) \, ,  \label{u'_beating}
\ee
is generated from the uniform plane wave seed $u= A_1 e^{\imag \, \varphi} (1,1)^T$ by the Darboux transformation. Now we have $\|A\|^2 = 2 A_1^2$. The structure of the solution 
displayed in (\ref{u'_beating}), and with a relative oscillation of $u_2'$, is characteristic for what has been called 
``beating solitons" in the NLS literature (cf. \cite{LCA24} and references cited there).
We have 
\bez
    u'_{\mathrm{new}}(x',t+T'/2) 
  = \frac{1}{2} \left( \begin{array}{c} u_1'(x',t) + u_2'(x',t) \\ 
  	u_1'(x',t) - u_2'(x',t) \end{array} \right) \, ,
\eez
which means an exchange of the two components, 
and $u'_{\mathrm{new}}(x',t+T') = u'_{\mathrm{new}}(x',t)$. 

An extreme example from the resulting class of solutions is presented in  Fig.~\ref{fig:2vFL_special_beating}. Here the parameters are chosen in such a way that the coefficient of $t$ in the expression for $\psi$ vanishes, i.e., $4 A_1^2 - \alpha^{-1} +\mathrm{Im}(\gamma) + \alpha^{-1} \mathrm{Re}(w) =0$, so that $T$ is infinite while $T'$ is finite.
\begin{figure}[h]
	\begin{center}
		\includegraphics[scale=.34]{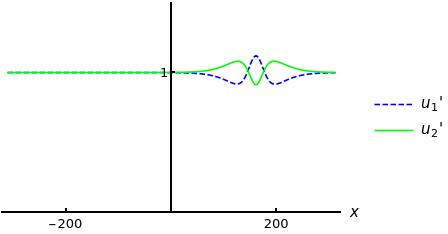} 
		\includegraphics[scale=.34]{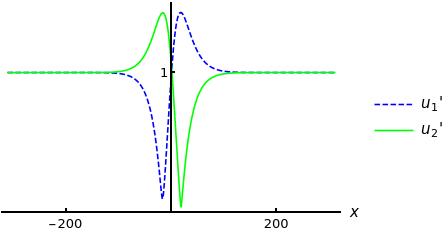} 
		\includegraphics[scale=.34]{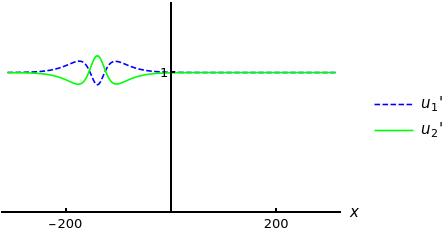} 
		\parbox{15cm}{
\caption{Plots of the absolute value of the components of an $n=1$ beating soliton solution from the class in Example~\ref{ex:special_pw_beating_solitons} at $t=-19.2$ (left), $t=0$ (middle) and $t=16.9$ (right). Here we set $A_2=A_1$, $\alpha = 1/(4 A_1^2)$ (so that $\beta=0$ and $w = \sqrt{1-\gamma^2/(16 A_1^4)}$),  
$A_1 = c_0 = c_1 = 1$, $c_2 =0$ and $\gamma=10$. 
Comparing the first and the third plot, we observe an exchange of the 
profiles of the absolute values of the components $u_1'$ and $u_2'$. 
				\label{fig:2vFL_special_beating} } 
		}
	\end{center}
\end{figure}
Now  
\bez
x \mapsto x - 2 \pi/\big( \alpha - \mathrm{Im}(\gamma)/|\gamma|^2 - \mathrm{Im}(w/\gamma) \big) 
\eez 
exchanges the signs in the expression for 
the components of $u'_{\mathrm{new}}$. Together with
\bez
t \mapsto t - 2 \pi \alpha \, \frac{\mathrm{Re}(\gamma) /|\gamma|^2 - \mathrm{Re}(w/\gamma)}{\big( \alpha - \mathrm{Im}(\gamma)/|\gamma|^2 - \mathrm{Im}(w/\gamma) \big) \big( \mathrm{Im}(w) - \alpha \, \mathrm{Re}(\gamma) \big)} \, ,
\eez  
it maps $|u'_{\mathrm{new},1}|$ into $|u'_{\mathrm{new},2}|$, and vice versa. Applying this map twice, 
$|u'_{\mathrm{new},i}|$ at the shifted values of the independent variables coincides with $|u'_{\mathrm{new},i}|$ itself, for $i=1,2$. 
\hfill $\Box$
\end{example}

Next we turn to the case $\mathrm{Re}(\gamma) = 0$, where the spectrum condition for the Lyapunov equation does not hold.

\begin{example}
\label{ex:special_pw_gamma_imaginary}
Let $\gamma = - \imag \, k$ with $k \in \mathbb{R} \setminus \{0\}$. Then the two roots of the characteristic equation (\ref{special_quadratic_R_eq}) are 
\bez
	r_\pm =   \pm  \frac{1}{2 k}  \sqrt{ 2 \, \alpha^2 (2 \|A\|^2 - \alpha^{-1}) \, k - \alpha^2 k^2 - 1} \, .
\eez	
The discriminant is real. If it is negative (positive), then $r_\pm$ is imaginary (respectively, real). 
The Lyapunov equation can only be satisfied if just one of the constants $c_\kappa$ is different from zero and, moreover,  $c_0=0$. We thus restrict the solution of the linear system to
\bez
	\chi_1 = e^{\frac{\imag}{2} \varphi} \, e^{r ( x - \frac{k}{\alpha} \, t)} \, c_1 \,  , \qquad 
        \chi_2=0 \, ,
\eez
where $c_1 \neq 0$ and $r$ now stands for any of the roots $r_\pm$. 
As a consequence of $\chi_2=0$, the resulting solution is equivalent to a solution of the scalar FL equation.
The Lyapunov equation reduces to the constraint
\bez
  k \, \big| r - \frac{\imag}{2} ( \alpha + \frac{1}{k} ) \big|^2 
       = \alpha^2 \|A\|^2 \, ,
\eez
which requires $k >0$. This is only compatible with the expression 
for $r$ if it is real, in which case the constraint is automatically satisfied. In the following, we therefore assume that $r$ is real.
The differential equations for $\Omega$ read
\bez
\Omega_x = \frac{|c_1|^2}{2\,k^2} \big( 1 + 2 \, \imag \, k \, ( r + \imag \, \frac{\alpha}{2}) \big) \, e^{2\theta} \, , \qquad 
\Omega_t = - \frac{k}{\alpha} \, \Omega_x \, ,
\eez
where
\bez
 \theta = r \, \big( x-\frac{k}{\alpha} \, t \big) \, .
\eez  
\noindent
\textbf{(i)} $r \neq 0$. We obtain
\bez
\Omega = \frac{\imag \, |c_1|^2\,(1+\imag \, k \, (2 r + \imag \, \alpha) )}{2\,k^2\,(\alpha+\imag \, (2 r + \imag \, \alpha) )} \, e^{2\theta} + c \, ,
\eez 
with a constant $c$, which (\ref{2vFLeq_pw}) requires to be real. The generated solution of the two-component vector FL equation is
\bez
 u_\mu' = A_\mu \,  e^{\imag \, \varphi} \, \Big( 1 
 + \frac{ k\, |c_1|^2 \big( k\,  (2 r + \imag \, \alpha) - \imag \, (1+2 \, \alpha \, k) \big) \, (\alpha+\imag \, (2 r + \imag \, \alpha) ) \, e^{2\theta}}{\alpha \|A\|^2 \big( \imag \, |c_1|^2 (1+\imag \, k \,  (2 r + \imag \, \alpha) ) \, e^{2\theta} + 2 \, k^2 \, c \, (\alpha+ \imag \, (2 r + \imag \, \alpha) )\big)} \Big) 
 \, , \quad  \mu=1,2 \, .
\eez

\noindent
\textbf{(ii)}  $r =0$. Then we have 
\bez
  k = 2 \, \|A\|^2 - \alpha^{-1} \pm 2 \sqrt{\|A\|^2 (\|A\|^2 - \alpha^{-1})} \, ,
\eez
which is real if and only if $\|A\|^2 \geq \alpha^{-1}$. In case of the plus sign, $k$ is then positive. 
In case of the minus sign, $k$ is only positive if $\alpha<0$. Since now  $\theta =0$, we obtain
\bez
     \Omega = \frac{|c_1|^2}{2 k} (1-k \alpha) \, \big( \frac{x}{k} - \frac{t}{\alpha} \big) + c \, ,
\eez 
with a constant $c$, for which (\ref{2vFLeq_pw}) requires $\mathrm{Im}(c) =|c_1|^2/(2 k)$. Obviously, $\mathrm{Re}(c)$ can be absorbed by a shift (respectively redefinition) of one of the independent variables and multiplication with a constant phase factor (which is a special case of the unitary symmetry of the vector FL equation). 
We obtain the non-singular quasi-rational solution
\be
 u' &=& e^{\imag \, \varphi} \, \Big( 1 - \frac{\imag \, |c_1|^2}{2 \alpha \, \|A\|^2 \, \Omega} \big( \alpha + \frac{1}{k} \big) \Big) 
\left( \begin{array}{cc} A_1 \\ A_2 \end{array} \right) \nonumber \\
  &=& e^{\imag \, \varphi} \, \Big( 1 - \frac{\imag \, k \, ( 1 + k \alpha )}{\|A\|^2 \, \big( (1-k \alpha) \, ( \alpha \, x - k \, t ) + \imag \, k \alpha \big)}  \Big) 
  \left( \begin{array}{cc} A_1 \\ A_2 \end{array} \right) 
  \, ,   \label{special_quasi-rational_soliton} 
\ee
which represents a soliton rationally concentrated along the line $x = (k/\alpha) \, t$ (hence with velocity $k/\alpha$), where $|u'_\mu| = |A_\mu| \, |1-(\alpha^{-1}+k)/\|A\|^2|$. This shows that we also obtain dark solitons in this way. An example is displayed in Fig.~\ref{fig:special_rational_dark_soliton}.\footnote{This  solution is in the ``repeated real root" class, obtained via the ``limit technique" in \cite{LFZ18}, page 206.}

\begin{figure}[h]
	\begin{center}
		\includegraphics[scale=.4]{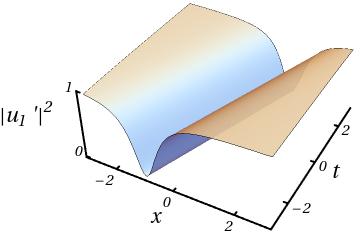} 
		\parbox{15cm}{
	\caption{Plot of the absolute square of the first component of an $n=1$ quasi-rational solution of the two-component vector FL equation, from the class given by (\ref{special_quasi-rational_soliton}). Here we chose 
	$\alpha = 1$, $A_1 = 1$, $A_2 = 1/\sqrt{3}$ and $k=1/3$. The second component just differs in the amplitude.    \label{fig:special_rational_dark_soliton} } 
		}
	\end{center}
\end{figure}
\hfill $\Box$
\end{example}

In case of a double root, i.e., $r_+ = r_-$, (\ref{special_pw_sol_chi1}) does not constitute the complete solution of the corresponding part of the linear system. Then we have to find the missing solution, which is done in the next example. This leads, in particular, to quasi-rational solutions.

\begin{example}
\label{ex:special_pw_double_root}
We have a double root if and only if $r=0$, so that 
\be
  \gamma 
  = - \imag \, \beta \pm 2 \sqrt{ \|A\|^2 (\alpha^{-1} - \|A\|^2) } \, .
       \label{special_gamma_for_r=0}
\ee
Then the general solution of (\ref{2vFL_spw_linsys1}) is
\bez
   \chi_1 = ( c_1 + c_2 \, \rho ) \, e^{\frac{\imag}{2} \, \varphi} \, ,  \qquad \rho = x -\imag \, \frac{\gamma}{\alpha} \, t \, ,
\eez
with complex constants $c_1,c_2$. If the real discriminant in the expression for $\gamma$ is non-positive, then $\gamma$ is imaginary and the constraint resulting from the 
Lyapunov equation enforces $c_2=0$, which takes us back to the case covered under (ii) in Example~\ref{ex:special_pw_gamma_imaginary}. In the remaining case, where $\alpha^{-1} > \|A\|^2$, we have
$\mathrm{Re}(\gamma) \neq 0$. We obtain
\bez
\Xi = - \frac{\|A\|}{2 \gamma} \, \Big( (1 + \imag \, \alpha \gamma)  (c_1 + c_2 \, \rho) - 2 c_2 \gamma  \Big) \, e^{\frac{\imag}{2} \, \varphi}  \, , 
\eez
and the Lyapunov equation has the unique solution
\bez
  \Omega = \frac{1}{2 \, \mathrm{Re}(\gamma)} \big( \frac{\imag \, \gamma^\ast}{\alpha^2 \|A\|^4} |\Xi|^2 + |\chi_1|^2 + |\chi_2 |^2 \big)   \, .
\eez
According to (\ref{special_pw_u'}), the generated solution of the two-component FL equation is now given by 
\bez
  u' = e^{\imag \, \varphi} \, \boldsymbol{U}_0 \, \left( \begin{array}{c} \|A\| + (\alpha \, \|A\|^2 \, \Omega)^{-1}  
   	\chi_1^\ast \, \Xi \\ 
   (\alpha \, \|A\|^2 \, \Omega)^{-1} \, \chi_2^\ast \, \Xi \end{array} \right) \, .
\eez
This solution is quasi-rational if $c_0 =0$ and then models a rogue wave, rationally localized in space and time and equivalent to a corresponding solution of the scalar FL equation. Fig.~\ref{fig:2vFL_special_rogue_wave} displays an example. 
\begin{figure}[h]
	\begin{center}
		\includegraphics[scale=.4]{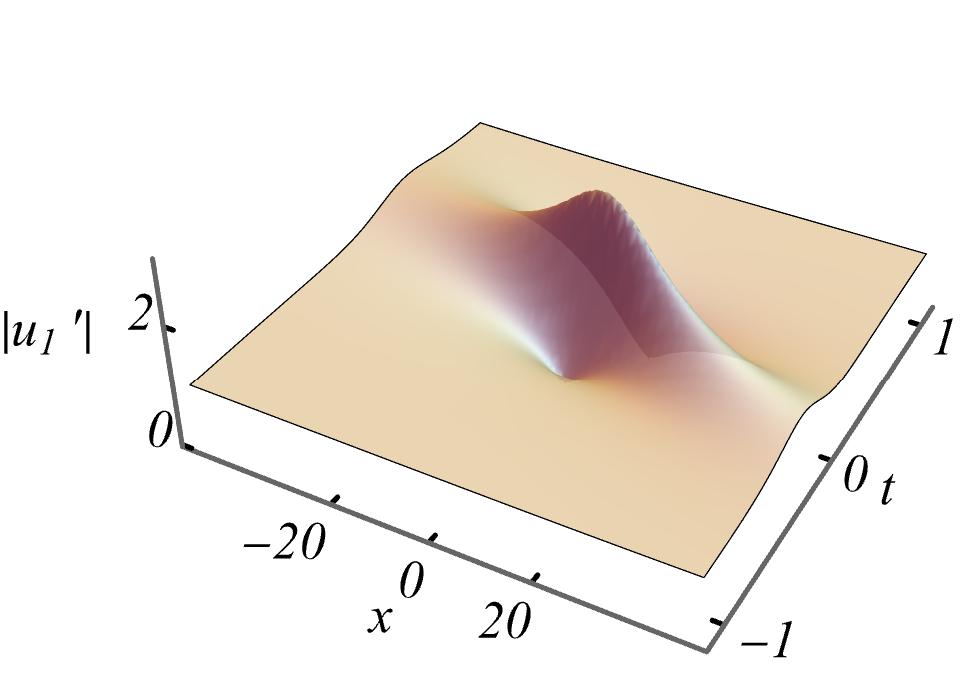} 
		\parbox{15cm}{
\caption{Plots of the absolute value of the first component of an $n=1$ quasi-rational solution from the class in Example~\ref{ex:special_pw_double_root}. We chose $\alpha = 1/10$, $\gamma =8+6 \imag$, $A_1 = A_2 = c_1 = c_2 = 1$ and $c_0=0$. 	\label{fig:2vFL_special_rogue_wave} } 
		}
	\end{center}
\end{figure}
If $c_0 \neq 0$, the solution is not equivalent to a solution of the scalar equation. In this case, we have a superposition with a breather, see Fig.~\ref{fig:2vFL_special_rw_and_breather} for an example. The double root case treated in this example has already
been considered in \cite{Xu+Chen19}.
\begin{figure}[h]
	\begin{center}
		\includegraphics[scale=.4]{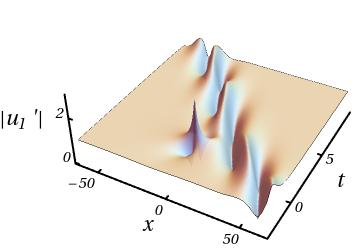} 
                   \hspace{1cm}
                   \includegraphics[scale=.4]{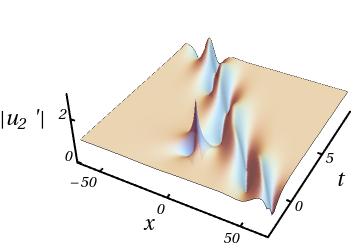} 
		\parbox{15cm}{
\caption{Plots of the absolute values of the two components of an $n=1$ solution from the class in Example~\ref{ex:special_pw_double_root}. We chose $\alpha = 1/3$, $A_1 = A_2 = c_0 = c_2 = 1$ and $c_1=10$. 	
\label{fig:2vFL_special_rw_and_breather} } 
		}
	\end{center}
\end{figure}
\hfill $\Box$
\end{example}

\subsubsection{$n=2$}

If $\Gamma$ is diagonal and either satisfies the spectrum condition or has imaginary eigenvalues (which is one of the cases where the spectrum condition does not hold), then the corresponding $n=2$ solution is easily obtained by use of results in Section~\ref{sec:2vFL_special_pw_n=1} and an application of the superposition rule,  formulated in Remark~\ref{rem:vFL_superposition}. 
Next we address another case, where $\Gamma$ is diagonal and the spectrum condition is violated, and then the case, where $\Gamma$ is a Jordan block.
 
\begin{example}
	\label{ex:2vFL_spw_anticonjugate pair}
Let $\Gamma = \mathrm{diag}(\gamma,-\gamma^\ast)$ with $\mathrm{Re}(\gamma) \neq 0$. Then $R = \mathrm{diag}(\pm r,\pm r^\ast)$ with
\bez
    r = \frac{1}{2} \sqrt{\gamma^{-2}-2 \imag \, \alpha^2 \beta \, \gamma^{-1} - \alpha^2} \, .
\eez
Although the matrix equation for $R$ thus has four different matrix roots, it is plausible that already two of them are sufficient to obtain the complete solution of the linear system. Assuming $r \neq 0$, this is 
\bez
    \chi_1 = e^{\frac{\imag}{2} \varphi} \left( \begin{array}{c} c_{11} \, e^\psi + c_{12} \, e^{-\psi} \\
    c_{21} \, e^{\psi^\ast} + c_{22} \, e^{-\psi^\ast} \end{array}\right) \, ,  \quad
    \chi_2 = e^{\imag \, \varphi} \left( \begin{array}{c}  c_{01} \, e^{\frac{1}{2} (\gamma^{-1} x + \gamma \, t)} \\ c_{02} \, e^{-\frac{1}{2} (\gamma^{\ast -1} x + \gamma^\ast \, t)} 
    \end{array}\right) \, ,
\eez
where
\bez
     \psi = r \, (x-\imag \, \alpha^{-1} \gamma \, t) \, .
\eez
We further obtain
\bez
   \Xi = \frac{\|A\|}{2} \, e^{\frac{\imag}{2} \varphi} \left( \begin{array}{c} 
   	c_{11} \, (2 r - \gamma^{-1} -\imag \, \alpha) \, e^\psi 
   	- c_{12} \, (2 r + \gamma^{-1} +\imag \, \alpha) \, e^{-\psi} \\
   	c_{21} \, (2 r^\ast + \gamma^{\ast -1} - \imag \, \alpha ) \, e^{\psi^\ast} - c_{22} \, (2 r^\ast - \gamma^{\ast -1} + \imag \, \alpha ) \, e^{-\psi^\ast} 
   \end{array}\right) \, .
\eez
The Lyapunov equation determines the diagonal entries of $\Omega$,
\bez
    \Omega_{11} &=& \frac{1}{2 \mathrm{Re}(\gamma)} \Big( \frac{\imag}{\alpha^2 \|A\|^4} \Xi \Xi^\dagger \, \Gamma^\dagger
     + \chi_1 \chi_1^\dagger + \chi_2 \chi_2^\dagger \Big)_{11}  \\
     &=& \frac{1}{2\mathrm{Re}(\gamma)}\Big(\big(\frac{\imag \, \gamma^\ast}{4\alpha^2 \|A\|^2} |2r-\gamma^{-1}-\imag \, \alpha|^2+1\big)\,|c_{11}|^2 e^{2\mathrm{Re}(\psi)}\\
     &&-\big(\frac{\imag \, \gamma^\ast}{4\alpha^2\|A\|^2}(2r-\gamma^{-1}-\imag \, \alpha)\,(2r^\ast+\gamma^{\ast -1}-\imag \, \alpha)-1 \big)\,c_{11}c_{12}^\ast e^{2\imag \, \mathrm{Im}(\psi)}\\
     && -\big(\frac{\imag \, \gamma^\ast}{4\alpha^2 \|A\|^2}(2r+\gamma^{-1}+\imag \, \alpha) \, (2r^\ast-\gamma^{\ast -1} +\imag \, \alpha)-1 \big) \, c_{11}^\ast c_{12} e^{-2\imag \, \mathrm{Im}(\psi)} \\
     &&+\big(\frac{\imag \, \gamma^\ast}{4\alpha^2 \|A\|^2} |2r+\gamma^{-1}+\imag \, \alpha|^2+1 \big) \, |c_{12}|^2 e^{-2\mathrm{Re}(\psi)} + |c_{01}|^2 e^{\mathrm{Re}(\gamma^{-1}x+\gamma t)} \Big) \, , \\
   \Omega_{22} &=& -\frac{1}{2 \mathrm{Re}(\gamma)} \Big( \frac{\imag}{\alpha^2 \|A\|^4} \Xi \Xi^\dagger \, \Gamma^\dagger
   + \chi_1 \chi_1^\dagger + \chi_2 \chi_2^\dagger \Big)_{22} \\
   &=& \frac{1}{2\mathrm{Re}(\gamma)}\Big(\frac{\imag \, \gamma}{4\alpha^2\|A\|^2}\big(|2r+\gamma^{-1}+\imag \, \alpha|^2-1 \big) \, |c_{21}|^2 e^{2\mathrm{Re}(\psi)} \\
   &&-\big(\frac{\imag \, \gamma}{4\alpha^2\|A\|^2}(2r+\gamma^{-1}+\imag \, \alpha) \, (2r^\ast-\gamma^{\ast -1}+\imag \, \alpha)+1\big)\,c_{21}^\ast c_{22}e^{2\imag \, \mathrm{Im}(\psi)} \\
   &&-\big(\frac{\imag \, \gamma}{4\alpha^2 \|A\|^2}(2r-\gamma^{-1}-\imag \, \alpha) \, (2r^\ast+\gamma^{\ast -1}-\imag \, \alpha)+1\big) \, c_{21}c_{22}^\ast e^{-2\imag \, \mathrm{Im}(\psi)} \\
   &&  +\frac{\imag \, \gamma}{4\alpha^2 \|A\|^2} \big(|2r-\gamma^{-1}-\imag \, \alpha|^2-1\big) \, |c_{22}|^2 e^{-2\mathrm{Re}(\psi)} - |c_{02}|^2 e^{-\mathrm{Re}(\gamma^{-1} x + \gamma t)} \Big) \, ,   
\eez
and imposes the constraints
\bez
    \Big( \frac{\imag}{\alpha^2 \|A\|^4} \Xi \Xi^\dagger \, \Gamma^\dagger
    + \chi_1 \chi_1^\dagger + \chi_2 \chi_2^\dagger \Big)_{ij} = 0 \qquad \mbox{for} \quad i \neq j \, .
\eez
The latter amount to
\bez
(c_{11}c_{22}^\ast+c_{12}c_{21}^\ast) \, \big( 4\alpha^2\|A\|^2\gamma+\imag \, (1+\imag \, \alpha \gamma)^2 \big) - 2 \imag \, \gamma\, r \, (c_{11}c_{22}^\ast - c_{12}c_{21}^\ast) \,  (1+\imag \, \alpha \gamma)+2c_{01} c_{02}^\ast\, \alpha^2 \|A\|^2 \gamma =0 \, .
\eez
The differential equations for $\Omega$ lead to
\bez
\Omega_{12x} &=& \frac{\imag}{\alpha^2 \|A\|^4} \big( \Xi_{1x} - \frac{1}{2}\gamma^{-1} \Xi_1 \big) \, \Xi_2^\ast \, , \qquad
\Omega_{12t} = - \frac{\imag}{\alpha^2 \|A\|^4}\big( \alpha \|A\|^3 \chi_{12}^\ast - \gamma \, \Xi_2^\ast \big) \, \Xi_1 \, , \\
\Omega_{21x} &=& \frac{\imag}{\alpha^2 \|A\|^4} \big( \Xi_{2x} + \frac{1}{2}\gamma^{\ast -1} \Xi_2 \big) \, \Xi_1^\ast \, , \qquad \Omega_{21t} = -\frac{\imag}{\alpha^2 \|A\|^4} \big(\alpha \|A\|^3 \chi_{11}^\ast + \gamma^\ast \Xi_1^\ast \big) \, \Xi_2 \, ,
\eez
which leads to
\bez
\Omega_{12} &=& \frac{\imag}{16 \alpha^2 \|A\|^2 \gamma^3 \, r} \Big( c_{11} c_{21}^\ast \big( 1+\gamma \, (\imag \, \alpha+2r) \big) \, \big( 1+\gamma^2 \, (\alpha^2+4r^2) - 4 \gamma \,r \big) \, e^{2\psi} \\
&& - c_{12} c_{22}^\ast \big( 1+\gamma \, (\imag \, \alpha -2r) \big) \, \big( 1 + \gamma^2 (\alpha^2 + 4r^2) + 4\gamma \, r \big) \, e^{-2\psi} \Big) \\
&& +\frac{\imag \, x}{8\alpha^2 \|A\|^2 \gamma^3} \Big( c_{12} c_{21}^\ast \,(1+2\gamma r+\imag\,\alpha \gamma)^2\,(1+2\gamma r-\imag\,\alpha \gamma) \\
&& + c_{11} c_{22}^\ast \, (1-2\gamma r+\imag\,\alpha \gamma)^2\,(1-2\gamma r-\imag\,\alpha \gamma) \Big) \\
&& + \frac{\imag \, t}{4\alpha^2 \|A\|^2\gamma} \Big( 4\gamma\, r\,(c_{11}c_{22}^\ast-c_{12}c_{21}^\ast)\,\big(\alpha\,(\imag\,\gamma-\|A\|^2)+1\big)\\
&&-(c_{11}c_{22}^\ast+c_{12}c_{21}^\ast)\,\big(1-\gamma^2\,(\alpha^2-4r^2)-2\alpha\|A\|^2(\imag\,\alpha\gamma+1)+2\imag\,\alpha\gamma\big) \Big) +\tilde{c}_{12} \, , \\
\Omega_{21} &=& -\frac{\imag}{16\alpha^2 \|A\|^2 \gamma^{\ast 3} \, r^\ast} \Big(c_{11}^\ast  c_{21} \big( 1-\gamma^\ast (\imag \, \alpha+2r^\ast) \big) \, \big( 1+\gamma^{\ast 2} \, (\alpha^2+4r^{\ast 2}) + 4 \gamma^\ast r^\ast \big) \, e^{2\psi^\ast} \\
&& -c_{12}^\ast c_{22} \big( 1-\gamma^\ast (\imag \, \alpha -2r^\ast) \big) \, \big( 1+\gamma^{\ast 2} (\alpha^2+4r^{\ast 2}) - 4\gamma^\ast r^\ast \big) \, e^{-2\psi^\ast} \Big) \\
&& - \frac{\imag \, x}{8\alpha^2 \|A\|^2\gamma^{\ast 3}} \Big(c_{11}^\ast c_{22}\,(1-2\gamma^\ast r^\ast-\imag\,\alpha\gamma^\ast)^2\,(1-2\gamma^\ast r^\ast+\imag\,\alpha\gamma^\ast)\\
&&+c_{12}^\ast c_{21}\,(1+2\gamma^\ast r^\ast-\imag\,\alpha\gamma^\ast)^2\,(1+2\gamma^\ast r^\ast+\imag\,\alpha\gamma^\ast) \Big) \\
&& + \frac{\imag \, t}{4\alpha^2 \|A\|^2\gamma^\ast} \Big( 4\gamma^\ast r^\ast (c_{11}^\ast c_{22}-c_{12}^\ast c_{21}) \, \big(\alpha \, (\imag \, \gamma^\ast + \|A\|^2)-1 \big) \\
&& +(c_{11}^\ast c_{22}+c_{12}^\ast c_{21}) \, \big( 1-\gamma^{\ast 2}(\alpha^2-4r^{\ast 2}) + 2\alpha \|A\|^2 (\imag \, \alpha \gamma^\ast -1) - 2\imag \, \alpha \gamma^\ast \big) \Big) + \tilde{c}_{21} \, ,
\eez
with complex constants $\tilde{c}_{12}$ and $\tilde{c}_{21}$. The remaining condition (\ref{mFLeq_key}) becomes
\bez
\gamma \, (\tilde{c}_{12}-\tilde{c}_{21}^\ast) = c_{01} c_{02}^\ast + c_{11} c_{22}^\ast + c_{12} c_{21}^\ast \, .
\eez 
The resulting solution of the two-component FL equation is obtained by inserting the expressions for $\chi_1$, $\chi_2$, $\Xi$ and $\Omega$ in  (\ref{special_pw_u'}). In addition to the exponential dependence of the solution via $e^\psi$ and its complex conjugate, there is also a rational 
dependence on the variables $x$ and $t$.
Fig.~\ref{fig:2vFL_spw_anticonjugate_Im(gamma)=0} displays an example 
from this class.
\begin{figure}[h]
	\begin{center}
	\includegraphics[scale=.4]{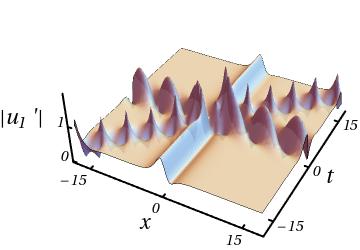} 
	\hspace{1cm}
	\includegraphics[scale=.4]{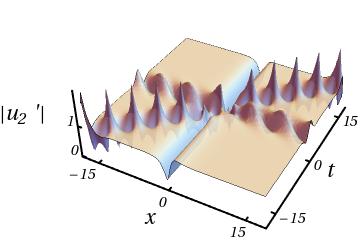} 
	\parbox{15cm}{
	\caption{Plots of the absolute values of the two components of a solution from the class in Example~\ref{ex:2vFL_spw_anticonjugate pair}, showing 
    a bright, respectively dark soliton, a breather and a 
    periodic rogue wave. 
    We chose $\alpha=1$, $A_1=1/2$, $A_2 = \sqrt{3}/2$, $\gamma=1$ (so that $r=-1/2+\imag/2$), $c_{01}=c_{02}=c_{12}=c_{21}=c_{22}=\tilde{c}_{12}=1$, $c_{11}=\tilde{c}_{21}=-1/2+3\imag/2$.
				\label{fig:2vFL_spw_anticonjugate_Im(gamma)=0} }
		}
	\end{center} 
\end{figure} 	
\hfill $\Box$
\end{example}

\begin{example}
Let $\Gamma$ be the $2 \times 2$ lower triangular Jordan block with eigenvalue $\gamma$. Then
\be
   R = \pm \left( \begin{array}{cc} r & 0 \\
   	-(1-\imag \, \alpha^2 \beta \gamma) (4 \gamma^3 \, r)^{-1} & r \end{array} \right) \, , \qquad r = \frac{1}{2\gamma} \sqrt{1 - \alpha^2 \gamma \, (\gamma + 2 \imag \, \beta)} \, .
   	   \label{2vFL_spw_n=2_Jordan_R}
\ee	
Let $r \neq 0$. We obtain
\bez
 \chi_1 &=& e^{\frac{1}{2} \imag \, \varphi} \left(\begin{array}{c}  
 c_{11} \, e^\psi + c_{21} \, e^{-\psi} \\  
(4 \gamma^3 r)^{-1} \Big( \big( -c_{11} \, \rho + 4 c_{12} \gamma^3 r \big) \, e^\psi + \big( c_{21} \, \rho + 4 c_{22} \gamma^3 r \big) \, e^{-\psi} \Big) \end{array} \right) \, , \\
 \chi_2 &=& e^{\imag \, \varphi} \, e^{\frac{1}{2} (\gamma^{-1} x + \gamma \, t)} \left(\begin{array}{c} c_{01} \\
 c_{02}	- \frac{1}{2} c_{01} \gamma^{-1} (\gamma^{-1} x - \gamma \, t) \, 
 	\end{array} \right) \, , \\
 \Xi &=& \frac{\|A\|}{8 \gamma^4 r} \, e^{\frac{\imag}{2} \varphi} \left(\begin{array}{c}
 	- 4 \gamma^3 r \, \big( c_{11} \, b_- \, e^\psi 
 	+ c_{21} \, b_+ \, e^{-\psi} \big) \\ 
 	\big( c_{11} ( b_- \, \rho - 2 \gamma \, d_- )  
 	- 4 c_{12} b_- \gamma^3 r \big) \, e^\psi 
 	- \big( c_{21} ( b_+ \, \rho - 2 \gamma \, d_+ )
 	+ 4 c_{22} b_+ \gamma^3 r \big) \, e^{-\psi} 
 	\end{array} \right) \, ,	
\eez
where 
\bez
  && \rho = (1- \imag \, \alpha^2 \beta \gamma) \, x + \alpha \gamma^2 (  \beta - \imag \, \gamma ) \, t \, , 
   \qquad
   \psi = r \, \big( x - \imag \, \frac{\gamma}{\alpha} t \big) \, , \\
  && b_\pm = 1 \pm 2 \gamma \, r + \imag \, \alpha \, \gamma \, , \qquad
   d_\pm = 1 \pm 2 \gamma \, r - \imag \, \alpha^2 \beta \, \gamma \, .
\eez
It is now straightforward to elaborate the equations determining $\Omega$ and then the generated solution (\ref{special_pw_u'}) of the two-component vector FL equation. If $\mathrm{Re}(\gamma) \neq 0$, then $\Omega$ is completely determined by the Lyapunov equation and given by Example~\ref{ex:n=2Jordan_Lyapunov_sol}. If $\mathrm{Re}(\gamma) = 0$, we have to employ the differential equations for $\Omega$ in addition, a more tedious task. We skip the details.
\hfill $\Box$
\end{example}

\begin{example}
Again, let $\Gamma$ be the $2 \times 2$ lower triangular Jordan block with eigenvalue $\gamma$, so that $R$ and $r$ are given 
by (\ref{2vFL_spw_n=2_Jordan_R}). But now we consider 
the double root case $r=0$. Then the expression on the right hand side of (\ref{special_quadratic_R_eq}) is nilpotent and does not possess a square root, and $\gamma$ is fixed by (\ref{special_gamma_for_r=0}). The general solution of the linear system is now given by $\tilde{\chi}_i = (\tilde{\chi}_{i1},\tilde{\chi}_{i2})^T$, 
$i=1,2$, with
\bez
  \tilde{\chi}_{11} &=& a_1 \, \rho + a_0 
  \qquad \mbox{where} \quad \rho = x- \frac{\imag}{2 \gamma} \, t \, , \\
  \tilde{\chi}_{12} &=& 
   - \frac{(1 - \imag \, \alpha^2 \beta \gamma)}{12 \gamma^3} ( a_1 \, 
    \rho^3 + 3 a_0 \, \rho^2 ) + c_2 \rho + c_1 + \frac{6 \, \imag \, a_1}{12 \gamma^2} \, t \, , \\
  \tilde{\chi}_{21} &=&  b_1 \, e^{\frac{1}{2} (\gamma^{-1} x + \gamma \, t)} \, , \qquad  
  \tilde{\chi}_{22} = \big( b_2 - \frac{1}{2} b_1 ( \gamma^{-2} x - t)  \big) \, e^{\frac{1}{2} (\gamma^{-1} x + \gamma \, t)}  \, , 
\eez
with constants $a_0,a_1,b_1,b_2,c_1,c_2$. Then we obtain the following components of $\Xi$,
\bez
  \Xi_1 &=& - \frac{\|A\|}{2 \gamma} \, e^{\imag \, \varphi} \, \Big(  
  a_1 \, (\rho - 2 \gamma) + a_0 \Big) \, ,  \\
  \Xi_2 &=& \|A\| \, e^{\imag \, \varphi} \, \Big(  
    \frac{a_1}{24 \, \gamma^4} (1 - \imag \, \alpha^2 \beta \gamma) \, \rho^3 
    + \frac{a_0 - 2 a_1 \gamma}{8 \gamma^4} (1-\imag \, \alpha^2 \beta \gamma) \, \rho^2 \\
   && - \frac{1}{2 \gamma^3} \big( a_0 (1- \imag \, \alpha^2 \beta \gamma) 
    - a_1 \gamma + c_2 \gamma^2 \big) \, \rho 
    + \frac{1}{4 \gamma^3} \big( -\imag \, a_1 \, t + 2 a_0 \gamma 
    - 2 c_1 \gamma^2 + 4 c_2 \gamma^3 \big) \, . 
\eez
If $\alpha^{-1} > \|A\|^2$, then $\mathrm{Re}(\gamma) \neq 0$ (cf. Example~\ref{ex:special_pw_double_root}), and $\Omega$ 
is completely determined by the Lyapunov equation alone,
\bez
   \Omega_{11} &=& \frac{1}{2 \, \mathrm{Re}(\gamma)} \Big(
   	 \frac{\gamma^\ast}{\alpha^2 \|A\|^4} |\Xi_1|^2 
   	 +|\tilde{\chi}_{11}|^2 + |\tilde{\chi}_{21}|^2 \Big) \, , \\
   \Omega_{12} &=& \frac{1}{4 \, \mathrm{Re}(\gamma)^2 \alpha^2 \|A\|^4} \Big( \gamma \, |\Xi_1|^2 
    + 2 \, \mathrm{Re}(\gamma) \, \gamma^\ast \, \Xi_1 \, \Xi_2^\ast \\
   && - \alpha^2 \|A\|^4 \big( |\tilde{\chi}_{11}|^2 + |\tilde{\chi}_{21}|^2 
    - 2 \, \mathrm{Re}(\gamma) \, ( \tilde{\chi}_{11} \tilde{\chi}_{12}^\ast + \tilde{\chi}_{21} \tilde{\chi}_{22}^\ast ) \big) \Big) \, , \\
   \Omega_{21} &=& \frac{1}{4 \, \mathrm{Re}(\gamma)^2 \alpha^2 \|A\|^4} \Big( -\gamma^\ast \, |\Xi_1|^2 + 2 \, \mathrm{Re}(\gamma) \, \gamma^\ast \, \Xi_1^\ast \, \Xi_2 \\
   && - \alpha^2 \|A\|^4 \big( |\tilde{\chi}_{11}|^2 + |\tilde{\chi}_{21}|^2 - 2 \, \mathrm{Re}(\gamma) \, ( \tilde{\chi}_{11}^\ast \tilde{\chi}_{12} + \tilde{\chi}_{21}^\ast \tilde{\chi}_{22}) \big) \Big) \, , \\
   \Omega_{22} &=& \frac{1}{8 \, \mathrm{Re}(\gamma)^3 \alpha^2 \|A\|^4} \Big( - 2 \imag \, \mathrm{Im}(\gamma) \, |\Xi_1|^2  
   + 4 \mathrm{Re}(\gamma)^2 \, \gamma^\ast |\Xi_2|^2
   + 4 \imag \, \mathrm{Re}(\gamma) \, \mathrm{Im}(\gamma \, \Xi_1^\ast \, \Xi_2) \\
  && + 2 \alpha^2 \|A\|^4 \big( |\tilde{\chi}_{11}|^2 + |\tilde{\chi}_{21}|^2 
   + 2 \mathrm{Re}(\gamma)^2 \, (|\tilde{\chi}_{12}|^2 + |\tilde{\chi}_{22}|^2 ) - 
  2 \, \mathrm{Re}(\gamma) \, \mathrm{Re}(\tilde{\chi}_{11} \tilde{\chi}_{12}^\ast + \tilde{\chi}_{21} \tilde{\chi}_{22}^\ast) 
    \big)  \Big) \, .   	    
\eez
Now (\ref{special_pw_u'}) determines a solution of the two-component 
vector FL equation, which is a second order version of the solution 
in Example~\ref{ex:special_pw_double_root}. It has already been derived in a different way in \cite{Xu+Chen19}.  
This solution is quasi-rational if and only if $\tilde{\chi}_2=0$, in which case it is equivalent to a solution of the scalar FL equation.  
\hfill $\Box$
\end{example}

\begin{remark}
The solution set of a polynomial matrix equation, like (\ref{special_quadratic_R_eq}), is in general more intricate than for polynomial equations over a field. Here we take a look at an example in our context. Let $\Gamma = \mathrm{diag}(\gamma_1,\gamma_2)$, 	
where $\gamma_1,\gamma_2$ have the form (\ref{special_gamma_for_r=0}). In this case, 
\bez
	R_1 = \left( \begin{array}{cc} 0 & z_1 \\
		0 & 0 \end{array} \right) \, , \qquad
	R_2 = \left( \begin{array}{cc} 0 & 0 \\
		z_2 & 0 \end{array} \right) \, ,	
\eez
solve (\ref{special_quadratic_R_eq}), which is now $R^2 = 0$, for arbitrary $z_1,z_2 \in \mathbb{C}$.
Assuming $z_1 z_2 \neq 0$, (\ref{special_pw_sol_chi1}) yields
\bez
	\chi_1 = e^{\frac{\imag}{2} \, \varphi} \left( \begin{array}{c} c_{11} + c_{12} \, (x- \frac{\imag}{\alpha} \gamma_1 t) \\
		c_{21} + c_{22} \, (x- \frac{\imag}{\alpha} \gamma_2 t) 	 \end{array} \right) \, .
\eez
Here $z_1$ and $z_2$ have been absorbed by redefinition of the 
constant 2-vectors $c_\kappa$. It is thus obvious that, in this way, we get directly a solution describing a (nonlinear) superposition of two solitons of the kind given in  Example~\ref{ex:special_pw_double_root}, without having to reconsider the linear system in order to find missing solutions.
\hfill $\Box$
\end{remark}

\subsection{The generic plane wave case $\alpha_1 \neq \alpha_2$}
\label{subsec:alpha1_neq_alpha2}
If $\alpha_1 \neq \alpha_2$ 
and if the matrix in the last brackets in (\ref{alpha-branch_eq}) is invertible, the latter equation can be solved for $\chi_2$, 
\bez
\chi_2 &=& \frac{\|A\|^2}{(\alpha_2-\alpha_1) \, A_1 A_2} 
 \Big( \imag \, (1 - A^\dagger \boldsymbol{\alpha} A) \, \Gamma^{-1} - \frac{\alpha_1 \alpha_2 \|A\|^2}{A^\dagger \boldsymbol{\alpha} A} I_n \Big)^{-1} \Big( \chi_{1xx} - \imag \, \frac{A^\dagger \boldsymbol{\alpha}^2 A}{A^\dagger \boldsymbol{\alpha} A} \, \chi_{1x} - \frac{1}{4} \Gamma^{-2} \chi_1  \\
 && + \frac{\imag}{2 \|A\|^2 \, A^\dagger \boldsymbol{\alpha} A} 
 \big( (\alpha_2-\alpha_1)^2 |A_1|^2 |A_2|^2 
 - (A^\dagger \boldsymbol{\alpha} A)^2 (1-2 A^\dagger \boldsymbol{\alpha} A) \big) \, \Gamma^{-1} \chi_1 \Big) \, .   
\eez
Inserted in (\ref{2vFL_chi_2x}), this yields a third order ODE for 
$\chi_1$,
\bez 
&& \chi_{1xxx} - \big( \imag \, (\alpha_1+\alpha_2) I_n + \frac{1}{2} \Gamma^{-1} \big) \, \chi_{1xx} - \big( \frac{1}{4} \Gamma^{-2} - \imag \, A^\dagger \boldsymbol{\alpha}^2 A \, \Gamma^{-1} + \alpha_1 \alpha_2 I_n \big) \, \chi_{1x} \\
&& + \Big( \frac{1}{8} \Gamma^{-3} + \frac{1}{4} \imag \, \big( \alpha_1+\alpha_2 - 2 A^\dagger \boldsymbol{\alpha}^2 A \big) \Gamma^{-2} + \alpha_1 \alpha_2 (A^\dagger \boldsymbol{\alpha} A - \frac{1}{2} \big) \Gamma^{-1} \Big) \, \chi_1 = 0 \, .
\eez
(\ref{vFL_chi_t}) reads
\bez
\chi_{1t} &=& - (\imag \, \| A\|^2I_n - \frac{1}{2} \Gamma) \, \chi_1 
+ \frac{\imag}{\|A\|^2} \big( A^\dagger \boldsymbol{\beta} \, A \, \chi_1 + (\beta_2-\beta_1) \, A_1 A_2 \, \chi_2 \big) 
- \imag \, \frac{\|A\|^2}{ A^\dagger \boldsymbol{\alpha}  A} \Big( 
\Gamma \,\chi_{1x}  \\
&&  - \frac{\imag}{\|A\|^2} \Gamma \big( A^\dagger \boldsymbol{\alpha} \, A \, \chi_1 + (\alpha_2-\alpha_1) \, A_1 A_2 \, \chi_2 \big) - \frac{1}{2} \chi_1 \Big)  \, , \\
 \chi_{2t} &=& \frac{1}{2} \Gamma \, \chi_2 + \frac{\imag}{\|A\|^2} \big( (\beta_2-\beta_1) \, A_1^\ast A_2^\ast \, \chi_1 + (\beta_1 |A_2|^2 + \beta_2 |A_1|^2) \, \chi_2 \big) \, ,
\eez
where now
\bez
\beta_1 = 2 |A_1|^2 + (1+\alpha_2/\alpha_1) |A_2|^2 - \frac{1}{\alpha_1} \, , \qquad
\beta_2 = (1+\alpha_1/\alpha_2) |A_1|^2 + 2 |A_2|^2 - \frac{1}{\alpha_2} \, . 
\eez
After elimination of $\chi_2$ in the equation for $\chi_{1t}$, we get
\bez
\alpha_1 \alpha_2 \, \chi_{1t} = \Gamma \chi_{1xx} - \imag \, (\alpha_1+\alpha_2) \Gamma \chi_{1x} - \big( \frac{1}{2} \alpha_1 \alpha_2 \, \Gamma + \frac{1}{4} \Gamma^{-1} - \imag \, (\alpha_1+\alpha_2)(A^\dagger \boldsymbol{\alpha} A - \frac{1}{2}) I_n \big) \, \chi_1 \, ,
\eez
which determines the $t$-dependence of $\chi_1$. 
By consistency, the equation for $\chi_{2t}$ should be identically satisfied as a consequence of the others. This can indeed be verified by a tedious computation. 
The linear system is now (in special cases not necessarily completely) solved by
\be
\chi_1 &=& \sum_\kappa e^{R_\kappa x + S_\kappa t} c_\kappa \, , 
    \label{2vFL_gen_sol_chi1}  \\
\chi_2 &=& \frac{\|A\|^2}{(\alpha_2-\alpha_1) \, A_1 A_2} 
\Big( \imag \, (1 - A^\dagger \boldsymbol{\alpha} A) \, \Gamma^{-1} - \frac{\alpha_1 \alpha_2 \|A\|^2}{A^\dagger \boldsymbol{\alpha} A} I_n \Big)^{-1} \sum_\kappa \Big( R_\kappa^2-  \frac{\imag \, A^\dagger \boldsymbol{\alpha}^2 A}{A^\dagger \boldsymbol{\alpha} A} \, R_\kappa - \frac{1}{4} \Gamma^{-2}  \nonumber  \\
&& + \frac{\imag}{2 \|A\|^2 \, A^\dagger \boldsymbol{\alpha} A} 
\big( (\alpha_2-\alpha_1)^2 |A_1|^2 |A_2|^2 
- (A^\dagger \boldsymbol{\alpha} A)^2 (1-2 A^\dagger \boldsymbol{\alpha} A) \big) \, \Gamma^{-1}   \Big)\,e^{R_\kappa x+S_\kappa t} c_\kappa \, ,  \label{2vFL_gen_sol_chi2}  
\ee
where $\kappa$ runs over a set numbering different roots $R_\kappa$, 
commuting with $\Gamma$, of the cubic $n \times n$ matrix equation
\bez
&& R^3 - \big( \imag \, (\alpha_1+\alpha_2) I_n + \frac{1}{2} \Gamma^{-1} \big) R^2 - \big( \frac{1}{4} \Gamma^{-2} - \imag \, A^\dagger \boldsymbol{\alpha}^2 A \, \Gamma^{-1} + \alpha_1 \alpha_2  I_n \big) R \nonumber \\
&& + \frac{1}{8} \Gamma^{-3} + \frac{1}{4} \imag \, \big( \alpha_1+\alpha_2-2A^\dagger \boldsymbol{\alpha}^2 A \big) \Gamma^{-2} + \alpha_1 \alpha_2 \, (A^\dagger \boldsymbol{\alpha} A - \frac{1}{2} \big) \Gamma^{-1} = 0 \, .
\eez
Further, $\{c_\kappa\}$ are constant $n$-component column vectors and
\bez
S_\kappa = \frac{1}{\alpha_1 \alpha_2} \Big( \Gamma R_\kappa^2 - \imag \, (\alpha_1+\alpha_2) \Gamma R_\kappa - \frac{1}{2} \alpha_1 \alpha_2 \Gamma - \frac{1}{4} \Gamma^{-1} + \imag \, (\alpha_1+\alpha_2)(A^\dagger \boldsymbol{\alpha} A - \frac{1}{2}) I_n \Big) \, .
\eez
It remains to determine $\Omega$. Then, finally, 
\be
u' &=& \left( \begin{array}{cc} A_1 \, e^{\imag \, \varphi_1} & - A_2^\ast \, e^{\imag \, \varphi_1} \\
	A_2 \, e^{\imag \, \varphi_2}  & A_1^\ast \, e^{\imag \, \varphi_2}  \end{array} \right) \left[ \left( \begin{array}{c}    
	1  \\ 0   \end{array} \right)  
+ \frac{1 }{\|A\| \, A^\dagger \boldsymbol{\alpha} A} \, 
\left( \begin{array}{c}   
	\chi_1^\dagger \, \Omega^{-1} \Xi  \\ 
	\chi_{2}^\dagger \, \Omega^{-1} \Xi \end{array} \right) \right]  \nonumber \\
 &=& \left( \begin{array}{cc} e^{\imag \, \varphi_1} 
 	&  0  \\  0  & e^{\imag \, \varphi_2}  \end{array} \right)
 \left( \begin{array}{cc} A_1 + (\|A\| \, A^\dagger \boldsymbol{\alpha} A)^{-1} (A_1 \chi_1^\dagger - A_2^\ast \chi_2^\dagger) \, \Omega^{-1} \Xi  \\[4pt]
   A_2 + (\|A\| \, A^\dagger \boldsymbol{\alpha} A)^{-1} (A_2 \chi_1^\dagger + A_1^\ast \chi_2^\dagger) \, \Omega^{-1} \Xi \end{array} \right) \, ,    \label{2vFL_sol}
\ee
with $\Xi$ given by (\ref{2vFL_Xi}), 
is a solution of the two-component vector FL equation.
\vspace{.3cm}

\noindent
Again, we have to distinguish the following two cases.
\vspace{.2cm}

\noindent
\textbf{(1)} 
If the spectrum condition for $\Gamma$ holds, then the Lyapunov equation (\ref{pw_vFLeq_Lyap}) possesses a unique solution, irrespective of the 
$n \times n$ matrix on its right hand side. In particular, 
if $\Gamma = \mathrm{diag}(\gamma_1,\ldots,\gamma_n)$, with 
$\gamma_i \neq -\gamma_j^\ast$, $i,j=1,\ldots,n$, 
the solution $\Omega$ is given by the generalized Cauchy matrix with components 
\bez
\Omega_{ij} = \frac{\imag \, \Xi_i \, \Xi_j^\ast \, \gamma_j^\ast + (A^\dagger  \boldsymbol{\alpha} A)^2 \big( \chi_{1i} \chi_{1j}^\ast + \chi_{2i} \chi_{2j}^\ast \big)}{(A^\dagger \boldsymbol{\alpha} A)^2 \, (\gamma_i + \gamma_j^\ast)} 
\, ,
\eez
where
\bez
\Xi_i = \|A\| \sum_\kappa 
\Big( r_{i \kappa} - \frac{1}{2} \gamma_i^{-1} - \frac{\imag \, A^\dagger \boldsymbol{\alpha} A}{\|A\|^2} \Big) \, e^{r_{i\kappa} x+s_{i\kappa} t} \, c_{\kappa i}
- \frac{\imag \, (\alpha_2-\alpha_1) A_1 A_2}{\|A\|} \chi_{2i} \, ,
\eez
writing $R_\kappa = \mathrm{diag}(r_{1\kappa},\ldots,r_{n\kappa})$ 
and $S_\kappa = \mathrm{diag}(s_{1\kappa},\ldots,s_{n\kappa})$.
\vspace{.2cm}

\noindent
\textbf{(2)} 
If the spectrum condition does \emph{not} hold, the Lyapunov equation (\ref{pw_vFLeq_Lyap}) constrains its right hand side. Moreover, it does not 
(fully) determine $\Omega$. In this case, we have to take (\ref{vFLeq_Om_deriv}) into account, i.e.,
\bez
\Omega_x &=& \frac{\imag}{(A^\dagger \boldsymbol{\alpha} A)^2} \, \Xi_x \, 
 \Xi^\dagger + \frac{1}{2} \Gamma^{-1} \chi_1 \chi_1^\dagger \Gamma^{\dagger -1} \\
&& -\Big(\frac{A_1^\ast A_2^\ast}{A^\dagger \boldsymbol{\alpha} A}(\alpha_2-\alpha_1)\big(\chi_{1x}-\frac{1}{2}\Gamma^{-1}\chi_1\big)-\chi_{2x}
+\frac{\imag \,\|A\|^2}{A^\dagger \boldsymbol{\alpha} A} \alpha_1 \alpha_2 \chi_2 \Big) \chi_2^\dagger \, \Gamma^{\dagger -1} \, , \nonumber \\
\Omega_t &=& -\frac{\imag}{2 (A^\dagger \boldsymbol{\alpha}  A)^2} \, \Xi \,
\, \Xi^\dagger \, \Gamma^\dagger   
 -\frac{\imag \, \|A\|}{A^\dagger \boldsymbol{\alpha}  A} \, \Xi \, \chi_1^\dagger + \frac{1}{2} (\chi_1\chi_1^\dagger+\chi_2\chi_2^\dagger) \, ,  
\eez
and (\ref{vFLeq_key}), which is (\ref{2vFLeq_pw}).

\subsubsection{$n=1$}
\label{subsec:2vFL_generic_pw_n=1}
For $n=1$, we write $\Gamma=\gamma$ and $R_\kappa=r_\kappa$, $S_\kappa = s_\kappa$, so that 
\bez
\chi_1 &=& \sum_{\kappa} e^{r_\kappa x+s_\kappa t} \, c_\kappa \, , \\
\chi_2 &=& \frac{\gamma\,\|A\|^2 A^\dagger \boldsymbol{\alpha}  A}{(\alpha_2-\alpha_1) A_1 A_2 \big(\imag \, (1-A^\dagger \boldsymbol{\alpha} A) \, A^\dagger \boldsymbol{\alpha}  A - \gamma \, \alpha_1 \alpha_2 \|A\|^2 \big)} \sum_\kappa \Big( r_\kappa^2 - \frac{\imag \, A^\dagger \boldsymbol{\alpha}^2 A}{A^\dagger \boldsymbol{\alpha} A} \, r_\kappa - \frac{1}{4} \gamma^{-2} \\
&& + \frac{\imag}{2 \|A\|^2 \, A^\dagger \boldsymbol{\alpha} A} 
\big( (\alpha_2-\alpha_1)^2 |A_1|^2 |A_2|^2 
- (A^\dagger \boldsymbol{\alpha} A)^2 (1-2 A^\dagger \boldsymbol{\alpha} A) \big) \, \gamma^{-1} \Big) \, e^{r_\kappa x+s_\kappa t} \, c_\kappa \, ,
\eez
assuming $\gamma \neq \imag \, (1-A^\dagger \boldsymbol{\alpha}  A) \, A^\dagger \boldsymbol{\alpha}  A/(\alpha_1 \alpha_2 \|A\|^2)$. Here
$\{ r_\kappa \}$ is a set of different roots of the cubic 
characteristic equation
\be
&& r^3 - \big( \imag \, (\alpha_1+\alpha_2)   + \frac{1}{2} \gamma^{-1} \big) r^2 - \big( \frac{1}{4} \gamma^{-2} - \imag \, A^\dagger \boldsymbol{\alpha}^2 A \gamma^{-1} + \alpha_1 \alpha_2  \big) r \nonumber \\
&& + \frac{1}{8} \gamma^{-3} + \frac{1}{4} \imag \, \big( \alpha_1+\alpha_2-2A^\dagger \boldsymbol{\alpha}^2 A \big) \gamma^{-2} + \alpha_1 \alpha_2 \, (A^\dagger \boldsymbol{\alpha} A - \frac{1}{2} \big) \gamma^{-1} = 0 \, ,
    \label{2vFL_n=1_cubic_eq}
\ee
and 
\be
s_\kappa = \frac{1}{\alpha_1 \alpha_2} \Big( \gamma\, r_\kappa^2 - \imag \, (\alpha_1+\alpha_2) \gamma\, r_\kappa - \frac{1}{2} \alpha_1 \alpha_2 \gamma - \frac{1}{4} \gamma^{-1} + \imag \, (\alpha_1+\alpha_2)(A^\dagger \boldsymbol{\alpha} A - \frac{1}{2}) \Big) \, .   \label{2vFL_n=1_s_kappa}
\ee
If there are three different roots $r_\kappa$, $\kappa=1,2,3$, then the above solution of the linear system is 
the general solution. Otherwise, the missing solutions still have to be found, which is addressed in the subsections \ref{subsec:2vFL_triple_root} and \ref{subsec:2vFL_double_root}. 

The Darboux transformation yields the following solution of the two-component vector FL equation,
\be
 u'= 
\left( \begin{array}{cc} e^{\imag \, \varphi_1} \, \big[ A_1 + (\|A\| \, A^\dagger \boldsymbol{\alpha} A \, \Omega)^{-1} (A_1 \chi_1^\dagger - A_2^\ast \chi_2^\dagger) \,   \Xi \big] \\[4pt]
e^{\imag \, \varphi_2} \, \big[ A_2 + (\|A\| \, A^\dagger \boldsymbol{\alpha} A \, \Omega)^{-1} (A_2 \chi_1^\dagger + A_1^\ast \chi_2^\dagger) \,  \Xi \big] \end{array} \right) \, ,
    \label{alphas_different_n=1_solution}
\ee
where now
\bez
 \Xi = \|A\| \Big( \chi_{1x} - (\frac{1}{2} \gamma^{-1} +\imag \, \frac{ A^\dagger \boldsymbol{\alpha} A}{\|A\|^2}) \, \chi_1 \Big) 
  - \imag \, (\alpha_2-\alpha_1)  \frac{A_1 A_2}{\|A\|} \, \chi_2 \, .
\eez
It only remains to determine $\Omega$.

\paragraph{(1)}  Let $\mathrm{Re}(\gamma) \neq 0$. Then we get
\bez
\Omega = \frac{\imag \, \gamma^\ast \,|\Xi|^2 + (A^\dagger \boldsymbol{\alpha}  A)^2 (|\chi_1|^2+|\chi_2|^2)}{2 \mathrm{Re}(\gamma) \, (A^\dagger \boldsymbol{\alpha} A)^2} \, ,
\eez
which is everywhere non-zero, so that the generated solution $u'$ is non-singular on $\mathbb{R}^2$. 
Fig.~\ref{fig:2vFL_general_breather} shows plots of the absolute values of the two components of $u'$ for an example from 
this class of solutions of the two-component vector FL equation.

\begin{figure}[h]
	\begin{center}
 \includegraphics[scale=.4]{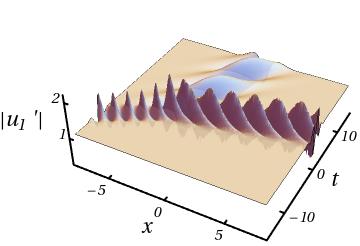} 
 \hspace{1cm}
	\includegraphics[scale=.4]{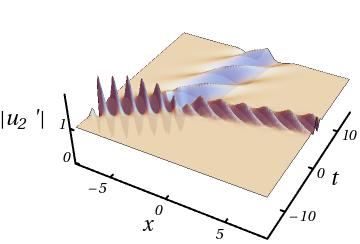} 
	\parbox{15cm}{
	\caption{Plots of the absolute value of the two components of a solution from the class in Section~\ref{subsec:2vFL_generic_pw_n=1} (1), showing a breather that splits, during evolution in $t$, into two breathers. Here we chose $\alpha_1=1$, $\alpha_2=2$, $A_1=\sqrt{6/5}$, $A_2=\frac{1}{2} \sqrt{41/10}$, $\gamma=1+\frac{1}{2}\imag$, $r_1=1-\imag$, $r_{2,3}=\frac{1}{10} (-3 + 19 \, \imag \pm \sqrt{-32-114 \, \imag})$, $c_1=c_2=1$, $c_3=10$. 
				\label{fig:2vFL_general_breather} } 
		}
	\end{center} 
\end{figure}

\paragraph{(2)}   Let $\mathrm{Re}(\gamma)=0$. We write $\gamma=-\imag \, k$ with real $k$. The Lyapunov equation then becomes the constraint
\bez
    (A^\dagger \boldsymbol{\alpha} A)^2 \, \big( |\chi_1|^2+|\chi_2|^2 \big) = k \,  |\Xi|^2 \, ,
\eez
which requires $k>0$. Since this has to hold for all $x$ and $t$, it further requires that only one of the $c_\kappa$ is nonzero. Hence 
\bez
\chi_1 &=&  e^{r  x+s t} \, c \, , \\
\chi_2 &=& -\frac{k\,\|A\|^2 A^\dagger \boldsymbol{\alpha}  A}{(\alpha_2-\alpha_1) A_1 A_2 \big( (1-A^\dagger \boldsymbol{\alpha}  A) \, A^\dagger \boldsymbol{\alpha} A + k \, \alpha_1 \alpha_2 \|A\|^2 \big)}  \Big( r^2 - \frac{\imag \, A^\dagger \boldsymbol{\alpha}^2 A}{A^\dagger \boldsymbol{\alpha} A} \, r + \frac{1}{4} k^{-2} \\
&& - \frac{1}{2 \|A\|^2 \, A^\dagger \boldsymbol{\alpha} A} 
\big( (\alpha_2-\alpha_1)^2 |A_1|^2 |A_2|^2 
-(A^\dagger \boldsymbol{\alpha} A)^2 (1-2 A^\dagger \boldsymbol{\alpha} A) \big) \, k^{-1} \Big) \, e^{r x+s t} \, c \, ,
\eez
where $r$ is any root of the characteristic equation (\ref{2vFL_n=1_cubic_eq}) (and $s$ the corresponding $s_\kappa$). 
 From (\ref{2vFL_n=1_cubic_eq}) and its complex conjugate, we get
\bez
&& |A_1|^2 =|1+2k\alpha_1+2\imag\,kr|^2 \, \frac{2\,\mathrm{Im}(r) - \alpha_2}{4\,k\,\alpha_1^2\,(\alpha_1-\alpha_2)} \, , \\
&&|A_2|^2 = -|1+2k\alpha_2+2\imag\,kr|^2 \, \frac{2\,\mathrm{Im}(r) - \alpha_1}{4\,k\,\alpha_2^2\,(\alpha_1-\alpha_2)} \, ,
\eez
which requires 
\be
 \frac{2\,\mathrm{Im}(r)-\alpha_2}{\alpha_1-\alpha_2} > 0 \, , \qquad \frac{2\,\mathrm{Im}(r)-\alpha_1}{\alpha_1-\alpha_2} < 0 \, .   \label{range}
\ee
By using these expressions for $|A_1|^2$ and $|A_2|^2$, a lengthy calculation reveals that the above constraint is automatically satisfied. 
Furthermore, we have
\bez
\Omega_x &=& -\frac{|c|^2 \big( \delta_1 - \imag \, k \, \mathrm{Re}(r) \big)}{k^2} \, \Big( 1 - \frac{\big( 4 \, \delta_2 \, \mathrm{Im}(r) + \alpha_1 \alpha_2 \big) \, \big( \delta_1^2 + k^2 \, \mathrm{Re}(r)^2 \big) }{  (2 \, \delta_1 \delta_2 + {k \, \alpha_1 \alpha_2})^2 + 4 \, k^2 \, \delta^2_2 \, \mathrm{Re}(r)^2 } \Big) \, e^{2 \theta} \, , \\ 
\Omega_t &=& \frac{2 \, k \, \delta_2}{\alpha_1 \alpha_2} \, \Omega_x \, ,
\eez
where
\bez
 \theta = \mathrm{Re}(r)\,\big( x+\frac{2\, k \,    \delta_2}{\alpha_1\alpha_2}\, t \big) \, , \qquad
 \delta_1 = k \, \mathrm{Im}(r)-\frac{1}{2} \, , \qquad
 \delta_2 = \mathrm{Im}(r) - \frac{1}{2}(\alpha_1+\alpha_2) \, .
\eez
\textbf{(i)} $\mathrm{Re}(r)=0$. Then $\theta=0$ and we obtain
\bez
\Omega &=& -\frac{|c|^2 \, \delta_1}{k^2} \Big( 1 - \frac{\delta_1^2 \, (4 \, \delta_2 \, \mathrm{Im}(r) + \alpha_1 \alpha_2)}{(2 \,\delta_1 \delta_2 + k \, \alpha_1 \alpha_2)^2} \Big) \, \Big( x + \frac{2 \, k \, \delta_2}{\alpha_1 \alpha_2} \, t + \tilde{c} \Big) \, ,
\eez 
with a constant $\tilde{c}$, for which (\ref{2vFLeq_pw}) requires
\bez
   \mathrm{Im}(\tilde{c}) = - \frac{k}{2 \delta_1} = \frac{k}{1 - 2 \, k \, \mathrm{Im}(r)} \, ,
\eez
assuming $\delta_1 \neq 0$. Since $\mathrm{Im}(\tilde{c}) \neq 0$, $\Omega$ vanishes nowhere, so that 
\bez
u_\mu' = A_\mu \, e^{\imag \, \varphi_\mu} \, \Big( 1 + \frac{\imag \, |c|^2}{(2 \, \delta_1 \delta_2 + k \, \alpha_1 \alpha_2) \, \Omega} \Big[ (-1)^{\mu} \, \frac{\delta_1 (k \, \alpha_1 - \delta_1) \, (k \, \alpha_2 - \delta_1) \,(4 \, \delta_2 \, \mathrm{Im}(r) + \alpha_1 \alpha_2)}{k \, |A_\mu|^2 (\alpha_1-\alpha_2) \, (2 \, \delta_1 \delta_2 + k \, \alpha_1 \alpha_2)} - \alpha_1 \alpha_2 \Big] \Big) \, ,
\eez
$\mu=1,2$, is non-singular on $\mathbb{R}^2$. Along the line $x + 2 \, k \, \delta_2 \, t/(\alpha_1\alpha_2) + \mathrm{Re}(\tilde{c}) = 0$, the absolute square of $u'_\mu$ is given by
\bez
    |u'_\mu|^2 = |A_\mu|^2 \, \big(1-b_\mu\big)^2 \, , \qquad \mu=1,2,
\eez
where
\bez
b_\mu &=& \frac{2}{(2\,\delta_1\delta_2+k\,\alpha_1\,\alpha_2)^2-\delta_1^2\,(4\,\delta_2\,\mathrm{Im}(r) +\alpha_1\alpha_2)} \Big(k\,\alpha_1\alpha_2\,(2\,\delta_1\delta_2+k\,\alpha_1 \alpha_2) \\
&&  -(-1)^\mu\,\frac{\delta_1\,(k\,\alpha_1-\delta_1)\,(k\,\alpha_2-\delta_1) \, (4\,\delta_2\,\mathrm{Im}(r)+\alpha_1 \alpha_2 )}{(\alpha_1-\alpha_2)\,|A_\mu|^2}\Big) \, .
\eez
If $0<b_\mu<2$, then $u'_\mu$ represents a dark soliton. 
If $b_\mu < 0$ or $b_\mu > 2$, it represents an anti-dark soliton. 
 \\
\textbf{(ii)} $\mathrm{Re}(r)\neq 0$. Then we have
\bez
\Omega &=& -\frac{|c|^2 \big( \delta_1 - \imag \, k \, \mathrm{Re}(r) \big)}{2 \, k^2 \, \mathrm{Re}(r)} \, \Big(1 - \frac{\big(4 \, \delta_2 \,\mathrm{Im}(r) + \alpha_1 \alpha_2 \big) \, |1+2\,\mathrm{i}\,r\,k|^2 }{4\, |2\,\delta_1\,\delta_2+k\,\alpha_1\,\alpha_2+2\,\mathrm{i}\,k\,\delta_2\,\mathrm{Re}(r)|^2 } \Big) \, e^{2\theta} + \tilde{c} \, ,
\eez 
with a constant $\tilde{c}$, which (\ref{2vFLeq_pw}) requires to be real. The generated solution of the two-component vector FL equation is
\bez
u_\mu'&=& A_\mu \, e^{\imag \, \varphi_\mu} \, \frac{4\,\tilde{c}\,k^2F_{1\mu}\,\mathrm{Re}(r)-|c|^2\,F\,F_{2\mu}\,e^{2\theta}}{4\,k^2\,F_{1\mu}\,\Omega\,\mathrm{Re}(r)}\, ,
\eez
where
\bez
F &=& \frac{1}{4}\big(4 \, \delta_2 \,\mathrm{Im}(r) + \alpha_1 \alpha_2 \big)\,|1+2\,\mathrm{i}\,r\,k|^2-|2\,\delta_1\delta_2+k\,\alpha_1\alpha_2+2\,\mathrm{i}\,k\,\delta_2\mathrm{Re}(r)|^2\, ,\\
F_{1\mu}&=&|1+2\,k\,\alpha_\mu+2\,\mathrm{i}\,r\,k|^2\,|2\,\delta_1\,\delta_2+k\,\alpha_1\,\alpha_2+2\,\mathrm{i}\,k\,\delta_2\,\mathrm{Re}(r)|^2 \, , \\
F_{2\mu}&=&1+2\,\mathrm{i}\,k\,\big(2\,(r-\mathrm{i}\,\alpha_\mu)-r^\ast\big)+4\,k^2\,(r-\mathrm{i}\,\alpha_\mu)\big(2\,r^\ast-(r-\mathrm{i}\,\alpha_\mu)\big)+8\,\mathrm{i}\,k^3\,r^\ast(r-\mathrm{i}\,\alpha_\mu)^2 \, .
\eez
Along the line $\theta=0$, the absolute square of $u_\mu'$ is given by
\bez
|u_\mu'|^2 = |A_\mu|^2 \Big(1-\frac{32\,\alpha_\mu\,\tilde{c}\, |c|^2\,k^5\, F\,\mathrm{Re}(r)^3 \,|2\,\delta_1\,\delta_2+k\,\alpha_1\,\alpha_2+2\,\mathrm{i}\,k\,\delta_2\,\mathrm{Re}(r)|^2}{|1+2\,k\,\alpha_\mu+2\,\mathrm{i}\,r\,k|^2\,\big(\tilde{F}^2+|c|^4\,k^2\,F^2\,\mathrm{Re}(r)^2\big)}\Big) \, ,
\eez
where
\bez
\tilde{F} &=& |c|^2\,\delta_1\,F+2\,\tilde{c}\,k^2 \,\mathrm{Re}(r)\, |2\,\delta_1\,\delta_2+k\,\alpha_1\,\alpha_2+2\,\mathrm{i}\,k\,\delta_2\,\mathrm{Re}(r)|^2 \, .
\eez
Using (\ref{range}), we find that $F<0$. Thus $u_\mu'$ corresponds to a dark soliton when  $\alpha_\mu\, \tilde{c}\, \mathrm{Re}(r)<0$, and to an anti-dark soliton when $\alpha_\mu\, \tilde{c}\, \mathrm{Re}(r)>0$. This class of solutions encompasses dark-dark, anti-dark--dark, and anti-dark--anti-dark solitons on a plane wave background. Notably, we recover the results on dark-dark soliton solutions reported in \cite{Mats18}. Fig.~\ref{fig:2vFL_antidark-dark soliton} shows plots of the absolute values of the two components of $u'$ for an example from 
this class of solutions of the two-component vector FL equation.
\begin{figure}[h]
	\begin{center}
		\includegraphics[scale=.4]{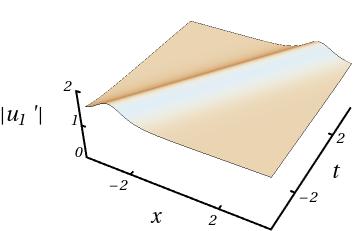} 
		\hspace{1cm}
		\includegraphics[scale=.4]{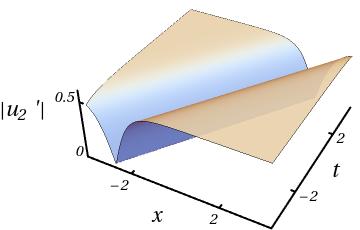} 
		\parbox{15cm}{
			\caption{Plots of the absolute values of the two components of a solution from the class in  Section~\ref{subsec:2vFL_generic_pw_n=1} (ii), showing an anti-dark--dark soliton. Here we chose $\gamma=-\imag$ (i.e., $k=1$), $r=-1+\imag/3$, $\alpha_1=-\alpha_2=1$, 
				$c=1$, $\tilde{c}=-1$. 
				\label{fig:2vFL_antidark-dark soliton} } 
		}
	\end{center} 
\end{figure}

\subsubsection{$n=2$ and generic roots}
Given two $n=1$ solutions from Section~\ref{subsec:2vFL_generic_pw_n=1}, with spectral parameters 
$\gamma_2 \neq -\gamma_1^\ast$, a corresponding superposition,  
which also solves the vector FL equation, is easily computed. With 
the correspondingly generalized restriction, an arbitrary number of $n=1$ solutions can easily be superposed. If, however, the (relative part of the) spectrum condition for the Lyapunov equation does not hold, it is more involved to compute the superposition of $n=1$ solutions. We elaborate an $n=2$ case. 

\begin{example}
\label{ex:2vFL_general_anticonjugate pair}
Let $\Gamma = \mathrm{diag}(\gamma,-\gamma^\ast)$ with $\mathrm{Re}(\gamma) \neq 0$. 
The Lyapunov equation determines the diagonal entries of $\Omega$,
\bez
	\Omega_{11} &=& \frac{1}{2\mathrm{Re}(\gamma)} \Big(\frac{\imag \, \gamma^\ast |\Xi_1|^2}{(A^\dagger \boldsymbol{\alpha}  A)^2}+|\chi_{11}|^2+|\chi_{21}|^2\Big) \, , \\
	\Omega_{22} &=& \frac{1}{2\mathrm{Re}(\gamma)} \Big( \frac{\imag \, \gamma |\Xi_2|^2}{(A^\dagger \boldsymbol{\alpha} A)^2}-|\chi_{12}|^2-|\chi_{22}|^2 \Big) \, ,
\eez
where
\bez
	\Xi = \|A\| \, \Big( \chi_{1x} - \big( \frac{1}{2}  \Gamma^{-1} + \imag \, \frac{A^\dagger \boldsymbol{\alpha} A}{\|A\|^2} \, I_n \big) \, \chi_1 \Big)
	- \imag \, (\alpha_2-\alpha_1) \, \frac{A_1 A_2}{\|A\|} \, \chi_2    \, ,
\eez
and the constraints
\bez
	\Big( \frac{\imag}{(A^\dagger \boldsymbol{\alpha} A)^2} \Xi \Xi^\dagger \, \Gamma^\dagger
	+ \chi_1 \chi_1^\dagger + \chi_2 \chi_2^\dagger \Big)_{ij} = 0 \qquad \mbox{for} \quad i \neq j \, .
\eez
We observe that if $r$ is a solution of (\ref{2vFL_n=1_cubic_eq}), for a given parameter $\gamma$, then $-r^\ast$ is a solution of (\ref{2vFL_n=1_cubic_eq}) with $\gamma$ replaced by $-\gamma^\ast$. Assuming that the three roots of (\ref{2vFL_n=1_cubic_eq}) are all different\footnote{If there is a multiple root, we still obtain exact solutions by restricting the sums to a set of different roots.}, the linear system is completely solved by (\ref{2vFL_gen_sol_chi1}) and (\ref{2vFL_gen_sol_chi2}), 
where now $R_\kappa=\mathrm{diag}(r_\kappa,-r_\kappa^\ast)$, $S_\kappa=\mathrm{diag}(s_\kappa,-s_\kappa^\ast)$, with $s_\kappa$ given by (\ref{2vFL_n=1_s_kappa}), and $c_\kappa=(c_{1\kappa}, c_{2\kappa})^T$. 
Then the above constraints boil down to 
\bez
	\sum_{\kappa=1}^3 c_{1\kappa} \, c_{2\kappa}^\ast \, \big(1 +\imag \, \gamma  h_{\kappa}^2+ \frac{d_\kappa^2\,A_1A_2}{A_1^\ast A_2^\ast}  \big)=0\, ,
\eez
where
\bez
	d_\kappa &=& \frac{1}{(\alpha_2-\alpha_1) A_1 A_2 \big( \imag \, (1-A^\dagger \boldsymbol{\alpha}  A) \, A^\dagger \boldsymbol{\alpha}  A - \gamma \, \alpha_1 \alpha_2 \|A\|^2 \big)}\Big(\gamma\|A\|^2A^\dagger \boldsymbol{\alpha} A\, r_\kappa^2-\imag \, \gamma\|A\|^2A^\dagger \boldsymbol{\alpha}^2 A \, r_\kappa \\
	&&-\frac{1}{4}\gamma^{-1}\|A\|^2A^\dagger \boldsymbol{\alpha} A +\frac{1}{2} \imag \, (A^\dagger \boldsymbol{\alpha} A)^2(2 A^\dagger\boldsymbol{\alpha} A-1)+\frac{1}{2}\imag \, (\alpha_2-\alpha_1)^2|A_1|^2|A_2|^2\Big) \, , \\
	h_{\kappa}&=&\frac{1}{2 \|A\|A^\dagger \boldsymbol{\alpha} A}\Big(2\,(\|A\|^2r_\kappa-\imag \,  k)-\gamma^{-1}\|A\|^2-2\imag \, d_\kappa A_1 A_2\,(\alpha_2-\alpha_1) \Big) \, .
\eez
The differential equations for $\Omega$ read
\bez
	\Omega_{12x} &=&  \frac{\imag}{(A^\dagger \boldsymbol{\alpha} A)^2}\Xi_{1x}\Xi_2^\ast-\frac{1}{2}\gamma^{-2}\chi_{11}\chi_{12}^\ast +\gamma^{-1}\Big(\frac{A_1^\ast A_2^\ast}{A^\dagger \boldsymbol{\alpha} A}(\alpha_2-\alpha_1)\,(\chi_{11x}-\frac{1}{2}\gamma^{-1}\chi_{11})-\chi_{21x}\\
	&&+\frac{\imag \, \|A\|^2}{A^\dagger \boldsymbol{\alpha} A} \alpha_1 \alpha_2 \chi_{21}\Big)\,\chi_{22}^\ast\, ,\\
	\Omega_{12t} &=& \frac{\imag}{(A^\dagger \boldsymbol{\alpha} A)^2}\big(\gamma\,\Xi_2^\ast-  \|A\| A^\dagger \boldsymbol{\alpha} A\,\chi_{12}^\ast\big) \, \Xi_1 \, , \\
	\Omega_{21x} &=& \frac{\imag}{(A^\dagger \boldsymbol{\alpha} A)^2}\Xi_1^\ast \Xi_{2x}-\frac{1}{2}\gamma^{\ast -2}\chi_{11}^\ast \chi_{12}-\gamma^{\ast -1}\Big(\frac{A_1^\ast A_2^\ast}{A^\dagger \boldsymbol{\alpha} A}(\alpha_2-\alpha_1)\,(\chi_{12x}+\frac{1}{2}\gamma^{\ast -1}\chi_{12})-\chi_{22x} \\
	&&+\frac{\imag \, \|A\|^2}{A^\dagger \boldsymbol{\alpha} A} \alpha_1 \alpha_2 \chi_{22}\Big)\,\chi_{21}^\ast\, , \\
	\Omega_{21t} &=& -\frac{\imag}{(A^\dagger \boldsymbol{\alpha} A)^2}\big(\gamma^\ast \Xi_1^\ast+\|A\| A^\dagger \boldsymbol{\alpha} A \chi_{11}^\ast\big) \, \Xi_2 \, .
\eez
Assuming $s_\kappa \neq s_\lambda$ for $\kappa \neq \lambda$, this integrates to
\bez
	\Omega_{12} &=& -\sum_{1 \leq \kappa< \lambda \leq 3} \frac{\imag}{s_\kappa-s_\lambda} \big( \rho_{\kappa\lambda} \, e^{\theta_{\kappa\lambda}} 
	-\rho_{\lambda\kappa} \, e^{-\theta_{\kappa\lambda}} \big) \\
	&& + \frac{1}{2\gamma^2 A_1^\ast A_2^\ast A^\dagger \boldsymbol{\alpha} A} \Big( (\alpha_2-\alpha_1) |A_1|^2 |A_2|^2  \sum_{\kappa=1}^3 c_{1\kappa} c_{2\kappa}^\ast d_{\kappa} (2\gamma r_\kappa -1) \\
	&& - \sum_{\kappa=1}^3 c_{1\kappa}c_{2\kappa}^\ast\big(A_1^\ast A_2^\ast A^\dagger \boldsymbol{\alpha} A \, (2\imag \, \gamma^2 r_\kappa h_\kappa^2 +1) - 2  \gamma A_1A_2    d_{\kappa}^2\,(\imag \, \alpha_1\alpha_2\|A\|^2-r_\kappa A^\dagger \boldsymbol{\alpha} A) \big) \Big) \, x\\
	&&-\imag \, \sum_{\kappa=1}^3 \rho_{\kappa\kappa}\,  t + \tilde{c}_{12} 
	\, , \\
	\Omega_{21} &=& \sum_{1\leq \kappa < \lambda \leq 3} \frac{\imag}{s_\kappa^\ast-s_\lambda^\ast} \big( \varrho_{\kappa\lambda} \, e^{\theta_{\kappa\lambda}^\ast} - \varrho_{\lambda\kappa} \, e^{-\theta_{\kappa\lambda}^\ast} \big) \\
	&&+\frac{1}{2\gamma^{\ast 2} A_1 A_2 A^\dagger \boldsymbol{\alpha} A} \Big((\alpha_2-\alpha_1) |A_1|^2 |A_2|^2 \sum_{\kappa=1}^3 c_{1\kappa}^\ast c_{2\kappa} d_\kappa^\ast (2\gamma^\ast r_\kappa^\ast-1) \\
	&& + \sum_{\kappa=1}^3 c_{1\kappa}^\ast c_{2\kappa}\big(A_1 A_2 A^\dagger \boldsymbol{\alpha} A \, (2 \imag \, \gamma^{\ast 2}r_\kappa^\ast h_\kappa^{\ast 2}-1)- 2 \gamma^\ast A_1^\ast A_2^\ast d_{\kappa}^{\ast 2} \, (\imag \, \alpha_1 \alpha_2 \|A\|^2 + r_\kappa^\ast A^\dagger \boldsymbol{\alpha} A) \big) \Big) \, x \\
	&& +\imag \, \sum_{\kappa=1}^3 \varrho_{\kappa\kappa}\, t + \tilde{c}_{21} \, ,
\eez
with $\theta_{\kappa\lambda} = (r_\kappa-r_\lambda) \, x + (s_\kappa-s_\lambda) \, t$,  $\rho_{\kappa\lambda} = c_{1\kappa} c_{2\lambda}^\ast h_\kappa \tilde{h}_\lambda$,  $\varrho_{\kappa\lambda} = c_{1\kappa}^\ast c_{2\lambda}\tilde{h}_\kappa^\ast h_\lambda^\ast$, $\tilde{h}_\kappa = \gamma h_\kappa + \|A\|$, and complex constants $\tilde{c}_{12}$, $\tilde{c}_{21}$.
The condition (\ref{2vFLeq_pw}) becomes
\bez
	\gamma \, (\tilde{c}_{12}-\tilde{c}_{21}^\ast) = \sum_{\kappa=1}^3    c_{1\kappa} \, c_{2\kappa}^\ast \big(1+\frac{d_\kappa^2\,A_1A_2}{A_1^\ast A_2^\ast}\big) \, .
\eez  
Solutions of the two-component FL equation are now obtained by inserting the expressions for 
$\chi_1$, $\chi_2$, $\Xi$ and $\Omega$ in (\ref{special_pw_u'}).  The resulting solution depends on the independent variables through three exponentials, $e^{\theta_{12}}$, $e^{\theta_{13}}$, $e^{\theta_{23}}$, and their complex conjugates, and in addition there is a rational dependence, a typical feature in cases, where the spectrum condition for the Lyapunov equation does not hold. 
Generically, we expect an interaction of three breathers. In the following, we analyze the asymptotic behavior, which can be confirmed by an asymptotic analysis. The real part of $\theta_{\kappa\lambda}$, $\kappa<\lambda$, determines the envelope of the corresponding asymptotic breather, whereas its imaginary part leads to a modulation. We have
$\mathrm{Re}(\theta_{\kappa\lambda}) = \mathrm{Re}(r_\kappa-r_\lambda)\, (x-v_{\kappa\lambda} \,t)$ with the velocity
\bez
v_{\kappa\lambda} = - \frac{(\alpha_1+\alpha_2)\,\mathrm{Im}\big(\gamma\,(r_\kappa-r_\lambda)\big)+\mathrm{Re}\big(\gamma\,(r_\kappa^2-r_\lambda^2)\big)}{\alpha_1\alpha_2 \, \mathrm{Re}(r_\kappa-r_\lambda)} \, .
\eez
\hfill $\Box$
\end{example}

\begin{example} 
Let $\Gamma$ be the $2 \times 2$ lower triangular Jordan block with eigenvalue $\gamma$. Then we have
\bez
	R_\kappa =  \left( \begin{array}{cc} r_\kappa & 0 \\
		\tilde{r}_\kappa & r_\kappa \end{array} \right) \, ,  \qquad 	S_\kappa =  \left( \begin{array}{cc} s_\kappa & 0 \\
		\tilde{s}_\kappa & s_\kappa \end{array} \right) \, ,
\eez	
with 
\bez
	\tilde{r}_\kappa &=& \frac{8\imag \, \gamma A^\dagger \boldsymbol{\alpha}^2 A \,(\gamma r_\kappa-1) 
		+4\gamma^2 \alpha_1 \alpha_2 \, (2 A^\dagger \boldsymbol{\alpha} A-1)+4 \imag \, \gamma \, (\alpha_1+\alpha_2)-4\gamma r_\kappa\, (1+\gamma r_\kappa)+3}{2\gamma^2 \big(4 \imag \, \gamma A^\dagger \boldsymbol{\alpha}^2 A- 8 \imag \, \gamma^2 r_\kappa\, (\alpha_1+\alpha_2)- 4 \gamma^2 \alpha_1 \alpha_2+12\gamma^2 r_\kappa^2-4\gamma r_\kappa-1\big)} \, , \\
	\tilde{s}_\kappa &=& -\frac{1}{4 \alpha_1 \alpha_2 \gamma^2}\Big(4 \imag \, \gamma^2(r_\kappa+\gamma \tilde{r}_\kappa)\,(\alpha_1+\alpha_2)+ 2 \gamma^2 \alpha_1 \alpha_2-8\gamma^3 r_\kappa \tilde{r}_\kappa-4\gamma^2 r_\kappa^2-1 \Big) \, .
\eez
Assuming that the three roots $r_\kappa$, $\kappa=1,2,3$, of (\ref{2vFL_n=1_cubic_eq}) are all different, the linear system is completely solved by (\ref{2vFL_gen_sol_chi1}) and (\ref{2vFL_gen_sol_chi2}). Then 
$\Xi$ is easily computed and it remains to determine $\Omega$, in order to be able to elaborate the resulting solution of the two-component vector FL equation, which is a second order version of that in Section~\ref{subsec:2vFL_generic_pw_n=1} with the same $\gamma$. If $\mathrm{Re}(\gamma) \neq 0$, then $\Omega$ is fully determined by the Lyapunov equation, see Example~\ref{ex:n=2Jordan_Lyapunov_sol}.
If $\mathrm{Re}(\gamma) = 0$, the differential equations for $\Omega$ have to be solved, which is more tedious. We omit the details. 
\hfill $\Box$
\end{example}

\subsubsection{The triple root case}
\label{subsec:2vFL_triple_root}
The solution of the linear equation for $\chi_1$, obtained in Section~\ref{subsec:alpha1_neq_alpha2} in terms of only exponential functions, is not complete in case of a multiple root of the 
characteristic equation (\ref{2vFL_n=1_cubic_eq}), because there are then additional solutions involving polynomials of the independent variables. 
In this section, we address the special case of a triple root (also see \cite{YZCBG19,Ling+Su24}). This leads to genuine\footnote{In contrast to solutions that are equivalent to solutions of the scalar FL equation.} quasi-rational solutions of the two-component vector FL equation, which turn out to describe in space and time (formally regarding $t$ as time and $x$ as space coordinate) localized perturbations of the plane wave background. 

\begin{example}
	\label{ex:2vFL_Peregrine}
Let $n=1$. If (\ref{2vFL_n=1_cubic_eq}) admits a triple root, then it is given by 
\be
   r =  \frac{\imag}{3} (\alpha_1+\alpha_2) + \frac{1}{6}\gamma^{-1} 
   \, .      \label{2vFL_triple_root}
\ee
The triple root condition further requires
\be
	\gamma &=& \frac{2 \,  \big( (\alpha_1 - \alpha_2)^2 + \alpha_1 \alpha_2 \big) }
	{\imag \, (\alpha_1 + \alpha_2) (2\alpha_1-\alpha_2)(\alpha_1-2\alpha_2) \pm 3 \sqrt{3} \, \alpha_1 \alpha_2 (\alpha_1 - \alpha_2) } \, , \label{gamma_rw} \\
	|A_\mu|^2 &=& 
   \frac{\alpha_\mu\,(\alpha_1-\alpha_2)^2}{\big( (\alpha_1-\alpha_2)^2 + \alpha_1 \alpha_2 \big)^2}
	\, , \qquad \mu=1,2 \,  .  \label{2vFL_triple_root_case_A_mu} 
\ee
Introducing 
\bez
    \tilde{\chi}_{1} = \chi_{1} \, e^{-\theta} \, ,
\eez
with 
\be
	\theta = r\,x+\frac{1}{\alpha_1 \alpha_2} \Big(r^2\gamma-\imag \, r \gamma (\alpha_1+\alpha_2)-\big(\frac{1}{2} \alpha_1 \alpha_2 \gamma + \frac{1}{4}\gamma^{-1}-\imag \, (\alpha_1+\alpha_2) (A^\dagger \boldsymbol{\alpha} A-\frac{1}{2}) \big) \Big) \, t \, ,  \label{2vFL_triple_root_case_theta}
\ee
the linear equations for $\chi_1$ boil down to
\bez
\tilde{\chi}_{1xxx} = 0 \, ,  \qquad
\alpha_1 \alpha_2 \, \tilde{\chi}_{1t} = \gamma \, \tilde{\chi}_{1xx} + \frac{1}{3}\tilde{\gamma} \, \tilde{\chi}_{1x} \, , 
\eez
where 
\be
   \tilde{\gamma} = 1-\imag \, (\alpha_1+\alpha_2)\gamma \, .
        \label{2vFL_triple_root_case_tgamma}
\ee
The linear system is now completely solved by
\bez
	\chi_1 &=& \big(f_1 \,x^2 + f_2 \, x + f_3\big)\,e^{\theta} \, , \\
	\chi_2 &=& \frac{\|A\|^2}{(\alpha_2-\alpha_1) \, A_1 A_2} 
	\Big( \imag \, (1 - A^\dagger \boldsymbol{\alpha} A) \, \gamma^{-1} - \frac{\alpha_1 \alpha_2 \|A\|^2}{A^\dagger \boldsymbol{\alpha} A}  \Big)^{-1} \Big( \chi_{1xx} - \imag \, \frac{A^\dagger \boldsymbol{\alpha}^2 A}{A^\dagger \boldsymbol{\alpha} A} \, \chi_{1x} - \frac{1}{4} \gamma^{-2} \chi_1  \\
	&& + \frac{\imag}{2 \|A\|^2 \, A^\dagger \boldsymbol{\alpha} A} 
	\big( (\alpha_2-\alpha_1)^2 |A_1|^2 |A_2|^2 
	- (A^\dagger \boldsymbol{\alpha} A)^2 (1-2 A^\dagger \boldsymbol{\alpha} A) \big) \, \gamma^{-1} \chi_1 \Big) \\
    &=& \frac{\|A\|^2}{(\alpha_2-\alpha_1) \, A_1 A_2} 
	\Big( \imag \, (1 - A^\dagger \boldsymbol{\alpha} A) \, \gamma^{-1} - \frac{\alpha_1 \alpha_2 \|A\|^2}{A^\dagger \boldsymbol{\alpha} A}  \Big)^{-1} \Big( \Big[ \frac{\imag}{2 \|A\|^2 \, A^\dagger \boldsymbol{\alpha} A} 
	\big( (\alpha_2-\alpha_1)^2 |A_1|^2 |A_2|^2  \\
	&& - (A^\dagger \boldsymbol{\alpha} A)^2 (1-2 A^\dagger \boldsymbol{\alpha} A) \big) \, \gamma^{-1} + r^2 - \imag \, \frac{A^\dagger \boldsymbol{\alpha}^2 A}{A^\dagger \boldsymbol{\alpha} A} r - \frac{1}{4} \gamma^{-2} \Big] \, (f_1 \, x^2 + f_2 \, x + f_3) \\
	&& - \big( \imag \,  \frac{A^\dagger \boldsymbol{\alpha}^2 A}{A^\dagger \boldsymbol{\alpha} A}-2r) \, (2 f_1 \, x + f_2) + 2 f_1 \Big) \, e^{\theta} \, ,   
\eez
where the functions $f_1,f_2,f_3$ satisfy 
\bez
     \left( \begin{array}{c} f_1 \\ f_2 \\ f_3 \end{array} \right)_t 
    = C  \left( \begin{array}{c} f_1 \\ f_2 \\ f_3 \end{array} \right) \, , \qquad
   C = \frac{1}{3\,\alpha_1 \alpha_2} \, \left( \begin{array}{ccc} 
	0 & 0 & 0 \\
	2 \tilde{\gamma} & 0 & 0 \\
	6\gamma & \tilde{\gamma} & 0 \end{array} \right)   \, .
\eez
Hence
\be
   \left( \begin{array}{c} f_1 \\ f_2 \\ f_3 \end{array} \right) =e^{C t}   \left( \begin{array}{c} c_1 \\ c_2 \\ c_3  \end{array} \right)
   = \left( \begin{array}{ccc} 
	1 & 0 & 0 \\ 2 \tilde{\gamma} t/(3\, \alpha_1 \alpha_2) & 1 & 0 \\
    \tilde{\gamma}^2 t^2/(3\,\alpha_1 \alpha_2)^2 + 2 \gamma t/(\alpha_1 \alpha_2) &
           \tilde{\gamma} t/(3 \, \alpha_1 \alpha_2) &  1 \end{array} \right)  
          \left( \begin{array}{c} c_1 \\ c_2 \\ c_3  \end{array} \right)  \, ,   \label{f_1,f_2,f_3}
\ee
with complex constants $c_1,c_2,c_3$. 
Since we assume $\alpha_1 \neq \alpha_2$, we have
\bez
\mathrm{Re}(\gamma) = \pm \frac{3\sqrt{3} \, \alpha_1 \alpha_2 \, (\alpha_1-\alpha_2)}{2 \, \big( (\alpha_1-\alpha_2)^2 +\alpha_1 \alpha_2 \big)^2} \neq 0 \, ,
\eez
so that the Lyapunov equation has the unique solution
\bez
  	\Omega = \frac{\imag \, \gamma^\ast |\Xi|^2 + (A^\dagger \boldsymbol{\alpha} A)^2 (|\chi_1|^2 + |\chi_2|^2)}{2 \mathrm{Re}(\gamma) \, (A^\dagger \boldsymbol{\alpha} A)^2} \, .
\eez
The generated solution (\ref{alphas_different_n=1_solution}) of the two-component vector FL equation is now
\bez
u_1' &=& e^{\imag \, \varphi_1} \Big( A_1 + \frac{2\mathrm{Re}(\gamma) A^\dagger \boldsymbol{\alpha} A \, (A_1 \chi_1^\ast - A_2^\ast \chi_2^\ast) \, \Xi}{\|A\| \big( \imag \, \gamma^\ast |\Xi|^2 + (A^\dagger \boldsymbol{\alpha}  A)^2(|\chi_1|^2 + |\chi_2|^2) \big)} \Big) \, , \\
u_2' &=&	 e^{\imag \, \varphi_2} \Big( A_2 + \frac{2\mathrm{Re}(\gamma)A^\dagger \boldsymbol{\alpha} A \, (A_2 \chi_1^\ast + A_1^\ast \chi_2^\ast)\, \Xi}{\|A\| \big( \imag \, \gamma^\ast |\Xi|^2 + (A^\dagger \boldsymbol{\alpha} A)^2 (|\chi_1|^2 + |\chi_2|^2) \big)} \Big)  \, ,
\eez
where
\bez
\Xi &=& \|A\| \, \Big( \big(r-(\frac{1}{2} \gamma^{-1} + \imag \, \frac{A^\dagger \boldsymbol{\alpha} A}{\|A\|^2}) \big) \, (f_1 x^2 + f_2 x + f_3) + 2 f_1 x + f_2 \\
&& - \Big( (1 - A^\dagger \boldsymbol{\alpha} A) \, \gamma^{-1} + \imag \, \frac{\alpha_1 \alpha_2 \|A\|^2}{A^\dagger \boldsymbol{\alpha} A} \Big)^{-1} 
 \Big[ \Big( \frac{\imag}{2 \gamma \, \|A\|^2 \, A^\dagger \boldsymbol{\alpha} A} 
\big( (\alpha_2-\alpha_1)^2 |A_1|^2 |A_2|^2 \\
&& - (A^\dagger \boldsymbol{\alpha} A)^2 (1-2 A^\dagger \boldsymbol{\alpha} A) \big) 
 + r^2 - \imag \, \frac{A^\dagger \boldsymbol{\alpha}^2 A}{A^\dagger \boldsymbol{\alpha} A} r - \frac{1}{4} \gamma^{-2} \Big) (f_1 x^2 + f_2 x + f_3) \\
&& - \big( \imag \, \frac{A^\dagger \boldsymbol{\alpha}^2 A}{A^\dagger \boldsymbol{\alpha}  A} - 2 r) \, (2 f_1 x + f_2) +2 f_1 \Big] \Big) \, e^{\theta} \, . 
\eez
We note that $e^\theta$ drops out in the expression for $u'$, so that the solution $u'$ is quasi-rational. 

If $c_1=c_2=0$, then $u_1'$ and $u_2'$ reduce to the plane wave background solutions. 

If $c_1=0$ and $c_2 \neq 0$, so that $f_1=0$ and $f_2=c_2 \neq 0$, then $\chi_1$, $\chi_2$ and $\Xi$ only involve polynomials  (in the independent variables) that are linear. In this case we obtain the single Peregrine-type soliton solution
\bez
u_\mu'=\frac{1}{2}A_\mu\, e^{\imag \, \varphi_\mu}\Big(-1 \pm (-1)^\mu \, \imag \, \sqrt{3}-\frac{4\,\mathrm{Re}(\gamma)A^\dagger \boldsymbol{\alpha} A \, F_\mu}{D \, A_\mu} \Big) \qquad \quad \mu =1,2 \, ,
\eez
where
\bez
F_\mu &=& \varrho_3 \, \Big( \imag \, \gamma^\ast \|A\|^2 \big( \varrho_2^\ast \varrho_3 \, (A_\mu+\varrho_4^\ast \tilde{A}_\mu) - \varrho_1^\ast |\varrho_3|^2 \tilde{A}_\mu \big) + (A^\dagger \boldsymbol{\alpha} A)^2 \varrho_1^\ast (\varrho_4 A_\mu - \tilde{A}_\mu) \Big) \, \rho \\
&& + (A^\dagger \boldsymbol{\alpha} A)^2 (A_\mu+\varrho_4^\ast \tilde{A}_\mu) \, (\varrho_1 \varrho_3 \varrho_4^\ast - \varrho_2| \varrho_4|^2 - \varrho_2) \, \rho^\ast \\
&& + \varrho_3 (A_\mu+\varrho_4^\ast \tilde{A}_\mu) \, \big( \imag \, \gamma^\ast |\varrho_2|^2 \|A\|^2 
+|\varrho_1|^2 (A^\dagger \boldsymbol{\alpha} A)^2 \big) \\
&& -\varrho_1^\ast \varrho_2 \tilde{A}_\mu \, \big( \imag \, \gamma^\ast| \varrho_3|^2 \|A\|^2
+ (A^\dagger \boldsymbol{\alpha} A)^2 (1+|\varrho_4|^2) \big) \, , \\
D &=& \big(\imag \, \gamma^\ast |\varrho_3|^2\|A\|^2+(A^\dagger \boldsymbol{\alpha} A)^2 (1+|\varrho_4|^2) \big) \, \Big(\big((A^\dagger \boldsymbol{\alpha} A)^2 (1+|\varrho_4|^2)+\imag \, \gamma^\ast |\varrho_3|^2\|A\|^2\big)\,|\rho|^2\\
&&+\big((A^\dagger \boldsymbol{\alpha} A)^2\varrho_1^\ast\varrho_4+\imag \, \gamma^\ast \varrho_2^\ast\varrho_3 \|A\|^2 \big) \, \rho + \big( (A^\dagger \boldsymbol{\alpha} A)^2 \varrho_1 \varrho_4^\ast + \imag \, \gamma^\ast \varrho_2 \varrho_3^\ast \|A\|^2 \big) \, \rho^\ast \\
&&+(A^\dagger \boldsymbol{\alpha} A)^2 |\varrho_1|^2 + \imag \, \gamma^\ast |\varrho_2|^2 \|A\|^2 \Big) \, ,
\eez
with $\tilde{A}_\mu=(-1)^\mu  A_{3-\mu}^\ast$,  $\rho=c_2\,(x+(3\alpha_1 \alpha_2)^{-1}\tilde{\gamma}\,t)+c_3$,  and 
\begin{alignat*}{2}
	\varrho_1 &=\pm  \frac{4\sqrt{3}\,c_2\,(\nu_1^2-\nu_2^2)}{A_1A_2\nu_2^2\,(2\nu_1-\nu_2)^2} \, , 
	&\qquad 
	\varrho_2 &= \frac{c_2\,(\nu_1^2-4\nu_1\nu_2+\nu_2^2)}{(2\nu_1-\nu_2)^2} \, , \\
	\varrho_3 &= \mp\frac{(\nu_1+\nu_2)\,(\nu_1^2-4\nu_1\nu_2+\nu_2^2)}{6\nu_1(2\nu_1-\nu_2)} \, , 
	&\qquad 
	\varrho_4 &= \frac{2\sqrt{3}(\nu_1+\nu_2)^2(\nu_1^2-\nu_1\nu_2+\nu_2^2)}{9A_1A_2\nu_1^2\nu_2^2(2\nu_1-\nu_2)} \, .
\end{alignat*}
Here  $\nu_1=\sqrt{3}(\alpha_1-\alpha_2)\pm \imag \, (\alpha_1+\alpha_2),\, \nu_2=\sqrt{3}(\alpha_1-\alpha_2)\mp \imag \, (\alpha_1+\alpha_2)$. The $\pm$ corresponds to that in the denominator of (\ref{gamma_rw}). 
This solution has been explored in particular in \cite{CYSGB18}, where it has been shown that it exhibits surprisingly high amplitudes for particular parameter values, in particular as compared with the Manakov (two-component vector NLS) system (reached in the limit $\epsilon \to 0$, see Remark~\ref{rem:other_vFL} for this parameter).
Indeed, choosing for example $c_2=c_3 = 1$, we find that the absolute values of the components of $u'$ at the point
\bez
(x_0, t_0) =
\Big(-1\mp\frac{\sqrt{3}\,  \big((\alpha_1-\alpha_2)^4-\alpha_1^2 \alpha_2^2 \big)}{6\alpha_1\alpha_2(\alpha_1^3-\alpha_2^3)} \, , \mp\frac{\sqrt{3}\,\big(\alpha_1^2-\alpha_1 \alpha_2 +\alpha_2^2 \big)^3}{6\alpha_1 \alpha_2 (\alpha_1^3-\alpha_2^3)} \Big) 
\eez
are given by
\bez 
|u_1'| = |A_1| \, \frac{|\alpha_1-5 \alpha_2|}{|\alpha_1+\alpha_2|}\, ,\qquad |u_2'| = |A_2| \, \frac{|5\alpha_1- \alpha_2|}{|\alpha_1+\alpha_2|} \, .
\eez 
Note that (\ref{2vFL_triple_root_case_A_mu}) requires $\alpha_\mu>0$. The last equations show that the peak amplitude can approach nearly five times the background level, which is consistent with the results reported in \cite{CYSGB18}. 
An example of a rogue wave from this class of solutions is displayed in Fig.~\ref{fig:2vFL_rogue wave}. 

\begin{figure}[h]
	\begin{center}
		\includegraphics[scale=.5]{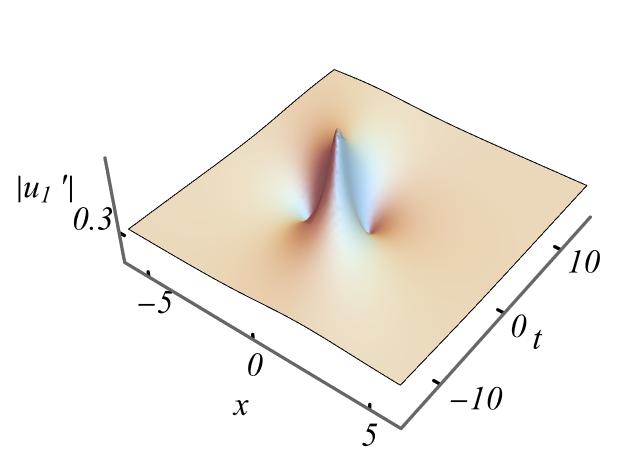} 
		\hspace{1cm}
		\includegraphics[scale=.5]{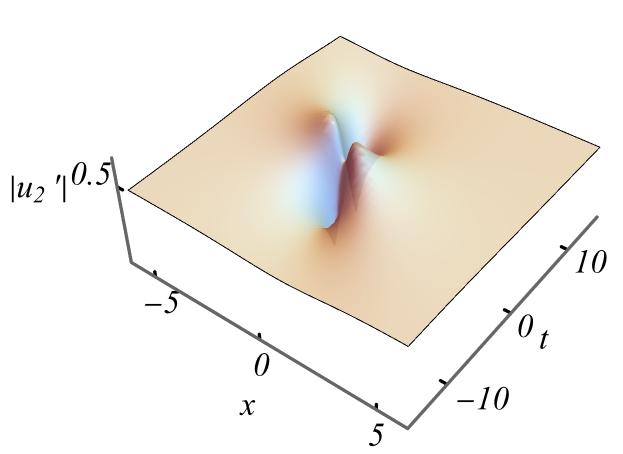} 
		\parbox{15cm}{
		\caption{Plots of the absolute values of the two components of a solution from the class in Example~\ref{ex:2vFL_Peregrine}, both showing rogue waves. Here we chose $\alpha_1=1$, $\alpha_2=2$, $A_1=1/3$, $A_2=\sqrt{2}/3$, $c_1=0$, $c_2=c_3=1$. 
				\label{fig:2vFL_rogue wave} } 
		}
	\end{center} 
\end{figure} 

If $c_1 \neq 0$, the asymptotically leading polynomial 
in $\chi_1$, $\chi_2$ and $\Xi$ is $x^2 + \tilde{\gamma}^2 (3 \alpha_1 \alpha_2)^{-2} t^2$. Excluding $\mathrm{Re}(\tilde{\gamma})=0$, which is $(\alpha_1+\alpha_2) \mathrm{Im}(\gamma) 
= -1$, this polynomial is different from zero for all $x,t$. 
Then we obtain
\bez
u'_\mu \sim \frac{1}{2} A_\mu \, e^{\imag \, \varphi_\mu} \Big( -1 \pm (-1)^\mu \, \imag \, \sqrt{3} \Big)
\qquad \quad \mu =1,2 \, ,
\eez
as $x \to \infty$ or $t \to \infty$. Here the $\pm$ corresponds to that in the denominator of (\ref{gamma_rw}). As a consequence, asymptotically we have $|u'_\mu| = |A_\mu|$. Hence 
the components of the solution $u'$ are both localized on the respective plane wave background, which is the characteristic property of rogue waves (interpreting $x$ and $t$ as ``space" and ``time" coordinates, or vice versa). 
The components $u_1'$ and $u_2'$ form pairs of (elementary) rogue waves. An example is shown in Fig.~\ref{fig:2vFL_rw_doublet}. 

\begin{figure}[h]
	\begin{center}
		\includegraphics[scale=.5]{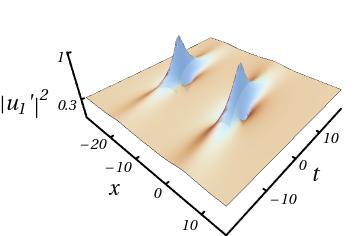} 
		\hspace{1cm}
		\includegraphics[scale=.5]{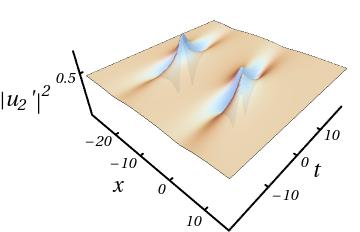} 
		\parbox{15cm}{
			\caption{Plots of the absolute squares of the two components of a solution from the class in Example~\ref{ex:2vFL_Peregrine}, exhibiting rogue wave pairs. Here we chose the same data as in Fig.~\ref{fig:2vFL_rogue wave}, but now $c_1=1$ and  $c_3=0$.  
				\label{fig:2vFL_rw_doublet} } 
		}
	\end{center} 
\end{figure} 
\hfill $\Box$
\end{example}

\begin{example}
\label{ex:triple_root_n=2}	
Let $n=2$ and 
\bez
	\Gamma= \left( \begin{array}{cc} 
		\gamma
		&  0  \\  1  & \gamma \end{array} \right)\,  , \qquad
		\chi_i = \left( \begin{array}{c} \chi_{i1} \\ \chi_{i2} 
		\end{array} \right) \qquad i=1,2 \, .
\eez 
 that $r$ is a triple root of the cubic equation (\ref{2vFL_n=1_cubic_eq}), hence given by (\ref{2vFL_triple_root}). Then $\gamma$ and $A_\mu$, $\mu=1,2$, are given by (\ref{gamma_rw}) and (\ref{2vFL_triple_root_case_A_mu}), respectively. 
Now the linear system takes the form
\bez
 \tilde{\chi}_{11xxx} &=& 0\, ,  \\
 \tilde{\chi}_{12xxx} &=& -\frac{1}{2} \gamma^{-2} \big( \tilde{\chi}_{11xx} + 2 \gamma^{-1} (w_{r} \, \tilde{\chi}_{11x} + w \, \tilde{\chi}_{11}) \big) \, , \\
\alpha_1 \alpha_2 \, \tilde{\chi}_{11t} &=& \gamma \, \tilde{\chi}_{11xx} + \frac{1}{3} \tilde{\gamma} \, \tilde{\chi}_{11x} \, , \\
\alpha_1 \alpha_2 \, \tilde{\chi}_{12t} &=& \gamma \, \tilde{\chi}_{12xx} + \frac{1}{3} \tilde{\gamma} \, \tilde{\chi}_{12x} + \gamma^{-1} \alpha_1 \alpha_2 \, \tilde{\chi}_{11t} + B \, \tilde{\chi}_{11} \, ,
\eez
 where 
\bez
 B &=& \frac{1}{18} \big( 5 \gamma^{-2} -\imag \, (\alpha_1+\alpha_2) \gamma^{-1} + ( 4 \alpha_1^2 - \alpha_1 \alpha_2 + 4 \alpha_2^2 ) \big)\, , \\
w &=& 3\gamma^2 r^3-\big(1+3\,\imag \, \gamma (\alpha_1+\alpha_2) \big) \gamma \, r^2 + \big( 2 \imag \, A^\dagger \boldsymbol{\alpha}^2 A \gamma - 3 \alpha_1 \alpha_2 \gamma^2 - \frac{1}{4} \big) r \\
&& - \frac{1}{2} \imag \, A^\dagger \boldsymbol{\alpha}^2 A + (2A^\dagger \boldsymbol{\alpha} A - 1) \, \alpha_1 \alpha_2 \gamma + \frac{1}{4} \imag \, (\alpha_1 + \alpha_2) \, .
\eez
It is solved by
\bez
\chi_{11}&=&\big(f_1 \, x^2 + f_2 \,x+f_3 \big) \, e^{\theta} \, ,\\
\chi_{12}&=&\Big(f_4 \, x^2+f_5 \,x+ f_6 - \frac{(w f_3+w_r f_2 + \gamma f_1   )\, x^3}{6\, \gamma^3}- \frac{(   w f_2+2\, w_r f_1 ) \, x^4}{24 \,\gamma^3}   -\frac{w f_1\, x^5}{60\, \gamma^3}   \Big)\, e^{\theta}\, ,\\
\chi_2 &=& \frac{\|A\|^2}{(\alpha_2-\alpha_1) \, A_1 A_2} 
\Big( \imag \, (1 - A^\dagger \boldsymbol{\alpha} A) \, \Gamma^{-1} - \frac{\alpha_1 \alpha_2 \|A\|^2}{A^\dagger \boldsymbol{\alpha} A} I_n \Big)^{-1} \Big( \chi_{1xx} - \imag \, \frac{A^\dagger \boldsymbol{\alpha}^2 A}{A^\dagger \boldsymbol{\alpha} A} \, \chi_{1x} - \frac{1}{4} \Gamma^{-2} \chi_1  \\
&& + \frac{\imag}{2 \|A\|^2 \, A^\dagger \boldsymbol{\alpha} A} 
\big( (\alpha_2-\alpha_1)^2 |A_1|^2 |A_2|^2 
- (A^\dagger \boldsymbol{\alpha} A)^2 (1-2 A^\dagger \boldsymbol{\alpha} A) \big) \, \Gamma^{-1} \chi_1 \Big) 
 \, ,  
\eez
where $\theta$ is given by (\ref{2vFL_triple_root_case_theta})  and $f = (f_1,\ldots,f_6)^T$ satisfies $f_t = C f$ with the constant nilpotent matrix
\bez
C = \frac{1}{3 \, \alpha_1 \alpha_2} \,\left( \begin{array}{cccccc}
	0 & 0 & 0 & 0 & 0 & 0 \\
	2 \tilde{\gamma} & 0 & 0 & 0 & 0 & 0 \\
	6\gamma & \tilde{\gamma} & 0 & 0 & 0 & 0 \\
	3(B - \gamma^{-2}w_r) - \frac{1}{2}\gamma^{-2}\tilde{\gamma} & -\frac{1}{2}\gamma^{-3}(3w\gamma + \tilde{\gamma}w_r)  & -\frac{1}{2}w\gamma^{-3}\tilde{\gamma} & 0 & 0 & 0 \\
	\gamma^{-1}(2\tilde{\gamma} - 3) & 3(B - \gamma^{-2}w_r) & -3\gamma^{-2} w & 2 \tilde{\gamma} & 0 & 0 \\
	6 &\gamma^{-1}\tilde{\gamma}  & 3B  & 6\gamma & \tilde{\gamma} & 0
\end{array} \right)\, .
\eez
Consequently, $f_1, f_2,f_3$ are given by (\ref{f_1,f_2,f_3}) and, in addition, we obtain 
\bez
f_4 &=& -\frac{c_1w\tilde{\gamma}^3} {6 \gamma^3 (3 \alpha_1 \alpha_2)^3} \, t^3 -\frac{\tilde{\gamma}}{4\gamma^3(3 \alpha_1 \alpha_2)^2} \, \big(2 c_1 \, (w_r \tilde{\gamma} +6 w\gamma)+c_2 w \tilde{\gamma} \big) \, t^2 \\
&& + \frac{1}{6 \gamma^3 \alpha_1 \alpha_2}\Big(c_1 \gamma\,(6\,(B \gamma^2-w_r) -\tilde{\gamma}) -c_2 \, (3w \gamma +w_r \tilde{\gamma}) -c_3 w \tilde{\gamma} \Big) \, t + c_4 \, , \\
f_5 &=& -\frac{c_1 w \tilde{\gamma}^4}{12 \gamma^3 (3 \alpha_1 \alpha_2)^4} \, t^4 -  \frac{\tilde{\gamma}^2}{6 \gamma^3(3 \alpha_1 \alpha_2)^3} \big(2 c_1 \, (w_r \tilde{\gamma} +9 w \gamma) + c_2 w \tilde{\gamma} \big) \, t^3 \\
&& + \frac{1}{2 \gamma^3 (3 \alpha_1 \alpha_2)^2} \Big( c_1 \big( 12 \gamma \tilde{\gamma} \, (B \gamma^2 - w_r) - \gamma \, (18 w \gamma + \tilde{\gamma}^2) \big) - c_2 \tilde{\gamma} \, (w_r \tilde{\gamma} +6 w \gamma)-c_3 w \tilde{\gamma}^2 \Big) \, t^2 \\
&& + \frac{1}{3 \alpha_1 \alpha_2 \gamma^2} \Big( c_1 \gamma \, (2 \tilde{\gamma}-3) + 3 c_2 \, (B \gamma^2 -w_r) - 3 c_3 w + 2 c_4 \gamma^2 \tilde{\gamma} \Big) \,t + c_5 \, , \\
f_6 &=& -\frac{c_1 w \tilde{\gamma}^5}{60 \gamma^3(3 \alpha_1 \alpha_2)^5} \, t^5 - \frac{\tilde{\gamma}^3}{24 \gamma^3 (3 \alpha_1 \alpha_2)^4} \big(2 c_1 \, (w_r \tilde{\gamma} + 12 w \gamma) + c_2 w \tilde{\gamma} \big) \, t^4 \\
&& + \frac{\tilde{\gamma}}{6 \gamma^3 (3 \alpha_1 \alpha_2)^3} \Big( c_1 \gamma \, \big(18 \gamma \, (B \gamma \tilde{\gamma} - 3 w) - \tilde{\gamma} \, (18 w_r + \tilde{\gamma}) \big) - c_2 \tilde{\gamma} (9w \gamma + w_r \tilde{\gamma}) - c_3 w \tilde{\gamma}^2 \Big) \, t^3 \\
&& + \frac{1}{2 \gamma^2 (3 \alpha_1 \alpha_2)^2} \Big( 2 c_1 \gamma \, \big( 9 \, (2 B \gamma^2 - w_r) + \tilde{\gamma} \, (2 \tilde{\gamma}-3) \big) + c_2 \, \big(6 \tilde{\gamma} \, (B \gamma^2-w_r) - 9 w \gamma \big) - 6 c_3 w \tilde{\gamma} \\
&& + 2 c_4 \gamma^2 \tilde{\gamma}^2 \Big) \, t^2 + \frac{1}{3 \gamma \, \alpha_1 \alpha_2} \big( 6 c_1 \gamma + c_2 \tilde{\gamma} + 3 c_3 B \gamma + 6 c_4 \gamma^2 + c_5 \gamma \tilde{\gamma} \big) \, t + c_6 \, ,
\eez
with complex constants $c_1,\ldots,c_6$, and $\tilde{\gamma}$ given by (\ref{2vFL_triple_root_case_tgamma}). 
The solution of the Lyapunov equation is
\bez
 	\Omega = \left( \begin{array}{cc} \kappa^{-1} M_{11} &
 		- \kappa^{-2} M_{11} + \kappa^{-1} M_{12} \\[4pt] 
 		- \kappa^{-2} M_{11} + \kappa^{-1} M_{21} &
 		2 \kappa^{-3} M_{11} - \kappa^{-2} (M_{12}+M_{21})	
 		+ \kappa^{-1} M_{22} \end{array} \right) \, ,
\eez
with $\kappa = 2 \, \mathrm{Re}(\gamma)$ and
\bez
 	M &=& (M_{ij}) =  \frac{\imag}{(A^\dagger \boldsymbol{\alpha} A)^2} \Xi \, \Xi^\dagger \Gamma^\dagger + \chi \, \chi^\dagger \\ 	
 	&=&   \frac{\imag}{(A^\dagger \boldsymbol{\alpha} A)^2} 
        \left( \begin{array}{cc} \gamma^\ast  |\Xi_1|^2 & |\Xi_1|^2 + \gamma^\ast \Xi_1 \Xi_2^\ast \\[2pt] 
 	 \gamma^\ast \Xi_1^\ast \Xi_2 &
 	\Xi_1^\ast \Xi_2 + \gamma^\ast |\Xi_2|^2 \end{array} \right)
     +  \left( \begin{array}{cc}   |\chi_{11}|^2+|\chi_{21}|^2 &
 	\chi_{11}\chi_{12}^\ast+\chi_{21}\chi_{22}^\ast \\[2pt] 
 	\chi_{11}^\ast\chi_{12}+\chi_{21}^\ast\chi_{22}	 &
 	|\chi_{12}|^2+|\chi_{22}|^2
 	\end{array} \right) \, ,
\eez
where
\bez
 	\Xi = (\Xi_1,\Xi_2)^T = \|A\| \Big( \chi_{1x} - (\frac{1}{2} \Gamma^{-1} +\imag \, \frac{ A^\dagger \boldsymbol{\alpha} A}{\|A\|^2} I_2) \, \chi_1 
 	- \imag \, (\alpha_2-\alpha_1)  \frac{A_1 A_2}{\|A\|^2} \, \chi_2 \Big) \, .
\eez
Hence
\bez
 	\Omega^{-1} = \frac{1}{\det(\Omega)} 
 	\left( \begin{array}{cc} 2 \kappa^{-3} M_{11} - \kappa^{-2} (M_{12}+M_{21}) + \kappa^{-1} M_{22} &
 		\kappa^{-2} M_{11} - \kappa^{-1} M_{12} \\[4pt] 
 		\kappa^{-2} M_{11} - \kappa^{-1} M_{21} &
 		\kappa^{-1} M_{11} \end{array} \right) \, ,
\eez
with
\bez
 	\det(\Omega) = \kappa^{-4} M_{11}^2 + \kappa^{-2} \det(M) \, .
\eez
Inserting the expressions for $\chi_1,\chi_2$ and $\Omega^{-1}$ 
in (\ref{2vFL_sol}), determines a family of quasi-rational solutions of the two-component vector FL equation, which are localized on the plane wave background. 
These are thus rogue waves. They are of second order in the sense that here $\Gamma$ 
is a $2 \times 2$ Jordan block, which causes the appearance of higher (than in the $n=1$ case) order polynomials of the dependent variables. 
\hfill $\Box$
\end{example}

Using an $n \times n$ Jordan block for $\Gamma$, and choosing 
$\gamma$, $A_1$, $A_2$ and $r$ as in the preceding examples, we obtain $n$-th order rogue waves (also see \cite{CMH17} for the NLS case).  
Results about rogue wave solutions of the two-component vector FL equation on the general plane wave background have been obtained earlier in \cite{YZCBG19}. A deeper analysis of the structure of rogue waves (location of peaks) in the triple root case appeared recently in \cite{Ling+Su24}.
\vspace{.2cm}

A different class of rogue waves is obtained in the next example.

\begin{example}
	\label{ex:2vFL_double_root_anticonjugate_pair}
Let $\Gamma = \mathrm{diag}(\gamma,-\gamma^\ast)$ with $\mathrm{Re}(\gamma) \neq 0$, and $r$ given by (\ref{2vFL_triple_root}). Then we have 
$R = \mathrm{diag}(r,-r^\ast)$. Using Example~\ref{ex:2vFL_Peregrine}, the general solution of the linear system is then
\bez
 \chi_{1\kappa} &=& (f_{1\kappa} x^2 + f_{2\kappa} x + f_{3\kappa}) \, e^{\theta_\kappa} \, , \\
 \chi_{2\kappa} &=& \frac{\|A\|^2}{(\alpha_2-\alpha_1) \, A_1 A_2} 
 \Big( \imag \, (1 - A^\dagger \boldsymbol{\alpha} A) \, \gamma_\kappa^{-1} - \frac{\alpha_1 \alpha_2 \|A\|^2}{A^\dagger \boldsymbol{\alpha} A} \Big)^{-1} \Big( \Big[ \frac{\imag}{2 \|A\|^2 \, A^\dagger \boldsymbol{\alpha} A} 
 \big( (\alpha_2-\alpha_1)^2 |A_1|^2 |A_2|^2  \\
 && - (A^\dagger \boldsymbol{\alpha} A)^2 (1-2 A^\dagger \boldsymbol{\alpha} A) \big) \, \gamma_\kappa^{-1} + r_\kappa^2 - \imag \, \frac{A^\dagger \boldsymbol{\alpha}^2 A}{A^\dagger \boldsymbol{\alpha} A} r_\kappa - \frac{1}{4} \gamma_\kappa^{-2} \Big] \, (f_{1\kappa} \, x^2 + f_{2\kappa} \, x + f_{3\kappa}) \\
 && - \big( \imag \,  \frac{A^\dagger \boldsymbol{\alpha}^2 A}{A^\dagger \boldsymbol{\alpha} A}-2r_\kappa) \, (2 f_{1\kappa} \, x + f_{2\kappa}) + 2 f_{1\kappa} \Big) \, e^{\theta_\kappa} \, ,   
\eez
where
\bez
	\left( \begin{array}{c} f_{1\kappa} \\ f_{2\kappa} \\ f_{3\kappa} \end{array} \right)= \left( \begin{array}{ccc} 
		1 & 0 & 0 \\ 2 \tilde{\gamma}_\kappa t/(3\,\alpha_1 \alpha_2) & 1 & 0 \\
		\tilde{\gamma}_\kappa^2 t^2/(3\,\alpha_1 \alpha_2)^2 + 2 \gamma_\kappa t/(\alpha_1 \alpha_2) &
		\tilde{\gamma}_\kappa t/(3 \, \alpha_1 \alpha_2) &  1 \end{array} \right)  
	\left( \begin{array}{c} c_{1\kappa} \\c_{2\kappa} \\ c_{3\kappa}  \end{array} \right)  \, ,  
\eez
with $\kappa=1,2$, 	and  $\gamma_1=\gamma$, $\gamma_2=-\gamma^\ast$, $r_1=r$, $r_2=-r^\ast$, $\theta_1=\theta$, $\theta_2=-\theta^\ast$. Here $r$ is again given by (\ref{2vFL_triple_root}) and $\theta$ by (\ref{2vFL_triple_root_case_theta}).
The Lyapunov equation determines the diagonal entries of $\Omega$,
\bez
	\Omega_{11} = \frac{1}{2\mathrm{Re}(\gamma)}\Big(\frac{\imag \, \gamma^\ast |\Xi_1|^2}{(A^\dagger \boldsymbol{\alpha} A)^2}+|\chi_{11}|^2+|\chi_{21}|^2\Big)\, , \qquad
	\Omega_{22} = \frac{1}{2\mathrm{Re}(\gamma)}\Big( \frac{\imag \, \gamma |\Xi_2|^2}{(A^\dagger \boldsymbol{\alpha} A)^2}-|\chi_{12}|^2-|\chi_{22}|^2\Big)\, ,
\eez
where
\bez
	\Xi_\kappa &=& \|A\| \, \Big( \big(r_\kappa-(\frac{1}{2} \gamma_\kappa^{-1} + \imag \, \frac{A^\dagger \boldsymbol{\alpha} A}{\|A\|^2}) \big) \, (f_{1\kappa} x^2 + f_{2\kappa} x + f_{3\kappa}) + 2 f_{1\kappa} x + f_{2\kappa} \\
	&& - \Big( (1 - A^\dagger \boldsymbol{\alpha} A) \, \gamma_\kappa^{-1} + \imag \, \frac{\alpha_1 \alpha_2 \|A\|^2}{A^\dagger \boldsymbol{\alpha} A} \Big)^{-1} 
	\Big[ \Big( \frac{\imag}{2 \gamma \, \|A\|^2 \, A^\dagger \boldsymbol{\alpha} A} 
	\big( (\alpha_2-\alpha_1)^2 |A_1|^2 |A_2|^2 \\
	&& - (A^\dagger \boldsymbol{\alpha} A)^2 (1-2 A^\dagger \boldsymbol{\alpha} A) \big) 
	+ r_\kappa^2 - \imag \, \frac{A^\dagger \boldsymbol{\alpha}^2 A}{A^\dagger \boldsymbol{\alpha} A} r_\kappa - \frac{1}{4} \gamma_\kappa^{-2} \Big) (f_{1\kappa} x^2 + f_{2\kappa} x + f_{3\kappa}) \\
	&& - \big( \imag \, \frac{A^\dagger \boldsymbol{\alpha}^2 A}{A^\dagger \boldsymbol{\alpha}  A} - 2 r_\kappa) \, (2 f_{1\kappa} x + f_{2\kappa}) +2 f_{1\kappa} \Big] \Big) \, e^{\theta_\kappa} \, ,
\eez
and the constraints
\bez
	\Big( \frac{\imag}{(A^\dagger \boldsymbol{\alpha} A)^2} \Xi \Xi^\dagger \, \Gamma^\dagger
	+ \chi_1 \chi_1^\dagger + \chi_2 \chi_2^\dagger \Big)_{ij} = 0 \qquad \mbox{for} \quad i \neq j \, ,
\eez
which evaluate to
\bez
c_{31}c_{32}^\ast + \frac{\|A\|^4 D_1 D_2}{(\alpha_2-\alpha_1)^2 |A_1|^2 |A_2|^2} + \frac{\imag \, \gamma \, \|A\|^2 \, H_1 H_2}{(A^\dagger \boldsymbol{\alpha} A)^2} = 0 \, .
\eez
Here we introduced
\bez
&& D_1 = d(c_{11},c_{21},c_{31})\, , \qquad D_2 = d(c_{12}^\ast,-c_{22}^\ast,c_{32}^\ast)\, , \\
&& H_1 = h(c_{11},c_{21},c_{31})\, , \qquad H_2 = h(c_{12}^\ast,-c_{22}^\ast,c_{32}^\ast) \, ,
\eez
and
\bez
d(a,b,c) &=& \varrho_3 \, \big(2a-\varrho_1 b+\varrho_4 c\big) \, , \\
h(a,b,c) &=& 2 \, \imag \, \varrho_3 a - (\imag \, \varrho_1 \varrho_3+1) \, b + (\imag \, \varrho_3 \varrho_4 - \varrho_2) \, c \, ,
\eez
with constants
\bez
\varrho_1 &=& \frac{\imag \, A^\dagger \boldsymbol{\alpha}^2 A}{A^\dagger \boldsymbol{\alpha} A}-2r \, , \\
\varrho_2 &=& r-\frac{1}{2}\gamma^{-1}-\frac{\imag \,  A^\dagger \boldsymbol{\alpha} A}{\|A\|^2} \, , \\
\varrho_3 &=& \frac{\gamma A^\dagger \boldsymbol{\alpha} A}{\imag \, A^\dagger \boldsymbol{\alpha} A \, (1 - A^\dagger \boldsymbol{\alpha} A) - \gamma \alpha_1 \alpha_2 \|A\|^2} \, , \\
\varrho_4 &=& \frac{\imag}{2 \|A\|^2 \, A^\dagger \boldsymbol{\alpha} A} 
\big( (\alpha_2-\alpha_1)^2 |A_1|^2 |A_2|^2 - (A^\dagger \boldsymbol{\alpha} A)^2 (1-2 A^\dagger \boldsymbol{\alpha} A) \big) \, \gamma^{-1} + r^2 - \frac{\imag \, A^\dagger \boldsymbol{\alpha}^2 A}{A^\dagger \boldsymbol{\alpha} A} r - \frac{1}{4} \gamma^{-2} \, .
\eez
Integrating the differential equations for $\Omega$, we find the lengthy expressions for the off-diagonal components of $\Omega$ in Appendix~\ref{app:2vFL_double_root_anticonjugate_pair}, containing two additional complex constants $\tilde{c}_{12}, \tilde{c}_{21}$.
The remaining condition (\ref{2vFLeq_pw}) results in
\bez
c_{31} c_{32}^\ast + \frac{\|A\|^4\,D_1D_2}{(\alpha_2-\alpha_1)^2|A_1|^2|A_2|^2} = \gamma \, (\tilde{c}_{12} - \tilde{c}_{21}^\ast) \, .
\eez 
Now (\ref{2vFL_sol}) determines a rather complicated class of quasi-rational solutions of the two-component vector FL equation. We emphasize that, in this case, a rational dependence on the two independent variables $x$ and $t$ enters in two different ways. First, because of the ``degeneracy" of the linear system. Secondly, via integration of (\ref{vFLeq_Om_deriv}), since also the Lypunov equation is ``degenerate", i.e., the spectrum condition is violated. 

We note that $\Omega_{11} \, e^{-2 \mathrm{Re}(\theta)}$ and 
$\Omega_{22} \, e^{2 \mathrm{Re}(\theta)}$ are polynomials of 
degree four in $x$ and $t$, whereas $\Omega_{12}$ and $\Omega_{21}$ are of degree five. It follows that $u_\mu' \sim A_\mu \, e^{\imag \, \varphi_\mu}$ as $x \to \infty$ or $t \to \infty$.	
The resulting solutions generically model rogue waves. 
\hfill $\Box$
\end{example}

\subsubsection{The double root case}
\label{subsec:2vFL_double_root}
First, we consider the case $n=1$ and then the case $n=2$ with $\Gamma$ a $2 \times 2$ Jordan block.
 
\begin{example}
	\label{ex:double_root_n=1}
Let $n=1$. If the cubic characteristic equation (\ref{2vFL_n=1_cubic_eq}) admits a single root $r_1$ and a double root $r_2$, then it is equivalent to
\be
      r_1 = \imag \, (\alpha_1+\alpha_2) + \frac{1}{2 \gamma} - 2 \, r_2  \label{2vFL_double_root_case_r_1}
\ee
and 
\be
|A_1|^2 = \frac{(2 \, r_2 - \imag \, \alpha_2) \, (2 \gamma \, r_2 -1 - 2 \imag \, \alpha_1 \gamma)^2}{4 \alpha_1^2 \gamma \, (\alpha_2 - \alpha_1)} \, , \quad
|A_2|^2 = - \frac{(2   r_2 - \imag \, \alpha_1) \, (2 \gamma \, r_2 -1 - 2 \imag \, \alpha_2 \gamma)^2}{4 \alpha_2^2 \gamma \, (\alpha_2 - \alpha_1)} \, .   \label{2vFL_double_root_case_A_mu}
\ee
The right hand sides of the last two equations have to be real and positive, which imposes constraints on the parameters, which we were unable to cast into a convenient and simple form. The linear system is now solved by
\bez
\chi_1 &=& c_1\, e^{\theta_1} + \Big( c_2 \, x - \frac{c_2\,\gamma}{\alpha_1 \alpha_2} \big( \imag \, (\alpha_1+\alpha_2)-2r_2 \big) \, t + c_3 \Big) \, e^{\theta_2} \, , \\
\chi_2 &=& \frac{\|A\|^2}{(\alpha_2-\alpha_1) \, A_1 A_2} 
\Big( \imag \, (1 - A^\dagger \boldsymbol{\alpha} A) \, \gamma^{-1} - \frac{\alpha_1 \alpha_2 \|A\|^2}{A^\dagger \boldsymbol{\alpha} A}  \Big)^{-1} \Big( \chi_{1xx} - \imag \, \frac{A^\dagger \boldsymbol{\alpha}^2 A}{A^\dagger \boldsymbol{\alpha} A} \, \chi_{1x} - \frac{1}{4} \gamma^{-2} \chi_1  \\
&& + \frac{\imag}{2 \|A\|^2 \, A^\dagger \boldsymbol{\alpha} A} 
\big( (\alpha_2-\alpha_1)^2 |A_1|^2 |A_2|^2 
- (A^\dagger \boldsymbol{\alpha} A)^2 (1-2 A^\dagger \boldsymbol{\alpha} A) \big) \, \gamma^{-1} \chi_1 \Big) \, ,   
\eez
with complex constants $c_1,c_2,c_3$ and 
\bez
\theta_\kappa &=& r_\kappa x+\frac{1}{\alpha_1 \alpha_2 \gamma} \Big( \big( r_\kappa^2-\imag \, (\alpha_1+\alpha_2) \,r_\kappa - \frac{1}{2} \alpha_1 \alpha_2 \big) \, \gamma^2 + \imag \, (\alpha_1+\alpha_2) \, (A^\dagger \boldsymbol{\alpha} A-\frac{1}{2}) \, \gamma-\frac{1}{4} \Big) \, t \, .
\eez
\textbf{(1)} Assuming $\mathrm{Re}(\gamma) \neq 0$, so that the spectrum condition for the Lyapunov equation is satisfied, the resulting solution  of the two-component vector FL equation is
\bez
u_1' = e^{\imag \, \varphi_1}\big(A_1+\frac{ (A_1 \chi_1^\ast - A_2^\ast \chi_2^\ast) \,   \Xi}{\|A\| \, A^\dagger \boldsymbol{\alpha} A \, \Omega} \big) \, , \qquad
u_2' = e^{\imag \, \varphi_2}\big(A_2+\frac{ (A_2 \chi_1^\ast + A_1^\ast \chi_2^\ast)\, \Xi}{\|A\| \, A^\dagger \boldsymbol{\alpha} A \, \Omega} \big) \, ,
\eez
where
\bez
\Omega &=& \frac{\imag \, \gamma^\ast |\Xi|^2 + (A^\dagger \boldsymbol{\alpha} A)^2(|\chi_1|^2+|\chi_2|^2)}{2 \mathrm{Re}(\gamma) \,|(A^\dagger \boldsymbol{\alpha} A)^2} \, , \\
\Xi &=& \|A\| \, \Big( \chi_{1x} - \big( \frac{1}{2} \gamma^{-1} +\imag \, \frac{ A^\dagger \boldsymbol{\alpha} A}{\|A\|^2} \big) \, \chi_1 
- \imag \, (\alpha_2-\alpha_1)  \frac{A_1 A_2}{\|A\|^2} \, \chi_2 \Big) \, .
\eez
If $c_1=0$, this solution is quasi-rational. If also $c_2 \neq 0$ and $r_2 \neq \frac{1}{2} \imag \, (\alpha_1+\alpha_2)$, the latter is localized in space and time on the plane wave background. It thus represents a rogue wave. 
   \\
\textbf{(2)}	If $\gamma$ is imaginary, i.e., $\gamma=-\imag \, k$ with $k \in \mathbb{R} \setminus \{0\}$, we have
\bez
	|A_1|^2 &=& \frac{(2 \imag \, r_2 + \alpha_2) \, (2 \imag \,  k \, r_2 + 1 + 2 \alpha_1 k)^2}{4 \alpha_1^2 \, k \, (\alpha_2 - \alpha_1)} \, , \qquad 	
	|A_2|^2 = - \frac{(2 \imag \,  r_2 + \alpha_1) \, (2 \imag \, k \, r_2 + 1 + 2 \alpha_2 k)^2}{4 \alpha_2^2 \, k \, (\alpha_2-\alpha_1)} \, .
\eez
The imaginary parts vanish if and only if $\mathrm{Re}(r_2) =0$. Hence  
\bez 
    r_2 = \imag \, s \, , \qquad s \in \mathbb{R} \, ,
\eez           
and thus
\bez
	|A_1|^2 &=& \frac{(\alpha_2-2 s) \, (2 k s -1 - 2 \alpha_1 k)^2}{4 \alpha_1^2 \, k \, (\alpha_2-\alpha_1)} \,  , \qquad 	
	|A_2|^2 = - \frac{(\alpha_1-2 s) \, (2 k s - 1 - 2 \alpha_2 k)^2}{4 \alpha_2^2 \, k \, (\alpha_2-\alpha_1)} \, ,
\eez
which requires
\bez
	\frac{\alpha_2 - 2 s}{k \, (\alpha_2-\alpha_1)} > 0 \, , \qquad \frac{\alpha_1 - 2 s}{k \, (\alpha_2-\alpha_1)} < 0 \, ,
\eez
and implies $(\alpha_1 - 2 s)(\alpha_2 - 2 s) < 0$.
The Lyapunov equation can only be satisfied if either $c_1=0$ or $c_2=c_3=0$.

\noindent
\textbf{(i)} $c_1=0$. The Lyapunov equation reduces to an equation, which  does not dependent on the independent variables,    
\bez
	(A^\dagger \boldsymbol{\alpha} A)^2 \, \Big( |c_3|^2 
	+ \frac{\|A\|^4 \, |\tilde{c}|^2}{|A_1|^2 \, |A_2|^2 \, (\alpha_2-\alpha_1)^2}
	 \Big) = k \, \|A\|^2 \, \Big| \imag \, c_2 - c_3 \, s + c_3 \, \big( \frac{1}{2 k} + \frac{A^\dagger \boldsymbol{\alpha} A}{\|A\|^2} \big) - \tilde{c} \, \Big|^2 \, ,
\eez
where
\bez
	\tilde{c} &=& \Big( (1 - A^\dagger \boldsymbol{\alpha} A) \, k^{-1} + \frac{\alpha_1 \alpha_2 \|A\|^2}{A^\dagger \boldsymbol{\alpha} A}  \Big)^{-1} \Big( 2 \imag \, c_2 \, s - c_3 \, s^2 - \imag \, \frac{A^\dagger \boldsymbol{\alpha}^2 A}{A^\dagger \boldsymbol{\alpha} A}(c_2 + \imag \, c_3 \, s) + \frac{1}{4} k^{-2} c_3 \\
	&& - \frac{1}{2 A^\dagger \boldsymbol{\alpha} A \, \|A\|^2} \big( (\alpha_2-\alpha_1)^2 |A_1|^2 |A_2|^2 - (A^\dagger \boldsymbol{\alpha} A)^2 (1 - 2 A^\dagger \boldsymbol{\alpha} A) \big) \, k^{-1} c_3 \Big) \, .
\eez	
Solving the differential equations for $\Omega$, we obtain
\bez
\Omega &=& \frac{\|A\|^2}{ k (A^\dagger \boldsymbol{\alpha} A)^2 \, \varrho_2^2} \Big( \frac{ \imag \, k \,\tilde{s} \, |c_2|^2}{3 \alpha_1 \alpha_2} \varrho_3 \, (\imag \, A^\dagger \boldsymbol{\alpha} A \, \varrho_2 - \varrho_3) \, (3 x t^2 - \frac{3 \imag \, \alpha_1 \alpha_2}{k \tilde{s}} x^2 t + \frac{\imag \, k \tilde{s}}{\alpha_1 \alpha_2} t^3)\\
&& - \imag \, \big( k |c_2|^2 A^\dagger \boldsymbol{\alpha} A \, \varrho_2 ( \varrho_2 - k \varrho_1) + 2 \imag \, \mathrm{Re}(c_2 c_3^\ast) \, \varrho_3 (\imag \, A^\dagger \boldsymbol{\alpha} A \, \varrho_2 - \varrho_3) \big) \, ( x t + \frac{\imag \, k \tilde{s}}{2 \alpha_1 \alpha_2} t^2) \\
&& - \big( c_3^\ast \, (\imag \, A^\dagger \boldsymbol{\alpha} A \, \varrho_2 - \varrho_3) - k c_2^\ast \, (\varrho_2 - k \varrho_1) \big) \, (k c_2\, (\varrho_2 - k \varrho_1) - c_3 \varrho_3) \, t \\
&& - \frac{|c_2|^2 \varrho_5}{6 k A^\dagger\boldsymbol{\alpha} A \, \|A\|^2} \, x^3 - \frac{1}{4 k A^\dagger \boldsymbol{\alpha} A \, \|A\|^2} \big(2 \mathrm{Re}(c_2 c_3^\ast) \, \varrho_5 - \imag \, k |c_2|^2 (A^\dagger \boldsymbol{\alpha} A)^2 \varrho_2 \varrho_6 \big) \, x^2 \\
&& + \frac{1}{k A^\dagger \boldsymbol{\alpha} A \, \|A\|^2 \varrho_4^2} \big( \|A\|^4 |c_2|^2 \varrho_1^2 \varrho_7 + k c_2 (A^\dagger \boldsymbol{\alpha} A)^2 \varrho_2 \varrho_4 \, (A^\dagger \boldsymbol{\alpha} A \, \varrho_2 - k \varrho_4) \, (\imag \, c_3^\ast \varrho_4 + c_2^\ast \|A\|^2 \varrho_1) \\
&& + \imag \, k |c_2|^2 A^\dagger \boldsymbol{\alpha} A \, \|A\|^2 \varrho_4^2 \, (k \varrho_1 - \varrho_2) \, \big( \imag \, k s \, (k \varrho_1-\varrho_2) + \varrho_3 \big) + \varrho_8 \big) \, x + c \Big) \, ,
\eez
with
\bez
\tilde{s} &=& \imag \, (\alpha_1 + \alpha_2 - 2 s) \, , \\
\varrho_1 &=& - 2 \, s\, A^\dagger \boldsymbol{\alpha} A  + A^\dagger \boldsymbol{\alpha}^2 A \, , \\
\varrho_2 &=& A^\dagger \boldsymbol{\alpha} A \, (A^\dagger \boldsymbol{\alpha} A - 1) - k \, \alpha_1 \alpha_2 \|A\|^2 \, , \\
\varrho_3 &=&-\imag \, k^2 \, \big(   A^\dagger \boldsymbol{\alpha} A \, s^2 -  (\alpha_1 \alpha_2 \, \|A\|^2 + A^\dagger \boldsymbol{\alpha}^2 A)\,s + \alpha_1 \alpha_2 \, A^\dagger \boldsymbol{\alpha} A \big)  \\
&& - \imag \, k \, \big( s \, (A^\dagger \boldsymbol{\alpha} A)^2 + \frac{1}{2} \alpha_1 \alpha_2 \|A\|^2 + \frac{1}{2} \varrho_1 \big)  + \frac{1}{4} \imag \, A^\dagger \boldsymbol{\alpha} A \, (2 A^\dagger \boldsymbol{\alpha} A - 1) \, , \\
\varrho_4 &=& \frac{1}{2 k^2} \Big( \|A\|^2 \, \big( 2 \imag \, (\imag \, k s \, \varrho_2 + \varrho_3) + \varrho_2 \big) + 2 k A^\dagger \boldsymbol{\alpha} A \, \varrho_2 \Big) \, , \\
\varrho_5 &=& (A^\dagger \boldsymbol{\alpha} A)^2 \varrho_2 \, (1 
- 2 k s) \, (k \varrho_4 - A^\dagger \boldsymbol{\alpha} A \, \varrho_2) - 2 \alpha_1 \alpha_2 \|A\|^2 (k (A^\dagger \boldsymbol{\alpha} A)^2 \varrho_2^2 + \|A\|^2 \varrho_3^2) \, , \\
\varrho_6 &=& \|A\|^2 \varrho_1 \, (1 - 2 k s) - 2 (k \varrho_4 - A^\dagger \boldsymbol{\alpha} A \, \varrho_2) \, , \\
\varrho_7 &=& ( \alpha_1 \alpha_2 \|A\|^2 - s \, A^\dagger \boldsymbol{\alpha} A ) \, (k \, (A^\dagger \boldsymbol{\alpha} A)^2 \varrho_2^2 + \|A\|^2 \varrho_3^2) \, , \\
\varrho_8 &=& |c_2|^2 A^\dagger \boldsymbol{\alpha} A \, \|A\|^4 \varrho_1 \varrho_3^2 \varrho_4 + \frac{1}{2} \varrho_4 \, \Big( 4 \imag \, \mathrm{Im}(c_2 c_3^\ast) \, \big( k \, s \, A^\dagger \boldsymbol{\alpha} A \, \|A\|^2 \varrho_3 \varrho_4 \, (k \varrho_1 - \varrho_2) + \imag \, \|A\|^2 \varrho_1 \varrho_7 \big) \\
&& + \imag \, k \, c_2^\ast c_3 \, (A^\dagger \boldsymbol{\alpha} A)^2 \|A\|^2 \varrho_1 \varrho_2 \varrho_4 \, (1 - 2 k s) - |c_3|^2 \varrho_4 \varrho_5 \Big) \, ,
\eez
and a constant $c$ for which (\ref{2vFLeq_pw}) requires
\bez
\mathrm{Im}(c) = \frac{1}{2 k \varrho_4^2} \Big( \varrho_1 \, \big( k (A^\dagger \boldsymbol{\alpha} A)^2 \varrho_2^2 + \|A\|^2 \varrho_3^2 \big) \, \big( 2 \mathrm{Im}(c_2 c_3^\ast) \varrho_4 - |c_2|^2 \|A\|^2 \varrho_1 \big) - |c_3|^2 \varrho_3^2 \varrho_4^2 \Big) \, .
\eez
An example from the resulting class of solutions of the two-component vector FL equation is presented in Fig.~\ref{fig:2vFL_double_root_imag_gamma}.
\begin{figure}[h]
	\begin{center}
	\includegraphics[scale=.4]{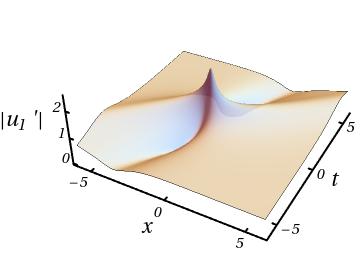}
		\hspace{1cm}
	\includegraphics[scale=.4]{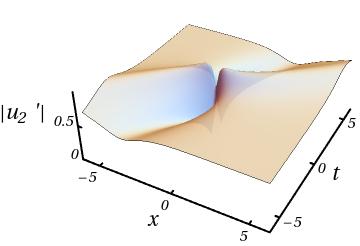}
		\parbox{15cm}{
	\caption{Plots of the absolute values of the two components of a ``double pole  solution" from the class in Example~\ref{ex:double_root_n=1} (2). 
	Here we set $\alpha_1 = 1, \alpha_2 =2, k=1, s = 3/4, c_1=0,c_2=1+\imag$ and $c_3=1+(353/1365) \, \imag$, so that the constraint from the Lyapunov equation is satisfied. 
				\label{fig:2vFL_double_root_imag_gamma} } 
		}
	\end{center}
\end{figure}

\noindent
\textbf{(ii)} $c_2=c_3=0$. The Lyapunov equation requires $c_1=0$, so that $u_1'$ and $u_2'$ reduce to the plane wave background solutions. 
\hfill $\Box$
\end{example}

\begin{remark}
	\label{rem:double_root_Re(r_2)=0}
Let $c_1=0$ and $r_2 = \frac{1}{2} \imag \, (\alpha_1+\alpha_2)$ in Example~\ref{ex:double_root_n=1}. 
Then $\chi_1$, $\chi_2$, and thus also $u'$, show a rational dependence only on $x$. 
This particular choice for $r_2$ actually follows from (\ref{2vFL_double_root_case_A_mu})  
by just assuming $\mathrm{Re}(r_2) = 0$ and $\mathrm{Re}(\gamma) \neq 0$.
Rogue waves thus require $\mathrm{Re}(r_2) \neq 0$.

In the case $\mathrm{Re}(r_2) = 0$,  (\ref{2vFL_double_root_case_A_mu}) further amounts to 
$\mathrm{Im}(\gamma) = (1-2 \alpha_2 |A_2|^2)/(\alpha_2-\alpha_1)$,
$\mathrm{Re}(\gamma)^2 = 4 \left(-\alpha_1 |A_1|^2 + \alpha_2  |A_2|^2 \left( \alpha_2 |A_2|^2 \left( 3 - 2 \alpha_2  |A_2|^2 \right) - 2 \right) + 1 \right) / ((\alpha_2-\alpha_1)^2 (2 \alpha_2 |A_2|^2 -1))$,
and $\alpha_1 |A_1|^2 + \alpha_2 |A_2|^2 = 1$. 
\hfill $\Box$
\end{remark}

Next we present a second-order version of the above class of solutions in Example~\ref{ex:double_root_n=1}. 

\begin{example}
	\label{ex:double_root_n=2Jordan}
Let $n=2$  and 
\bez
\Gamma= \left( \begin{array}{cc} 
	\gamma
	&  0  \\  1  & \gamma \end{array} \right)\,  , \qquad
\chi_i = \left( \begin{array}{c} \chi_{i1} \\ \chi_{i2} 
\end{array} \right) \qquad i=1,2 \, .
\eez 
Again, we assume that $r_1$ is a single root and $r_2$ is a double root of the cubic equation (\ref{2vFL_n=1_cubic_eq}). Then $r_1$ and $A_\mu$, $\mu=1,2$, are given by (\ref{2vFL_double_root_case_r_1}) and (\ref{2vFL_double_root_case_A_mu}), respectively. Now the linear system is solved by
\bez
\chi_{11} &=& c_1 e^{\theta_1} + (f_1 x + f_2) \, e^{\theta_2} \, , \\
\chi_{12} &=& \Big( \frac{c_1 (v_1-v_2^2)}{2v_2^2 \gamma^2} x + \frac{c_1 v_3}{2 \alpha_1 \alpha_2 v_2^2 \gamma^2} t + c_6 \Big) \, e^{\theta_1} \\
&& + \Big( \frac{f_1 v_4}{24 v_2 \gamma^3} x^3 + \frac{2 f_1 \gamma \, (v_2 v_5 +v_4) +f_2 v_2 v_4}{8 v_2^2 \gamma^3} x^2 + f_3 x + f_4 \Big) \, e^{\theta_2} \, , \\
\chi_2 &=& \frac{\|A\|^2}{(\alpha_2-\alpha_1) \, A_1 A_2} 
\Big( \imag \, (1 - A^\dagger \boldsymbol{\alpha} A) \, \Gamma^{-1} - \frac{\alpha_1 \alpha_2 \|A\|^2}{A^\dagger \boldsymbol{\alpha} A} I_n \Big)^{-1} \Big( \chi_{1xx} - \imag \, \frac{A^\dagger \boldsymbol{\alpha}^2 A}{A^\dagger \boldsymbol{\alpha} A} \, \chi_{1x} - \frac{1}{4} \Gamma^{-2} \chi_1  \\
&& + \frac{\imag}{2 \|A\|^2 \, A^\dagger \boldsymbol{\alpha} A} 
\big( (\alpha_2-\alpha_1)^2 |A_1|^2 |A_2|^2 
- (A^\dagger \boldsymbol{\alpha} A)^2 (1-2 A^\dagger \boldsymbol{\alpha} A) \big) \, \Gamma^{-1} \chi_1 \Big)  \, ,   
\eez
where
\bez
v_1 &=& 8\gamma^3\big(2r_2-\imag\,(\alpha_1+\alpha_2)\big)\,\big(4r_2^2-2\imag\,(\alpha_1+\alpha_2)\,r_2-\alpha_1\alpha_2\big) \, , \\
v_2 &=& 1-2\gamma\,\big(3r_2-\imag\,(\alpha_1+\alpha_2)\big) \, , \\
v_3 &=& v_2^2 \gamma^2\,\big( 8 r_2^2 - 4 \imag \, (\alpha_1+\alpha_2) \,r_2 - \alpha_1 \alpha_2 \big)  +v_1\,\Big(1-\gamma\,\big(4r_2-\imag\,(\alpha_1+\alpha_2)\big)\Big)\, ,\\
v_4 &=&  (2 r_2 \gamma-1) \Big( 4\gamma^2  \big(r_2^2-\imag \, (\alpha_1+\alpha_2) \, r_2 - \alpha_1 \alpha_2\big) + 2\gamma \,\big(\imag \, (\alpha_1+\alpha_2) -2 r_2\big)  + 1 \Big) \, , \\
v_5 &=&1+ 4\gamma^2\, \big( 3 r_2^2 - 2 \imag \, (\alpha_1+\alpha_2) \, r_2 - \alpha_1 \alpha_2 \big)   \, ,
\eez
and  $f = (f_1,\ldots,f_4)^T$ satisfies $f_t = C f$ with the constant nilpotent matrix
\bez
C = -\frac{1}{ \alpha_1\alpha_2} \,\left( \begin{array}{cccc}
	0 & 0 & 0 & 0  \\
	\tilde{s} \, \gamma & 0 & 0 & 0 \\
	\frac{1}{4} \big(B_1 v_2 + 2 (v_2 v_5 + v_4) \, \tilde{s} \gamma\big) v_2^{-2} \gamma^{-2} & \frac{1}{4}  v_2^{-1} v_4 \tilde{s} \gamma^{-2} & 0 & 0  \\
	\frac{1}{2} B_2 v_2^{-2} \gamma^{-1} & \frac{1}{4} B_1 v_2^{-1} \gamma^{-2} & \tilde{s} \, \gamma & 0 
\end{array} \right)\, ,
\eez
where
\bez
\tilde{s} &=& \imag \, (\alpha_1+\alpha_2)- 2 r_2 \, , \\
B_1 &=& -2 (2 r_2^2 - 2 \imag \, (\alpha_1+\alpha_2) \, r_2 - \alpha_1 \alpha_2) \, v_2 \gamma^2 - (v_2+v_4) \, , \\
B_2 &=& 2 v_2^2 \tilde{s} \gamma - v_2 v_5 - v_4 \, .
\eez
Consequently, we have
\bez
f_1 &=& c_2 \, ,\\
f_2 &=& -\frac{c_2 \tilde{s} \gamma }{\alpha_1 \alpha_2} \, t + c_3 \, , \\
f_3 &=& \frac{1}{8 (\alpha_1 \alpha_2 v_2 \gamma)^2} \Big(c_2 v_2 v_4 \tilde{s}^2 \gamma \, t^2 - 2 \alpha_1 \alpha_2 \big( (v_2\, (2 c_2 v_5 \gamma + c_3 v_4)  + 2 c_2 v_4 \gamma ) \, \tilde{s} + c_2 B_1 v_2 \big) \, t \\
&& + 8 c_4 (\alpha_1 \alpha_2 v_2 \gamma)^2 \Big) \, , \\
f_4 &=& \frac{1}{24 (\alpha_1 \alpha_2)^3 (v_2 \gamma)^2} \Big(-c_2 v_2 v_4 \tilde{s}^3 \gamma^2 \, t^3 + 3 \big( (2c_2 \gamma\, (v_2 v_5 + v_4)  +c_3 v_2 v_4 ) \, \tilde{s} + 2 c_2 B_1 v_2 \big) \, \alpha_1 \alpha_2 \tilde{s} \gamma \, t^2 \\
&& -6 \alpha_1^2 \alpha_2^2 \, (2 c_2 B_2 \gamma + c_3 B_1 v_2 + 4 c_4 v_2^2 \tilde{s} \gamma^3) \, t + 24 c_5 (\alpha_1 \alpha_2)^3 (v_2 \gamma)^2 \Big) \, ,
\eez
with complex constants $c_1,\ldots,c_6$. Assuming  $\mathrm{Re}(\gamma) \neq 0$, the expressions for $\Omega$ are obtained via those in Example~\ref{ex:triple_root_n=2}. Inserting the results in (\ref{2vFL_sol}), then yields a class of solutions of the two-component vector FL equation. We disregard the case 
$\mathrm{Re}(\gamma) = 0$ because of its complexity.
\hfill $\Box$
\end{example}

\begin{example}
Choosing the special data of Remark~\ref{rem:double_root_Re(r_2)=0} in the class of solutions in Example~\ref{ex:double_root_n=2Jordan}, a corresponding example of a solution of the two-component vector FL equation is displayed 
in Fig.~\ref{fig:double_root_Re(r_2)=0_n=2}. 
\begin{figure}[h]
	\begin{center}
		\includegraphics[scale=.4]{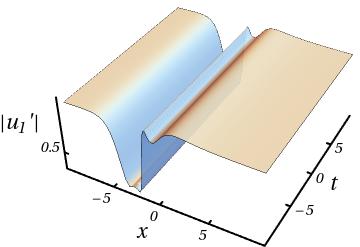} 
		\hspace{1cm}
		\includegraphics[scale=.4]{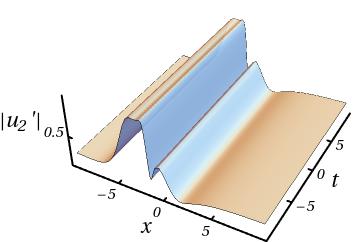} 
		\parbox{15cm}{
			\caption{Plots of the absolute values of the two components of a solution from the class in Example~\ref{ex:double_root_n=2Jordan}, with data according to the special case in Remark~\ref{rem:double_root_Re(r_2)=0}. We chose $\alpha_2=2$, $r_2 = 3 \imag/2$ (so that $\alpha_1=1$), $A_1=\sqrt{3}/2$, $A_2=\sqrt{2}/4$, $\mathrm{Im}(\gamma)=\sqrt{3}/2$, $c_1=c_6=0$ and $c_2 = \cdots = c_5 =1$.  
				\label{fig:double_root_Re(r_2)=0_n=2} } 
		}
	\end{center} 
\end{figure} 	
\hfill $\Box$	
\end{example}

Also see \cite{YZCBG19} for rogue wave solutions in the double 
root case.

\section{Conclusions}
\label{sec:concl}
We derived a \emph{vectorial} binary Darboux transformation (bDT) for matrix versions of the first negative flow of the Kaup-Newell (KN) hierarchy. 
For a reduction of the latter system to matrix generalizations (see (\ref{mFLeq})) of the Fokas-Lenells (FL) equation, we found a corresponding reduction of the vectorial bDT to a vectorial Darboux transformation (DT) of the latter and generated various classes of matrix soliton solutions of (\ref{mFLeq}) from the trivial seed solution $u=0$.  

Furthermore, we explored in detail the vectorial DT with a plane wave seed solution in the case of the two-component vector FL equation, recovering in a different and, in our opinion, more systematic way, classes of solutions obtained previously in \cite{LFZ18,YZCBG19,Ling+Su24,SFY25}. But we also presented many solutions that have not been obtained before to the best of our knowledge. 
This concerns especially solutions corresponding to cases, where the matrix $\Gamma$ in the Lyapunov equation does not satisfy the spectrum condition, so that the Lyapunov equation does not have a unique solution.  

The present analysis of the 2-component FL equation is actually so far the most systematic exploration of exact solutions of an integrable equation using a vectorial DT.

The physical relevance of the two-component vector FL equation has been 
emphasized in particular in \cite{TZLBMYAB22}. Besides ultrashort light pulse propagation in birefringent optical ﬁbers, it can model crossing sea waves in the open ocean. Of course, also our new classes of solutions could, in principle,  describe wave motions in such systems. In practise, however, it will be difficult to experimentally set up corresponding initial conditions that would evolve into such solutions. Furthermore, as every mathematical model, an equation like the vector FL equation is never an exact description of a physical system and it is difficult to figure out where the limits are. 
From this perspective, the new classes of solutions could provide corresponding tests. Before, the stability of these solutions should 
be explored, analytically or numerically.

We also would like to highlight again the nonlinear superposition 
property of a \emph{vectorial} Darboux transformation, for solutions generated from the same seed solution, 
as formulated in Remarks \ref{rem:FL_superposition} and \ref{rem:vFL_superposition} for the specific integrable equations studied in this work.  
In our opinion, this is a major advantage over the classical iterative Darboux transformation method.

An interesting feature, known to show up in vector NLS equations, is ``beating solitons" (see \cite{Park+Shin00,HCHE11,Gela+Rask23,LCA24}, and references cited there). We presented corresponding solutions 
for the two-component vector FL equation in Examples~\ref{ex:special_pw_breather} and \ref{ex:special_pw_beating_solitons}, apparently for the first time. As in the cited work 
on beating soliton solutions of vector NLS equations, here the relative wave number has been restricted to zero. 
Similar results have been obtained in \cite{Park+Shin00} for the Manakov system (i.e., the two-component vector NLS equation), also see \cite{Gela+Rask23,LCA24}.  

The ``chirped solitons" found in \cite{TZLBMYAB22} for the two-component 
vector FL equation, also see Appendix~\ref{app:chirped}, should rather be called ``chirping travelling wave" solutions or ``chirping solitary waves", since we doubt that they can be superposed, which is the characteristic property of solitons. They cannot be reached by a Darboux transformation from the seed solutions that we 
considered. These are interesting new solutions. In particular,  they can be chosen as a new seed in order 
to ``decorate" it with a new type of solitons via Darboux transformations.  

The stage is also set to generate solutions of vector FL equations with more components \cite{Wang+Chen19,SFY25}, and of matrix versions beyond the vector case, from a non-vanishing seed solution. 
In the matrix (not vector) case, we then have to deal with a linear system having left as well as right matrix coefficients, so that a vectorization map is required in order to solve it. We shall return to this in a separate work.

Finally, we presented further classes of integrable equations that fit into the bidifferential calculus framework. 

For generating the plots in this work and verifying examples of solutions, we used \textit{Mathematica} \cite{Mathematica}.

\makeatletter
\newcommand\appendix@section[1]{
	\refstepcounter{section}
	\orig@section*{Appendix \@Alph\c@section: #1}
	\addcontentsline{toc}{section}{Appendix \@Alph\c@section: #1}
}
\let\orig@section\section
\g@addto@macro\appendix{\let\section\appendix@section}
\makeatother

\begin{appendix}
\numberwithin{equation}{section}	

\section{Some facts about the Lyapunov equation}
\label{app:Lyapunov}
We consider the Lyapunov equation 
\be
   \Gamma \, \Omega + \Omega \, \Gamma^\dagger = K \, , 
\label{Lyap_K}
\ee	
with $n \times n$ matrices $\Gamma, \Omega, K$. Here we think 
of $\Gamma$ and $K$ as given matrices and look for a solution 
$\Omega$.

\paragraph{(1)} Let $\Gamma = \mathrm{diag}(\gamma_1,\ldots,\gamma_n)$. Then (\ref{Lyap_K}) reads
\bez
    (\gamma_i+\gamma_j^\ast) \, \Omega_{ij} = K_{ij} \qquad i,j=1,\ldots,n \, .
\eez
If $\Gamma$ satisfies the spectrum condition, i.e. $\gamma_i \neq - \gamma_j^\ast$, 
$i,j=1,\ldots,n$, the Lyapunov equation possesses a unique solution, 
given by the generalized Cauchy matrix
\bez
   \Omega = \left( \frac{K_{ij}}{\gamma_i + \gamma_j^\ast} \right) \, .
\eez
If, for some values of $i$ and $j$, $\gamma_i = -\gamma_j^\ast$, then a solution can only exist if the 
corresponding component $K_{ij}$ of $K$ vanishes, and $\Omega_{ij}$ is left undetermined. 

$\Omega$ is a Cauchy-like matrix if $K = u v^T$, with $n$-component 
column vectors $u,v$. Then
\bez
     \det(\Omega) = \big( \prod_{i=1}^n u_i \big) \, \big( \prod_{j=1}^n v_j \big) \, \det(\boldsymbol{C}) \, ,
\eez
where $\boldsymbol{C}$ is the $n \times n$ Cauchy matrix with components $\boldsymbol{C}_{ij} = 1/(\gamma_i + \gamma_j^\ast)$, and 
a concise expression for its determinant is well-known. $\boldsymbol{C}$ is invertible if and only if the spectrum condition for $\Gamma$ holds. If $\Omega$ is \emph{not} Cauchy-like, which is 
the relevant case in our work, invertibility is a more difficult 
question, but see \cite{Hein95}, for example. 

\begin{example}
For $n=2$ we have
\bez
    \det(\Omega) = \frac{K_{11} K_{22}}{4 \mathrm{Re}(\gamma_1) \,\mathrm{Re}(\gamma_2)} - \frac{K_{12} K_{21}}{|\gamma_1 + \gamma_2^\ast|^2} \, .  
\eez
If $\Omega$ is Cauchy-like, then this greatly simplifies to
\bez
  \det(\Omega) = u_1 u_2 v_1 v_2 \, \frac{|\gamma_1 - \gamma_2^\ast|} {4 \mathrm{Re}(\gamma_1) \, \mathrm{Re}(\gamma_2) \, |\gamma_1 + \gamma_2^\ast|^2} \, ,
\eez  
so that, in contrast to the general case, the invertibility condition is very simple.
\hfill $\Box$
\end{example}

\paragraph{(2)} 
Let $\Gamma$ be the lower triangular $n \times n$ Jordan block
\be
\Gamma = \left( \begin{array}{ccccc} \gamma  & 0      & \cdots & \cdots & 0      \\ 
	1  & \gamma  & \ddots  & \ddots & 0      \\ 
	0  & \ddots & \ddots & \ddots & \vdots \\
	\vdots & \ddots & \ddots & \ddots & 0      \\
	0  & \cdots &  0  &   1    & \gamma 
\end{array} \right)  \, ,   \label{Gamma_Jordan_block} 
\ee
with $\gamma \in \mathbb{C}$. 
If the spectrum condition holds, we have the following result \cite{CMH17}. 

\begin{proposition}
Let $\mathrm{Re}(\gamma) \neq 0$. For $1 \leq k < n$, let $K_{(k)}$ and $\Gamma_{(k)}$ be the $k \times k$ upper left part of $K$, respectively $\Gamma$, 
and $\Omega_{(k)}$ the solution of 
\bez
	\Gamma_{(k)} \, \Omega_{(k)} + \Omega_{(k)} \, \Gamma_{(k)}^\dagger = K_{(k)} \, . 
\eez
For $k=1,\ldots,n-1$, we have
\bez
	\Omega_{(k+1)} = \left( \begin{array}{cc}
		\Omega_{(k)} & B_{(k)}  \\
		B_{(k)}^\dagger & \omega_{(k+1)} \end{array} \right) \, ,                    
\eez
where
\bez
	&& B_{(k)} := \hat{\Gamma}_{(k)}^{-1} \, \Big( 
	(K_{1,k+1}, K_{2,k+1}, \ldots, K_{k,k+1})^T
	- \Omega_{(k)} (0,\ldots,0,1)^T \Big) \, , \\
	&& \omega_{(k+1)} := \frac{1}{\kappa} \Big( K_{k+1,k+1}
	- 2 \, \mathrm{Re}[(0,\ldots,0,1) B_{(k)}] \Big) \, ,  
\eez
and $\hat{\Gamma}_{(k)}$ is the Jordan block $\Gamma_{(k)}$ with $\gamma$ replaced by $\kappa = 2 \, \mathrm{Re}(\gamma)$.  
\hfill $\Box$
\end{proposition}

This proposition allows to recursively compute the solution of the Lyapunov equation (\ref{Lyap_K}) with a Jordan 
matrix $\Gamma_{(n)}$. The inverse of $\Omega_{(n)}$ can also be recursively computed via 
\be
\Omega_{(k+1)}^{-1} = \left( \begin{array}{cc}
	\Omega_{(k)}^{-1} + S_{\Omega_{(k)}}^{-1} \Omega_{(k)}^{-1} B_{(k)} B_{(k)}^\dagger \Omega_{(k)}^{-1}
	& - S_{\Omega_{(k)}}^{-1} \Omega_{(k)}^{-1} B_{(k)}  \\
	- S_{\Omega_{(k)}}^{-1} B_{(k)}^\dagger \Omega_{(k)}^{-1} &  S_{\Omega_{(k)}}^{-1} 
\end{array} \right) \, ,   \label{Omega_inverse}
\ee
with the scalar Schur complement (a quasi-determinant) 
\bez
S_{\Omega_{(k)}} = \omega_{(k+1)} - B_{(k)}^\dagger \Omega_{(k)}^{-1} B_{(k)} \, .
\eez

\begin{example}
	\label{ex:n=2Jordan_Lyapunov_sol}
For $n=2$, we obtain
\bez
	\Omega_{(2)} = \frac{1}{ \kappa } \left( \begin{array}{cc}
		K_{11} & K_{12} - \kappa^{-1} K_{11}  \\
		K_{21} - \kappa^{-1} K_{11} &
		K_{22} - \kappa^{-1} (K_{12} + K_{21}) + 2 \kappa^{-2} K_{11} 
	\end{array} \right)  \, .
\eez
Its determinant is
\be
	\det(\Omega_{(2)}) = \kappa^{-4} \, K_{11}^2  
	+\kappa^{-2} \, \det(K)  \, . \label{detOm(2)}
\ee
\hfill $\Box$
\end{example}

\begin{proposition}
Let $\Gamma$ be the $n \times n$ lower triangular Jordan block 
with imaginary eigenvalue $\gamma$.	Then the Lyapunov equation
(\ref{Lyap_K}) is equivalent to
\bez
&& K_{11} = 0 \, , \\
&& \Omega_{i-1,1} = K_{i1} \qquad i=2,3,\ldots,n \, , \\ 
&& \Omega_{1,j-1} = K_{1j} \qquad j=2,3,\ldots,n \, , \\ 
&& \Omega_{i-1,j} + \Omega_{i,j-1} = K_{ij} \qquad i,j = 2,3,\ldots,n \, . 
\eez
We note that $\Omega_{nn}$ drops out of the Lyapunov equation.	
\end{proposition}
\begin{proof}	
Expressed in components, the Lyapunov equation reads
\bez
   \sum_{k=1}^n \Big( \Gamma_{ik} \Omega_{kj} + \Omega_{ik} \Gamma^\dagger_{kj} \Big) 
   = \Gamma_{ii} \Omega_{ij} + \Gamma_{i,i-1} \Omega_{i-1,j} 
   + \Omega_{ij} \Gamma^\dagger_{jj} + \Omega_{i,j-1} \Gamma^\dagger_{j-1,j} = K_{ij} \, ,
\eez
where $i,j=1,\ldots,n$. This is equivalent to
\bez
  && (\gamma+\gamma^\ast) \, \Omega_{11} = K_{11} \, , \\
  && (\gamma+\gamma^\ast) \, \Omega_{i1} + \Omega_{i-1,1} = K_{i1} \qquad i=2,3,\ldots,n \, , \\ 
  && (\gamma+\gamma^\ast) \, \Omega_{1j} + \Omega_{1,j-1} = K_{1j} \qquad j=2,3,\ldots,n \, , \\ 
  && (\gamma+\gamma^\ast) \, \Omega_{ij} + \Omega_{i-1,j} + \Omega_{i,j-1} = K_{ij} \qquad i,j = 2,3,\ldots,n \, . 
\eez
If $\gamma$ is imaginary, this reduces to the equations in the 
proposition. 
\end{proof}

\begin{example}
	\label{ex:Lyap_n=2_Jordan_imag_gamma}
For the $2 \times 2$ Jordan block with imaginary eigenvalue, the Lyapunov equation imposes the constraints 
$K_{11} = 0$, $K_{21} = K_{12}$, and only determines part of $\Omega$ via
$\Omega_{11} = K_{12}$ and $\Omega_{12} + \Omega_{21} = K_{22}$.
\hfill $\Box$
\end{example}

\section{On the invertibility of the Darboux potential for the 2-soliton solution in Example~\ref{ex:bright-solitons}}
\label{app:2-soliton-solution,inv_of_Omega}
For $n=2$ and with the restriction to\footnote{For complex $\gamma$ the invertibility problem seems to be untractable.} 
\emph{real} $\gamma$, the determinant of the Darboux potential $\Omega$ in Example~\ref{ex:bright-solitons} vanishes if and only if
\bez
 && \gamma_1 \gamma_2 \, \Big( \frac{\boldsymbol{a}_{11} \, \boldsymbol{a}_{22}}{\gamma_1 \gamma_2} - 4 \frac{| \boldsymbol{a}_{12}|^2}{(\gamma_1 + \gamma_2)^2} \Big) \, e^{-2(\varphi_1+\varphi_2)} 
 - \Big( \frac{\boldsymbol{b}_{11} \, \boldsymbol{b}_{22}}{\gamma_1 \gamma_2} - 4 \frac{|\boldsymbol{b}_{12}|^2}{(\gamma_1 + \gamma_2)^2} \Big) \, e^{2(\varphi_1+\varphi_2)} \, \\
 && \hspace{3cm} 
  = 4 \frac{\gamma_1-\gamma_2}{(\gamma_1+\gamma_2)^2} \, \Big( \mathrm{Re}(\boldsymbol{a}_{12}) \, \mathrm{Im}(\boldsymbol{b}_{12}) - \mathrm{Im}(\boldsymbol{a}_{12}) \, \mathrm{Re}(\boldsymbol{b}_{12}) \Big) 
\eez
and 
\bez
    \frac{\boldsymbol{a}_{11} \, \boldsymbol{b}_{22}}{\gamma_2} \, 
    e^{2(\varphi_2 - \varphi_1)} + \frac{\boldsymbol{a}_{22} \, \boldsymbol{b}_{11}}{\gamma_1} \, 
    e^{-2(\varphi_2 - \varphi_1)} 
    = \frac{4}{\gamma_1 + \gamma_2} \Big( \mathrm{Re}(\boldsymbol{a}_{12}) \, \mathrm{Re}(\boldsymbol{b}_{12}) + \mathrm{Im}(\boldsymbol{a}_{12}) \, \mathrm{Im}(\boldsymbol{b}_{12}) \Big) \, ,
\eez	
where $\varphi_i = \frac{1}{2} (\gamma_i^{-1} x + \gamma_i \, t)$. 
The two equations correspond to real and imaginary parts of $\det(\Omega)=0$. We used that $\boldsymbol{a}_{ii}, \boldsymbol{b}_{ii}$, $i=1,2$, are 
real and $\boldsymbol{a}_{21} = \boldsymbol{a}_{12}^\ast$,  $\boldsymbol{b}_{21} = \boldsymbol{b}_{12}^\ast$. 
Both are quadratic equations in $e^{2(\varphi_1+\varphi_2)}$,
respectively $e^{2(\varphi_2-\varphi_1)}$. If the coefficients are such that one of these quadratic equations does \emph{not} have a real and positive solution (which holds in particular if the discriminant is negative), then $\Omega$ in Example~\ref{ex:bright-solitons}
(for $n=2$) is invertible and the corresponding 2-bright-soliton solution is non-singular. In particular, this shows that there is 
a large class of non-singular solutions. We expect that, generically, solutions are also non-singular if an imaginary part of $\gamma$ is switched on.

\section{Off-diagonal components of the Darboux potential in Example~\ref{ex:2vFL_double_root_anticonjugate_pair}}
\label{app:2vFL_double_root_anticonjugate_pair}
\bez
\Omega_{12} &=& -\frac{\imag \, c_{11}c_{12}^\ast \tilde{\gamma}^4 \|A\|^2}{405 \, \alpha_1^4 \alpha_2^4 (A^\dagger \boldsymbol{\alpha} A)^2}(\imag \, \varrho_3 \varrho_4 - \varrho_2  ) \, \big( \gamma \, (\imag \, \varrho_3 \varrho_4 - \varrho_2) - A^\dagger \boldsymbol{\alpha}  A \big) \, t^5 \\
&& + \frac{\imag \, \tilde{\gamma}^2 \|A\|}{108 \, \alpha_1^3 \alpha_2^3 (A^\dagger \boldsymbol{\alpha} A)^2} 
\Big((\imag \, \varrho_3 \varrho_4 -\varrho_2) \, \big(c_{12}^\ast (2c_{11} \|A\| A^\dagger \boldsymbol{\alpha} A \, (\tilde{\gamma} \, x - 3 \gamma) + 3 \gamma  \alpha_1 \alpha_2 \varrho_6) \\
&& + c_{11} (c_{22}^\ast  \tilde{\gamma} \|A\| A^\dagger \boldsymbol{\alpha} A 
- 3 \gamma \alpha_1 \alpha_2 \varrho_5) \big) - 3 c_{12}^\ast \alpha_1 \alpha_2 \varrho_6  A^\dagger \boldsymbol{\alpha} A \Big) \, t^4 \\
&& + \frac{\imag}{27 \, \alpha_1^2 \alpha_2^2 (A^\dagger \boldsymbol{\alpha} A)^2} \Big( c_{11} c_{12}^\ast \tilde{\gamma}^2 \|A\|^2 A^\dagger \boldsymbol{\alpha} A 
(\imag \, \varrho_3 \varrho_4 - \varrho_2) \, x^2 \\
&& + \tilde{\gamma} \|A\| A^\dagger \boldsymbol{\alpha} A \, \big(   c_{11} c_{22}^\ast \tilde{\gamma} \|A\| \, (\imag \, \varrho_3 \varrho_4 - \varrho_2) - 6   c_{12}^\ast \alpha_1 \alpha_2 \varrho_6 \big) \, x \\
&& - \tilde{\gamma}^2 \|A\|^2 (\imag \, \varrho_3 \varrho_4 - \varrho_2) \, \big( c_{12}^\ast \gamma  \varrho_7 - c_{11} (c_{32}^\ast A^\dagger \boldsymbol{\alpha} A - \gamma \, \varrho_8) \big) +  c_{12}^\ast \|A\| A^\dagger \boldsymbol{\alpha}  A \, (\tilde{\gamma}^2 \varrho_7 \|A\| \\
&& +18 \gamma \alpha_1 \alpha_2 \varrho_6) - 3 \alpha_1 \alpha_2 \varrho_6 (c_{22}^\ast \tilde{\gamma} \|A\| A^\dagger \boldsymbol{\alpha}  A -3 \gamma \alpha_1 \alpha_2 \varrho_5) \Big) \, t^3 \\
&& - \frac{\imag \, \|A\|}{6 \alpha_1 \alpha_2 (A^\dagger \boldsymbol{\alpha} A)^2} 
\Big(3 c_{12}^\ast \alpha_1 \alpha_2 \varrho_6 A^\dagger \boldsymbol{\alpha} A \, x^2 - A^\dagger \boldsymbol{\alpha}  A \, (2 c_{12}^\ast \tilde{\gamma} \varrho_7 \|A\| - 3 c_{22}^\ast  \alpha_1 \alpha_2 \varrho_6 ) \, x \\
&& + \varrho_7 \|A\| A^\dagger \boldsymbol{\alpha} A \, (6c_{12}^\ast \gamma - c_{22}^\ast \tilde{\gamma}) + 3 \alpha_1 \alpha_2 \big(c_{32}^\ast \varrho_6 A^\dagger \boldsymbol{\alpha}  A + \gamma (\varrho_5 \varrho_7 - \varrho_6 \varrho_8) \big) \Big) \, t^2 \\	
&& + \frac{\imag \, \varrho_7 \|A\|^2}{(A^\dagger \boldsymbol{\alpha} A)^2} \big( A^\dagger \boldsymbol{\alpha} A \, (c_{12}^\ast x^2 + c_{22}^\ast x + c_{32}^\ast) - \gamma \varrho_8 \big) \, t \\	
&& + \frac{c_{11} c_{12}^\ast}{10 \gamma^2(\alpha_2-\alpha_1)^2 |A_1|^2 |A_2|^2 A^\dagger \boldsymbol{\alpha} A} 
\Big( (\alpha_2-\alpha_1)^2 |A_1|^2 |A_2|^2 (2 \gamma r-1) \, (\varrho_3 \varrho_4 \|A\|^2 + A^\dagger \boldsymbol{\alpha} A) \\
&& + 2 \imag \, \gamma \alpha_1 \alpha_2 \varrho_3^2 \varrho_4^2\|A\|^6 \Big) \, x^5  \\
&& +\frac{1}{8\gamma^2(\alpha_2-\alpha_1)^2|A_1|^2|A_2|^2A^\dagger \boldsymbol{\alpha} A} \Big( (\alpha_2-\alpha_1)^2 |A_1|^2 |A_2|^2 \big(2c_{11} c_{12}^\ast (|A|^2 \varrho_3 (2\gamma \, (r \varrho_1 + \varrho_4) - \varrho_1) \\
&& + 2 \gamma A^\dagger \boldsymbol{\alpha} A) + (2\gamma r-1) \, (c_{11} c_{22}^\ast + c_{12}^\ast c_{21}) \, (\varrho_3 \varrho_4 \|A\|^2 + A^\dagger \boldsymbol{\alpha} A) \big) \\
&& + 2 \imag \, \gamma \alpha_1 \alpha_2 \varrho_3^2 \varrho_4^2 \|A\|^6  (c_{11} c_{22}^\ast + c_{12}^\ast c_{21}) \Big) \, x^4 \\
&& + \frac{1}{6 \gamma^2 (\alpha_2-\alpha_1)^2 |A_1|^2 |A_2|^2 (A^\dagger \boldsymbol{\alpha}  A)^2} \Big( \\ && (\alpha_2-\alpha_1)^2|A_1|^2 |A_2|^2  \big(c_{12}^\ast \varrho_{10} + c_{22}^\ast ((A^\dagger \boldsymbol{\alpha} A)^2 (c_{21} (2\gamma r-1) + 2 c_{11} \gamma ) \\
&& + \varrho_3 \|A\|^2 A^\dagger \boldsymbol{\alpha} A \, (2 \gamma r-1) \, (c_{11} \varrho_1 + c_{21} \varrho_4)  + 4 c_{11} \gamma \varrho_3 \varrho_4 \|A\|^2 A^\dagger \boldsymbol{\alpha} A - 2 \imag \,  c_{11} \gamma^2 \|A\|^2 (\imag \, \varrho_3 \varrho_4 - \varrho_2)^2) \\
&&+ c_{11} c_{32}^\ast A^\dagger \boldsymbol{\alpha} A \, (2\gamma r-1) (\varrho_3 \varrho_4 \|A\|^2 + A^\dagger \boldsymbol{\alpha}  A) \big) \\
&& + 2 \gamma \varrho_3^2 \|A\|^4 A^\dagger \boldsymbol{\alpha}  A \, \big( \imag \, \alpha_1 \alpha_2 \varrho_9 \|A\|^2 - c_{11} \varrho_4 A^\dagger \boldsymbol{\alpha}  A \, (2c_{12}^\ast \varrho_1 + c_{22}^\ast \varrho_4) \big) \Big) \, x^3 \\
&& - \frac{1}{4\gamma^2(\alpha_2-\alpha_1)^2 |A_1|^2 |A_2|^2 (A^\dagger \boldsymbol{\alpha}  A)^2} \Big((\alpha_2-\alpha_1)^2 |A_1|^2 |A_2|^2\big(4\imag \,  c_{11}\gamma^2H_2\|A\|^2(\imag \, \varrho_3\varrho_4-\varrho_2) \\
&&-D_2\|A\|^2 A^\dagger \boldsymbol{\alpha} A \, (2\gamma\,(2c_{11}+c_{21}r)-c_{21})-2\gamma\varrho_3 \|A\|^2 A^\dagger \boldsymbol{\alpha} A \, (c_{21}+c_{31}r)(2c_{12}^\ast \varrho_1 + c_{22}^\ast \varrho_4) \\
&&+c_{31}\varrho_3 \|A\|^2 A^\dagger \boldsymbol{\alpha} A \, (2c_{12}^\ast \varrho_1+c_{22}^\ast\varrho_4)\big)-\varrho^{(1)}+4c_{11}\gamma \varrho_3\varrho_4D_2\|A\|^4(A^\dagger \boldsymbol{\alpha} A)^2 \Big) \, x^2 \\
&& + \frac{1}{2\gamma^2 (\alpha_2-\alpha_1)^2 |A_1|^2 |A_2|^2 (A^\dagger \boldsymbol{\alpha} A)^2} \Big( (\alpha_2-\alpha_1)^2|A_1|^2|A_2|^2\big(c_{31}c_{32}^\ast(A^\dagger \boldsymbol{\alpha} A)^2 (2\gamma r-1) \\
&&-\|A\|^2 (2\gamma^2 H_2(\imag \,  c_{21}(\imag \, \varrho_3 \varrho_4-\varrho_2)	-2c_{11}(\imag-\varrho_1\varrho_3))-2\gamma D_2 A^\dagger \boldsymbol{\alpha} A \, (c_{21}+c_{31}r)+c_{31}D_2 A^\dagger \boldsymbol{\alpha} A)\big) \\
&&+2\gamma D_2\|A\|^4A^\dagger \boldsymbol{\alpha} A \, \big(\imag \, \alpha_1 \alpha_2 D_1 \|A\|^2 + \varrho_3 A^\dagger \boldsymbol{\alpha} A \, (2c_{11}\varrho_1-c_{21} \varrho_4)\big)\Big) \, x + \tilde{c}_{12} \, , \\
\Omega_{21} &=& \frac{\imag \, c_{11}^\ast c_{12} \tilde{\gamma}^{\ast 4} \|A\|^2}{405 \, \alpha_1^4 \alpha_2^4 (A^\dagger \boldsymbol{\alpha} A)^2}(\imag \, \varrho_3^\ast \varrho_4^\ast + \varrho_2^\ast) \, \big( \gamma^\ast (\imag \, \varrho_3^\ast \varrho_4^\ast + \varrho_2^\ast) + A^\dagger \boldsymbol{\alpha} A \big) \, t^5 \\
&& + \frac{\imag \, \tilde{\gamma}^{\ast 2} \|A\|}{108 \, \alpha_1^3 \alpha_2^3 (A^\dagger \boldsymbol{\alpha} A)^2} \Big( (\imag \, \varrho_3^\ast \varrho_4^\ast + \varrho_2^\ast) \, \big(c_{11}^\ast (2c_{12} \|A\| \, A^\dagger \boldsymbol{\alpha} A \, (\tilde{\gamma}^\ast x + 3\gamma^\ast)-3 \gamma^\ast \alpha_1 \alpha_2 \varrho_5^\ast) \\
&& +c_{12} (c_{21}^\ast \tilde{\gamma}^\ast \|A\| \, A^\dagger \boldsymbol{\alpha} A 
+ 3 \gamma^\ast \alpha_1 \alpha_2 \varrho_6^\ast) \big) - 3c_{11}^\ast \alpha_1 \alpha_2 \varrho_5^\ast A^\dagger \boldsymbol{\alpha} A \Big) \, t^4 \\
&& + \frac{\imag}{27 \alpha_1^2 \alpha_2^2 (A^\dagger \boldsymbol{\alpha} A)^2}
\Big( c_{11}^\ast c_{12} \tilde{\gamma}^{\ast 2} \|A\|^2 A^\dagger \boldsymbol{\alpha} A (\imag \, \varrho_3^\ast \varrho_4^\ast + \varrho_2^\ast) \, x^2 \\
&& + \tilde{\gamma}^\ast \|A\| \, A^\dagger \boldsymbol{\alpha} A \, \big(   c_{12} c_{21}^\ast \tilde{\gamma}^\ast \|A\| \, (\imag \, \varrho_3^\ast \varrho_4^\ast + \varrho_2^\ast)
- 6 c_{11}^\ast \alpha_1 \alpha_2 \varrho_5^\ast \big) \, x \\
&& - \tilde{\gamma}^{\ast 2} \|A\|^2 (\imag \, \varrho_3^\ast \varrho_4^\ast + \varrho_2^\ast) \, \big(c_{11}^\ast \gamma^\ast  \varrho_8^\ast - c_{12} \, (c_{31}^\ast A^\dagger \boldsymbol{\alpha} A - \gamma^\ast \varrho_7^\ast) \big)
- c_{11}^\ast \|A\| \, A^\dagger \boldsymbol{\alpha} A \, (\tilde{\gamma}^{\ast 2} \varrho_8^\ast \|A\| \\
&& + 18 \gamma^\ast \alpha_1 \alpha_2 \varrho_5^\ast) -3 \alpha_1 \alpha_2 \varrho_5^\ast \, (c_{21}^\ast \tilde{\gamma}^\ast \|A\| \, A^\dagger \boldsymbol{\alpha} A + 3 \gamma^\ast \alpha_1 \alpha_2 \varrho_6^\ast) \Big) \, t^3 \\	
&& -\frac{\imag \, \|A\|}{6 \alpha_1 \alpha_2 (A^\dagger \boldsymbol{\alpha} A)^2}  \Big(3 c_{11}^\ast \alpha_1 \alpha_2 \varrho_5^\ast A^\dagger \boldsymbol{\alpha} A \, x^2 + A^\dagger \boldsymbol{\alpha} A \, (2c_{11}^\ast \tilde{\gamma}^\ast \varrho_8^\ast \|A\| + 3c_{21}^\ast \alpha_1 \alpha_2 \varrho_5^\ast ) \, x \\
&& + \varrho_8^\ast \|A\| \, A^\dagger \boldsymbol{\alpha} A \, (6c_{11}^\ast \gamma^\ast + c_{21}^\ast \tilde{\gamma}^\ast) + 3 \alpha_1 \alpha_2 \, \big( c_{31}^\ast \varrho_5^\ast A^\dagger \boldsymbol{\alpha} A - \gamma^\ast (\varrho_5^\ast \varrho_7^\ast -\varrho_6^\ast \varrho_8^\ast) \big) \Big) \, t^2 \\	
&& - \frac{\imag \, \varrho_8^\ast \|A\|^2}{(A^\dagger \boldsymbol{\alpha} A)^2}\big( A^\dagger \boldsymbol{\alpha} A \, (c_{11}^\ast x^2 + c_{21}^\ast x + c_{31}^\ast)-\gamma^\ast \varrho_7^\ast \big) \, t \\
&& + \frac{c_{11}^\ast  c_{12}}{10 \gamma^{\ast 2} (\alpha_2-\alpha_1)^2 |A_1|^2 |A_2|^2 A^\dagger \boldsymbol{\alpha} A} \Big( (\alpha_2-\alpha_1)^2 |A_1|^2 |A_2|^2 (2\gamma^\ast r^\ast-1) \, (\varrho_3^\ast \varrho_4^\ast \|A\|^2 + A^\dagger \boldsymbol{\alpha}  A) \\
&& - 2 \imag \, \gamma^\ast \alpha_1 \alpha_2 \varrho_3^{\ast 2} \varrho_4^{\ast 2} \|A\|^6 \Big) \, x^5 \\
&& - \frac{1}{8 \gamma^{\ast 2} (\alpha_2-\alpha_1)^2 |A_1|^2 |A_2|^2 A^\dagger \boldsymbol{\alpha} A} \Big( (\alpha_2-\alpha_1)^2 |A_1|^2 |A_2|^2 \big(2c_{11}^\ast c_{12} (\|A\|^2 \varrho_3^\ast \, (2\gamma^\ast (r^\ast \varrho_1^\ast + \varrho_4^\ast) - \varrho_1^\ast) \\
&& +2\gamma^\ast A^\dagger \boldsymbol{\alpha} A) - (2\gamma^\ast r^\ast-1) \, (\varrho_3^\ast \varrho_4^\ast \|A\|^2 + A^\dagger \boldsymbol{\alpha} A) \, (c_{11}^\ast c_{22} + c_{12} c_{21}^\ast) \big) \\ 
&& + 2\imag \, \gamma^\ast \alpha_1 \alpha_2 \varrho_3^{\ast 2} \varrho_4^{\ast 2} \|A\|^6  (c_{11}^\ast c_{22} + c_{12} c_{21}^\ast) \Big) \, x^4 \\
&& + \frac{1}{6\gamma^{\ast 2}(\alpha_2-\alpha_1)^2 |A_1|^2 |A_2|^2 (A^\dagger \boldsymbol{\alpha} A)^2} \Big( \\
&& (\alpha_2-\alpha_1)^2 |A_1|^2 |A_2|^2 \, \big( c_{12} \,( \varrho_{10}^\ast -3 c_{21}^\ast A^\dagger \boldsymbol{\alpha} A \, (\varrho_3^\ast \|A\|^2  (2\gamma^\ast (r^\ast \varrho_1^\ast + \varrho_4^\ast) - \varrho_1^\ast) + 2\gamma^\ast A^\dagger \boldsymbol{\alpha} A)) \\
&& + c_{22} ((A^\dagger \boldsymbol{\alpha} A)^2  (c_{21}^\ast (2\gamma^\ast r^\ast-1) -4 c_{11}^\ast \gamma^\ast) + 2\imag \,  c_{11}^\ast \gamma^{\ast 2} \|A\|^2 (\imag \, \varrho_3^\ast \varrho_4^\ast + \varrho_2^\ast)^2 \\
&& - \varrho_3^\ast \|A\|^2 A^\dagger \boldsymbol{\alpha} A \, (2\gamma^\ast r^\ast-1)  (2c_{11}^\ast \varrho_1^\ast - c_{21}^\ast  \varrho_4^\ast) - 2\gamma^\ast c_{11}^\ast  \varrho_3^\ast \varrho_4^\ast \|A\|^2 A^\dagger \boldsymbol{\alpha} A) \\
&& + c_{11}^\ast  c_{32} A^\dagger \boldsymbol{\alpha} A \, (2\gamma^\ast r^\ast-1) ( \varrho_3^\ast \varrho_4^\ast \|A\|^2 + A^\dagger \boldsymbol{\alpha} A) \big) \\	
&& -2 \gamma^\ast \varrho_3^{\ast 2} \|A\|^4 A^\dagger \boldsymbol{\alpha}  A \, \big( \imag \, \alpha_1 \alpha_2 \varrho_9^\ast \|A\|^2 + c_{11}^\ast \varrho_4^\ast A^\dagger \boldsymbol{\alpha} A \, (2c_{12} \varrho_1^\ast + c_{22} \varrho_4^\ast) \big) \Big) \, x^3 \\ 
&& + \frac{1}{4\gamma^{\ast 2}(\alpha_2-\alpha_1)^2 |A_1|^2 |A_2|^2 (A^\dagger \boldsymbol{\alpha} A)^2} \Big((\alpha_2-\alpha_1)^2 |A_1|^2 |A_2|^2 \big(4\imag \, c_{12}\gamma^{\ast 2} H_1^\ast \|A\|^2 (\imag \, \varrho_3^\ast \varrho_4^\ast + \varrho_2^\ast ) \\
&& -D_1^\ast \|A\|^2 A^\dagger \boldsymbol{\alpha} A \, (2\gamma^\ast\,(2c_{12}-c_{22}r^\ast)+c_{22})+2\gamma^\ast \varrho_3^\ast \|A\|^2 A^\dagger \boldsymbol{\alpha} A \, (c_{22}-c_{32}r^\ast)(2c_{11}^\ast \varrho_1^\ast-c_{21}^\ast\varrho_4^\ast) \\
&&+c_{32}\varrho_3^\ast \|A\|^2  A^\dagger \boldsymbol{\alpha} A \, (2c_{11}^\ast \varrho_1^\ast -c_{21}^\ast \varrho_4^\ast) \big) - \varrho^{(2)^\ast} + 4 c_{12}\gamma^\ast \varrho_3^\ast \varrho_4^\ast D_1^\ast \|A\|^4 (A^\dagger \boldsymbol{\alpha} A)^2 \Big) \, x^2 \\
&& + \frac{1}{2\gamma^{\ast 2} (\alpha_2-\alpha_1)^2 |A_1|^2 |A_2|^2 (A^\dagger \boldsymbol{\alpha} A)^2} \Big( (\alpha_2-\alpha_1)^2|A_1|^2|A_2|^2\big(c_{31}^\ast c_{32}(A^\dagger \boldsymbol{\alpha} A)^2 (2\gamma^\ast r^\ast-1) \\
&& + \|A\|^2 (2\gamma^{\ast 2} H_1^\ast(\imag \, c_{22}(\imag \, \varrho_3^\ast \varrho_4^\ast+\varrho_2^\ast)-2c_{12}(\imag+\varrho_1^\ast\varrho_3^\ast))-2\gamma ^\ast D_1^\ast A^\dagger \boldsymbol{\alpha} A \, (c_{22}-c_{32}r^\ast)-c_{32} D_1^\ast A^\dagger \boldsymbol{\alpha} A) \big) \\
&&-2\gamma^\ast  D_1^\ast \|A\|^4 A^\dagger \boldsymbol{\alpha} A \, \big( \imag \, \alpha_1 \alpha_2 D_2^\ast \|A\|^2 -\varrho_3^\ast A^\dagger \boldsymbol{\alpha} A \, (2c_{12} \varrho_1^\ast + c_{22} \varrho_4^\ast) \big) \Big) \, x + \tilde{c}_{21} \, , 
\eez
with further constants
\bez
\varrho_5 &=& \frac{\|A\|}{3 \alpha_1 \alpha_2} \Big( \big(\imag \, \varrho_3 \varrho_4 - \varrho_2 \big) \big(2c_{12}^\ast (\tilde{\gamma} \,x -3 \gamma) + c_{22}^\ast \tilde{\gamma} \big) + 2c_{12}^\ast \tilde{\gamma} \, (\imag \, \varrho_1 \varrho_3 + 1) \Big) \, , \\
\varrho_6 &=& -\frac{\|A\|}{3 \alpha_1 \alpha_2} \Big(2c_{11} \tilde{\gamma} \, \big( ( \imag \, \varrho_3 \varrho_4-\varrho_2) \,x -\imag \, \varrho_1 \varrho_3 -1 \big) + (6 c_{11} \gamma + c_{21} \tilde{\gamma}) \, \big(\imag \, \varrho_3 \varrho_4 - \varrho_2 \big) \Big) \, , \\
\varrho_7 &=& c_{11} \big(\imag \, \varrho_3 \varrho_4 - \varrho_2 \big) \, x^2 + \big( c_{21} (\imag \, \varrho_3 \varrho_4 - \varrho_2) - 2c_{11} (\imag \, \varrho_1 \varrho_3 + 1) \big) \, x + \imag \, D_1  -c_{31} \varrho_2 - c_{21} \, , \\
\varrho_8 &=& c_{12}^\ast (\imag \, \varrho_3 \varrho_4 - \varrho_2) \, x^2 + \big( c_{22}^\ast (\imag \, \varrho_3 \varrho_4 - \varrho_2) + 2 c_{12}^\ast (\imag \, \varrho_1 \varrho_3+1) \big) \, x + \imag \, D_2  -c_{32}^\ast \varrho_2 + c_{22}^\ast \, , \\
\varrho_9 &=& \varrho_1 \varrho_4 (c_{12}^\ast  c_{21} - c_{11} c_{22}^\ast) + \varrho_4^2 (c_{11} c_{32}^\ast + c_{12}^\ast c_{31} + c_{21} c_{22}^\ast) - 4c_{11} c_{12}^\ast (\varrho_1^2-\varrho_4) \, , \\
\varrho_{10} &=& (A^\dagger \boldsymbol{\alpha} A)^2 (c_{31} (2\gamma r-1) + 2 \gamma  c_{21}) + \varrho_3 \|A\|^2 A^\dagger \boldsymbol{\alpha} A \, ((2\gamma r-1) \, (2c_{11} + 2c_{21} \varrho_1 + c_{31} \varrho_4) \\ 
&& + 2 \gamma \, (4c_{11} \varrho_1 + c_{21}\varrho_4)) - 4 \imag \,  c_{11} \gamma^2 \|A\|^2 (\imag \, \varrho_3 \varrho_4 - \varrho_2) \, (1+\imag \, \varrho_1 \varrho_3) \, , \\
\varrho^{(j)} &=& -2\gamma^2 \, (\alpha_2-\alpha_1)^2 |A_1|^2 |A_2|^2 \|A\|^2  \big(\imag \, \varrho_3^2 (2 c_{11} \varrho_1-c_{21} \varrho_4) \, (2 c_{12}^\ast \varrho_1 + c_{22}^\ast \varrho_4) \\
&& + 2 c_{22}^\ast \varrho_3 (c_{21} \varrho_2 \varrho_4 - c_{11} (\varrho_1 \varrho_2 - \varrho_4)) + 2 c_{12}^\ast \varrho_3 (4 c_{11} \varrho_1 + c_{21}(\varrho_1 \varrho_2-\varrho_4)) \\
&&-\imag \, (2c_{11} + c_{21} \varrho_2) \, (2 c_{12}^\ast - c_{22}^\ast \varrho_2) \big) + 2 \gamma \varrho_3^2 \|A\|^4 (A^\dagger \boldsymbol{\alpha} A)^2 (2 c_{11} \varrho_1 - c_{21} \varrho_4) \, (2 c_{12}^\ast \varrho_1 + c_{22}^\ast \varrho_4) \\
&& -(-1)^j \big( (\alpha_2-\alpha_1)^2 |A_1|^2 |A_2|^2 (A^\dagger \boldsymbol{\alpha} A)^2 (2 \gamma r-1) (c_{21} c_{32}^\ast + c_{22}^\ast c_{31}) \\
&& + 2 \imag \, \gamma \alpha_1 \alpha_2 \varrho_3 \|A\|^6  A^\dagger \boldsymbol{\alpha} A \, (D_1 (2c_{12}^\ast \varrho_1 + c_{22}^\ast \varrho_4) - D_2 (2 c_{11} \varrho_1 - c_{21} \varrho_4)) \big) \, ,
\eez
with complex constants $\tilde{c}_{12}$ and $\tilde{c}_{21}$. We were unable to cast these expressions into a more compact form. Essentially, we included this in order to demonstrate that, in cases where the spectrum condition does not hold, already for $n=2$ surprisingly complicated solutions emerge.

\section{Chirping travelling waves}
\label{app:chirped}
The following is based on and motivated by \cite{TZLBMYAB22}. Inserting the ansatz
\bez
    u = \left( \begin{array}{c} A_1 \, e^{\imag \, (\varphi_1 -k_1 t)} \\ A_2 \, e^{\imag \, (\varphi_2 -k_2 t)}
    	\end{array} \right) 
\eez
where $A_1,A_2,\varphi_1,\varphi_2$ are real functions of $\zeta = x-v t$ and $v,k_1,k_2$ are real constants, 
in the two-component vector FL equation with $\boldsymbol{\sigma}= \mathrm{diag}(1,\sigma)$, $\sigma \in \{\pm1\}$, and decomposing the result into real and imaginary parts, we obtain
\bez
  && v \, (A_1 \varphi_1'' + 2 A_1' \, \varphi_1') + k_1 A_1' + 2 A_1^2 A_1'
   + \sigma \, (A_2^2 A_1' + A_1 A_2 A_2') = 0 \, , \\
  && v \, (A_2 \varphi_2'' + 2 A_2' \, \varphi_2') + k_2 A_2' + 2 \sigma A_2^2 A_2' + A_1^2 A_2' + A_1 A_2 A_1' = 0 \, , \\
  && v (A_1'' - A_1 (\varphi_1')^2 ) - k_1 A_1 \varphi_1' 
  - A_1 (2 A_1^2 + \sigma A_2^2) \, \varphi_1' - \sigma A_1 A_2^2 \varphi_2'
  + A_1 = 0 \, , \\
  && v (A_2'' - A_2 (\varphi_2')^2 ) - k_2 A_2 \varphi_2' 
  - A_2 (A_1^2 + 2 \sigma A_2^2) \, \varphi_2' - A_1^2 A_2 \varphi_1'
  + A_2 = 0 \, ,
\eez
where a prime denotes a derivative with respect to $\zeta$. 
The first two equations integrate to
\bez
   \varphi_1' = \frac{1}{2v} \big( c_1 A_1^{-2} - k_1 - A_1^2 - \sigma A_2^2 \big)  \, , \qquad
   \varphi_2' = \frac{1}{2v} \big(c_2 A_2^{-2} - k_2 - A_1^2 - \sigma A_2^2 \big) \, ,
\eez
with real constants $c_1,c_2$. $\varphi_1$ and $\varphi_2$ are now 
determined by a further integration. If the amplitudes are not constant, this leads to nonlinear functions $\varphi_i$, which means that the frequencies of the carrier waves are not constant, so that the waves ``chirp". 
Inserting the last expressions in the remaining two equations, 
turns them into nonlinear coupled second order equations for $A_1$ and $A_2$,
\bez
 && 4 v^2 A_1'' + A_1 \Big( k_1^2 + 4 v - 2 (c_1+ \sigma c_2) + 4 k_1 A_1^2
  + 3 (A_1^4 + A_2^4) + 2 \sigma (k_1+k_2+3 A_1^2) A_2^2 \Big) 
  = c_1^2 A_1^{-3} \, , \\ 
 &&  4 v^2 A_2'' + A_2 \Big( k_2^2 + 4 v - 2 (c_1+ \sigma c_2) + 4 \sigma k_2 A_2^2  + 3 (A_1^4 + A_2^4) + 2 A_1^2 (k_1+k_2+3 \sigma A_2^2) \Big) 
 = c_2^2 A_2^{-3} \, .
\eez
In order to find at least special solutions of this highly nonlinear coupled system, we consider the reduction
\bez
        A_2 = c \, A_1 \, ,
\eez
with a constant $c$. This works if 
\bez
       k_2 = k_1 \, , \qquad c_2^2 = c_1^2 c^4 \, , 
\eez
and reduces the 
above system to a single nonlinear ODE for $A_1$, which in terms of $B = A_1^2$ can be integrated to
\be
     v^2 (B')^2  + c_1^2 - 2 \tilde{c} B + z_1 B^2 + \frac{1}{2} z_2 B^3 + \frac{1}{3} z_3 B^4 = 0 \, ,
     \label{chirp_B_eq}
\ee
where $\tilde{c}$ is an integration constant and we introduced
\bez
  z_1 = k_1^2+4v-2(c_1+\sigma c_2) \, ,  \qquad
  z_2 = 4 (1+ \sigma c^2) \, k_1 \, , \qquad
  z_3 = 3 (1+c^4+2\sigma c^2) \,  .
\eez
Solving (\ref{chirp_B_eq}) for $B'$, the resulting autonomous first order ODE can be integrated via separation of variables, 
which yields its general solution in implicit form.

\begin{example}
In order to get an explicit example, we seek a localized solution of the form $A_1^2 = B=a+b\,\mathrm{sech}(\mu\, \zeta )$, with real constants $a,b,\mu$.
Substituting this ansatz into  (\ref{chirp_B_eq}), 
using $ (B')^2 = \mu^2 (B-a)^2 \big( 1 - b^{-2} (B-a)^2 \big)$, and equating all coefficients of the resulting polynomial in $B$ separately to zero,  we obtain
\bez
     a = -\frac{3z_2}{8z_3} \, , \quad
     b^2 = \frac{3\big( 9z_2^2-32z_1z_3 \big)}{32z_3^2} \, , \quad
     \mu^2 = \frac{9z_2^2-32z_1z_3}{32v^2z_3} \, , \quad
    z_1 = \left(\frac{8 c_1 z_3}{3 z_2} \right)^2 + \frac{15 z_2^2}{64 z_3} \, ,
\eez
where the last equation restricts the parameters in (\ref{chirp_B_eq}) and can be regarded as determining the velocity $v$. 
\hfill $\Box$
\end{example}

In \cite{TZLBMYAB22}, the authors imposed the restriction $c_1=c_2=0$. They 
actually considered a generalization of the vector FL equation.
\end{appendix}

\end{document}